\documentclass[conference,letterpaper]{IEEEtran}

\usepackage{multirow}
\usepackage[left=0.71in,top=0.94in,right=0.71in,bottom=1.18in]{geometry}
\usepackage[english]{babel}
\usepackage[utf8]{inputenc}
\usepackage{algorithm}
\usepackage[noend]{algpseudocode}

\usepackage{stfloats}

\usepackage[noadjust]{cite}

\usepackage{array}
\usepackage{geometry}                
\usepackage{subcaption} 
\usepackage{graphicx}
\usepackage{amssymb}
\usepackage{amsthm}
\usepackage{color}
\usepackage{blindtext}
\usepackage{epstopdf}
\usepackage{amsmath}
\usepackage{filecontents,lipsum}
\usepackage[noadjust]{cite}
\usepackage{pseudocode}
\newcommand{\path}{\mathrm{path}}

\newcommand{\pgen}{\mathrm{P}_{\mathrm{gen}}}
\newcommand{\dist}{\mathrm{dist}}

\newcommand{\G}{\mathcal{G}}
\newcommand{\ths}{\mathrm{th}}
\newcommand{\Z}{\tilde{Z}}

\newcommand{\phs}{\mathrm{ph}}

\newcommand{\pchose}{\mathrm{P}_{\mathrm{choose}}}

\newcommand{\p}{\mathcal{P}}

\newcommand{\I}{\mathbb{I}}
\newcommand{\curr}{\mathrm{curr}}

\newcommand{\mincut}{\mathrm{mincut}}
\newcommand{\phys}{\mathrm{phys}}
\newcommand{\neigh}{\mathrm{Neigh}}
\newcommand{\E}{\mathcal{E}}

\newtheorem{definition}{Definition}
\newtheorem{theorem}{Theorem}
\newtheorem{lemma}{Lemma}
\newtheorem{corollary}{Corollary}
\DeclareMathOperator*{\argmin}{arg\,min}
\newcommand{\diam}{\mathrm{diam}}

\title{Distributed Routing in a Quantum Internet}

\author{\IEEEauthorblockN{Kaushik Chakraborty\IEEEauthorrefmark{1},
Filip Rozpedek\IEEEauthorrefmark{2}, Axel Dahlberg\IEEEauthorrefmark{3} and
Stephanie Wehner\IEEEauthorrefmark{4}}
\IEEEauthorblockA{QuTech, 
Delft University of Technology,\\
Lorentzweg 1, 2628 CJ Delft, The Netherlands.\\
Email: \IEEEauthorrefmark{1}K.Chakraborty@tudelft.nl,
\IEEEauthorrefmark{2}F.D.Rozpedek@tudelft.nl,
\IEEEauthorrefmark{3}E.A.Dahlberg@tudelft.nl,
\IEEEauthorrefmark{4}S.D.C.Wehner@tudelft.nl}}

\begin{document}

\maketitle
\begin{abstract}

We develop new routing algorithms for a quantum network with noisy quantum devices such that each can store a small number of qubits.
We thereby consider two models for the operation of such a network. The first is a \emph{continuous model}, in which entanglement between a subset of the nodes is produced continuously in the background. This can in principle allows the rapid creation of entanglement between more distant nodes using the already pre-generated entanglement pairs in the network.
The second is an \emph{on-demand} model, where entanglement production does not commence before a request is made. 
Our objective is to find protocols, that minimise the latency of the network to serve a request to create entanglement between two distant nodes in the network. 
We propose three routing algorithms and analytically show that as expected when there is only a single request in the network, then employing them on the continuous model yields a lower latency than on the on-demand one. We study the performance of the routing algorithms in a ring, grid, and recursively generated network topologies. We also give an analytical upper bound on the number of entanglement swap operations the nodes need to perform for routing entangled links between a source and a destination yielding a lower bound on the end to end fidelity of the shared entangled state. We proceed to study the case of multiple concurrent requests and show that in some of the scenarios the on-demand model can outperform the continuous one.
Using numerical simulations on ring and grid networks we also study the behaviour of the latency of all the routing algorithms. We observe that the proposed routing algorithms behave far better than the existing classical greedy routing algorithm. The simulations also help to understand the advantages and disadvantages of different types of continuous models for different types of demands.

\end{abstract}

\section{Introduction}

The goal of a quantum internet \cite{Van14, LSWK04, Kim08, Cast18} is to enable the transmission of quantum bits (qubits) between distant quantum devices to achieve the tasks that are impossible using classical communication. For example, with such a network we can implement cryptographic protocols like long-distance quantum key distribution (QKD) \cite{bb14,E91}, which enables secure communication. Apart from QKD, many other applications in the domain of distributed computing and multi-party cryptography \cite{BC16} have already been identified at different stages of quantum network development~\cite{WEH18}.

Like the classical internet, a quantum internet consists of the network components like physical communication links, and eventually routers  \cite{LSWK04,MLMN09,SRAS07,SDS09}. However, due to fundamental differences between classical and quantum bits, these components in a quantum network behave rather differently than their classical counterparts. For example, qubits cannot be copied, which rules out retransmission as a means to overcome qubit losses \cite{NC02}. To nevertheless send qubits reliably, a standard method is to first produce quantum entanglement between a qubit held by the sender and a qubit held by the receiver. Once this entanglement has been produced, the qubit can then be sent using quantum teleportation~\cite{NC02,teleport93}. This requires, in addition, the transmission of two classical bits per qubit from the sender to the receiver. Importantly, teleportation consumes the entanglement, meaning that it has to be re-established before the next qubit can be sent. When it comes to routing qubits in a network, one hence needs to consider routing entanglement \cite{Cal17,GI18,Van14,MSLM13,PLCL13}.

An important tool for establishing entanglement over long distances is the notion of entanglement swapping. If two nodes $A$ and $B$ are both connected to an intermediary node $r$, but not directly connected themselves by a physical quantum communication channel such as fiber, then $A$ and $B$ can nevertheless create entanglement between themselves with the help of $r$. First, $A$ and $B$ each individually create entanglement with $r$. This requires one qubit of quantum storage at $A$ and $B$ to hold their end of the entanglement, and two qubits of quantum storage at $r$. Node $r$ then performs an \emph{entanglement swap}~\cite{teleport93,ZZHE93,GWZ08}, destroying its own entanglement with $A$ and $B$, but instead creating entanglement between $A$ and $B$. This process can be understood as node $r$ teleporting its qubit entangled with $A$ onto node $B$ using the entanglement that it shares with $B$. In turn, using this process iteratively, node $r$ can with the assistance of $A$ and $B$, also establish entanglement with nodes that are far away in the physical communication network. Any node $r$ capable of storing qubits can thus simultaneously be entangled with as many nodes in the network as it can store qubits in its quantum memory. Such a node may function as an entanglement router by taking decisions for which of its neighbours it should perform an entanglement swap operation for sharing an entangled link between a source $s$ and a destination $e$ (see~\cite{SMIKW16} for a longer introduction).

In the domain of quantum information we use the term quantum state to represent the state of a multi-qubit quantum system. A pure $n$-qubit quantum state $|\psi_n\rangle$ can be mathematically described as a unit vector in a Hilbert space $\mathcal{H}$ of dimension $2^n$. The entangled target state $|\psi^+\rangle = \frac{1}{\sqrt{2}}(|00\rangle_{AB} + |11\rangle_{AB})$
is a pure quantum state of two qubits $A$ and $B$. An $n$-qubit mixed state $\rho_n$ on $\mathcal{H}$ is a Hermitian operator with unit
trace. $\rho_n$ is called as density matrix. A mixed state is a generalisation of a pure state that can model noisy quantum states, and the density matrix
representation $\rho$ of a pure state $|\psi\rangle$ is $|\psi\rangle\langle\psi|$, where $\langle\psi|$
is the transpose of the complex conjugate of $|\psi\rangle$. In this
paper, we use the quantity, known as fidelity to measure the closeness
between two quantum states. The fidelity between a target pure state $|\psi\rangle$ and a mixed state $\rho$ is defined as $F(\rho, |\psi\rangle\langle\psi|) = \langle\psi|\rho |\psi\rangle$. The mixed state $\rho$ has a unit trace and $|\psi\rangle$ is a unit vector, which implies $0 \leq F(\rho, |\psi\rangle\langle\psi|) \leq 1$. Moreover, $F(\rho, |\psi\rangle\langle\psi|) =1$ if and only if $\rho = |\psi\rangle\langle\psi|$. This implies that two states with high fidelity are close to each other. In this paper, we mostly consider the depolarising channels and for this type of channels, if a node $r$ performs a noise free entanglement swap operation between two entangled links with fidelity $F_1$, $F_2$ then the fidelity of the resulting entangled state is at least $F_1F_2$ \cite{BDCZ98}. 

A quantum network may generate entanglement on the demand only when a request arrives, which we call the \emph{on-demand} model. In this case, the routing problem reduces to routing entanglement on the physical communication graph ($G_{\phs} = (V, E_{\phs})$) corresponding to the fibres (or free-space links) connecting the quantum network nodes. This means that entanglement is produced by two nodes connected in $G_{\phs}$ followed by entanglement swapping operations along a path in this graph \footnote{In this case the nodes discover a path from a source $s$ to a destination $e$ in $G_{\phs}$ and use any entanglement distribution scheme (e.g. the schemes proposed in \cite{CZKM97,BDCZ98,DBCZ99,DLCZ01,GWZ08,MLMN09,SJM19}) in a repeater chain for generating entanglement between $s$ and $e$.}. Two such neighbour nodes in $G_{\phs}$ are called physical neighbours of each other. However, we may also pre-establish entanglement between two nodes which do not share a physical connection, in anticipation of future requests. Such pre-shared entanglement forms a \emph{virtual link}~\cite{SMIKW16}. Two such nodes, who are not directly connected by a physical link but share an entangled link or a virtual link, are called \emph{virtual neighbours} of each other. Here, we consider routing on the \emph{virtual graph} given by pre-shared entanglement ($\G = (V, \E )$). This virtual graph may have much lower diameter than the underlying physical one. Such virtual links are in spirit similar to forming an overlay network in peer-to-peer networks \cite{LCPS05,RD01,Rip01,RFHK01,SMKKB01,ZHSR04,MM02,MNR02,CDGR02}, with the important distinction that the graph is highly dynamic: each virtual link can be used only once and it must be re-established before further use. This can be a very time-consuming process. What's more, due to short lifetimes of quantum memories the virtual graph is continuously changing as virtual links expire after some time even if they have not been used.

Of course, on both graphs, one can nevertheless apply classical algorithms to select a path from the sender to the receiver, along which entanglement swapping is performed to create an end to end link. The performances of the centralised shortest path routing algorithms \cite{Dij59} are highly dependent on the network topology. If the topology changes rapidly then keeping the routing tables up to date becomes challenging. This type of situation also occurs in classical delay-tolerant networks \cite{Fall03, JFP04,JLS07,SPR05,HCY11}. Usually, for this type of network, the distributed routing algorithms \cite{JFP04,JLS07,SPR05,HCY11} always perform better than the centralised shortest path algorithm \cite{Dij59}. Hence, in this paper, we address the routing problem by modifying existing classical distributed routing algorithms. The main challenge for designing such algorithms is that the nodes need to decide which operation to perform (entanglement generation or entanglement swap) based on local information. Analysing those routing algorithms for multiple demands is also a challenging task. In this paper we use the mathematical tools from classical routing theory \cite{klein00,klein99,NW04,MNW04} for computing the latency of our proposed routing algorithms for single source and destination pairs. We use numerical simulations for analysing the performances in case of multiple source and destination pairs.

\section{Related Work}

In this paper, we focus on the classical decision making procedure for distributing entanglement in an arbitrary network. The underlying physical mechanism for distributing entanglement in a simple quantum network has a long literature. Distributing entanglement in a simple chain network has been studied before, see e.g., \cite{CZKM97,BDCZ98,DBCZ99,DLCZ01,GWZ08,MLMN09,SJM19, SSDG11} (see \cite{MATN15} for a review). The authors of \cite{ACL07, BDJ09,PLCL13} studied the problem of entanglement distribution from a percolation theory point of view. In their model, first they considered a physical graph and from there they constructed the virtual graph by assuming that two physically connected nodes can have a weakly entangled link with a probability $\pgen$. Given such a virtual graph (which may be disconnected) they were interested in finding whether there exists a path in the virtual graph between any source and destination pair. Their proposed solutions are highly dependent on $\pgen$ and the connectivity of the virtual graph, however, they did not propose any specific routing algorithms in their paper that can find a path between the source and the destination. In \cite{PJSV08, FWHL10,Per10,LCKL12,MGHH14} the researchers proposed solutions for distributing entanglement in a noisy network using the concept of quantum network coding. All of these approaches were based on the specific structure of the physical graphs and manipulation of multi-partite entangled states. However, with current day technologies, these solutions are very difficult to realise in practice.

On the other hand, the routing approaches in \cite{MSLM13, GI18, Cal17,IG12,GI17,SMIKW16} are based on classical techniques and these are arguably more likely to be implemented with the near future quantum technology. In all of these approaches, first, the nodes discover a path from a source to a destination and then distribute the entangled links along the path. The difference between these approaches comes from the path selection algorithms. For example, in \cite{MSLM13} Van Meter et al. defined a cost metric across each communication link and then used Dijkstra's algorithm \cite{Dij59} to find an optimal path between a source and a destination. In \cite{Cal17} Caleffi proposed a routing algorithm that tries to find an optimal path between a source and a destination based on a route metric called end-to-end entanglement rate. However, both of the path selection algorithms are centralised algorithms and work under the assumption that each node has information about the entire network. In \cite{PKTT19} Pant et al., proposed a greedy multi-path routing algorithm for distributing entanglement between a source destination pair in a network. Multi-path routing algorithms are useful for sharing entangled links between a source destination pair. However, this type of approaches is not scalable for a larger network with multiple demands. 

In \cite{SMIKW16,GI17,GI18} the researchers were interested in finding routing algorithms in a virtual graph. In \cite{SMIKW16} Schoute et al. first showed that an efficient construction of the virtual graph can reduce the latency rapidly. In both \cite{SMIKW16,GI18} the researchers first defined a physical graph (mostly ring and grid topology) and classified the nodes into different levels according to their entanglement generation capabilities and then from that physical graph, they constructed a virtual graph, using the concept of a small-world network, proposed in \cite{klein99,klein00,klein02}. Later they used a classical distributed greedy routing algorithm for routing. The main contribution of \cite{GI18} was to adopt the techniques proposed in \cite{Sand06} for constructing the virtual graph in a decentralised fashion.

The main differences between the approaches in \cite{SMIKW16,GI17,GI18}, and this paper is that in our model we put an upper bound on the distance (relative to the physical graph) of a pre-shared entangled link and put an upper bound on the storage time of the entangled link in a quantum memory. This makes our model more realistic. Besides this, we show that due to the dynamic nature of the virtual graph topology, certain routing strategies for the classical network do not work efficiently for a high number of demands and we also propose two new routing algorithms that work better than the one used in the classical case. One can find more details about our contributions in the next sections.

\section{Summary of our Results}
\label{sum_res}
The main contributions in this paper can be subdivided into three parts. In the first part, we propose different types of continuous model networks based on the choice of virtual neighbours. The entangled links between two virtual neighbours can reduce the diameter of the virtual graph, and increase the connectivity of the network. Inspired by complex network theory, we propose the following ways of choosing virtual neighbours.
\begin{itemize}
\item    Deterministically chosen virtual neighbours, where each node, chooses its virtual neighbours using a fixed deterministic strategy depending on the topology. We call this type of virtual graph a deterministic virtual graph.
\item    Randomly chosen virtual neighbours, where each node samples its virtual neighbours using a probability distribution, defined over the set of nodes. Here we study the following two such sampling distributions. We call this type of virtual graph as random virtual graph.
 \begin{itemize}
\item    Virtual neighbours are chosen following the uniform distribution: In this type of network, each node chooses its virtual neighbour uniformly randomly among all the nodes which are at most $d_{\ths}$ distance from it. In the literature of the random graph, this type of virtual neighbours can reduce the diameter of the virtual graph.
\item    Virtual neighbours are chosen following the power law distribution: In this type of network,  each node chooses its virtual neighbour among all the nodes at a distance at most $d_{\ths}$ following a power-law distribution. In the literature of small-world networks, it is well known that for certain kind of physical graph topologies, this type of virtual neighbours exponentially reduces the diameter of the virtual graph \cite{klein00,klein99,klein02}. However, in small-world networks, there is no such concept of $d_{\ths}$ and all the virtual links are permanent in nature. These make our construction different from the construction of \cite{klein00,klein99,klein02}. These differences in the construction also have a significant impact on the design and analysis of the routing algorithms.  
\end{itemize}
\end{itemize}
For details explanation we refer to section \ref{cont_net}.

In the second part, we present the design and analysis of the distributed routing algorithms. In section \ref{routing} we propose the following two routing algorithms.
\begin{itemize}
\item Modified greedy routing algorithm.
\item Local best effort routing algorithm.
\end{itemize}

The main structure of all of the proposed algorithms is similar and it is described in algorithm \ref{algo_greed_dem}. Both of the algorithms can be used in the continuous and on-demand model on any topology and use only local information of the virtual network topology. After getting a demand or request, they discover a path from the source to the destination and reserve the required amount of entanglement across the path. All of the nodes wait until the path is discovered. The one disadvantage of this type of approach is that, if the memory storage time is not large enough, then all of the reserved links might decohere when the path has been discovered. In this paper, we assume that the memory storage time is much longer compared to the classical path discovery time. After the path discovery, if all of the links are available then the intermediate nodes along the path performs entanglement swap. It might happen that due to decoherence, some of the links along the path becomes unusable. In that case, the demand waits until all of the links are being generated. In algorithm \ref{algo_greed_dem} only the path discovery procedure behaves differently in two different algorithms. 

Later in algorithm \ref{path_non_local_best_effort} we propose another version of the local best effort algorithm, where we assume that each node has the information about the virtual graph up to two hops from itself. For the details of these algorithms, we refer to algorithm \ref{path_greed} (modified greedy) and \ref{path_loc_best_eff} (local best effort), \ref{path_non_local_best_effort} (NoN local best effort). In section \ref{prop_net} we analytically show that for a single demand, continuous networks can reduce the latency rapidly compared to the on-demand network. In lemma \ref{lem_cap} we show that the number of entangled links any two nodes can share using the existing pre-shared entangled links in the continuous model is lower bounded by the minimum edge cut of the virtual graph. Later in section \ref{ring_const} we focus our studies on more structured network topologies, like the ring, grid, and recursively generated networks. For the continuous model, we give an analytical upper bound on the number of entanglement swap operations that the nodes have to perform for sharing an entangled link between a source-destination pair. The bounds are given in table \ref{ent_swap}. Based on the results in table \ref{ent_swap}, in lemma \ref{fid} we give a lower bound on the fidelity of the entangled state shared between any source and destination pair. This lower bound shows that for large $d_{\ths}$ the end to end fidelity of the shared entangled state is inversely proportional to the diameter of the ring and grid network. Then we study the behaviour of the proposed routing algorithms and compare their performances with the existing classical routing algorithm for multiple source-destination pairs. Using numerical simulations, we observe the latency of the routing algorithms in the following scenarios as a function of the number of demands. Here we consider set of demands as a $|V| \times |V|$ matrix $D$, where $(i,j)$-th entry $D_{i,j}$ denotes the number of entangled links the source node $i \in V$ wants to share with destination node $j \in V$ at a specific point in time. The summary of our main observations is given below.
 \begin{enumerate}
\item    The latency for classical greedy routing algorithm behaves similarly to the other proposed routing algorithms if the number of demands is small. However, the latency increases rapidly with a higher number of demands. We refer to figure \ref{comp_ring1}, \ref{comp_ring2}, \ref{comp_ring3} for the plots regarding this simulation. This simulation shows the importance of designing new routing algorithms for a quantum internet.
\item    All of the proposed routing algorithms perform much better in the continuous model compared to the on-demand one if the demand is low. I.e., $D_{i,j} \ll cap$ where $cap$ denotes the maximum number of entangled links any two neighbour nodes $u$ and $v$ can share simultaneously. We refer to figure \ref{cap4dem2} for the plots regarding this simulation.
\item    If each demand asks for more entangled links then in figure \ref{cap4dem4} we observe that in a ring network the deterministic virtual graph performs better than the random ones when the total number of demands is very high. However, for the grid network, the situation is opposite. For more detail, regarding this observation, we refer to section \ref{disc}.
\item    Here we also study the performance of the routing algorithms for different continuous models when each of the demands has a high distance in the physical graph. In figure \ref{cap4largedist} we observe that this simulation behaves similarly to the simulation of figure \ref{cap4dem4}.

\end{enumerate}

In the third part, we study our routing algorithms in a heterogeneous network, where the nodes can be classified into several groups based on their entanglement generation capacity. We use the concept of recursively generated graphs for studying such networks. The key idea behind this type of graphs is that first, we start with an initial graph $G_0$ and if any subgraph of $G_0$ has certain structure then we substitute each of the edges of that subgraph with another graph. We call this operation as \textit{edge substitution}. Thus by repeating this procedure recursively we generate the final graph. In lemma \ref{diam_rrgg} we show that the diameter of such recursively generated physical graph increases exponentially with the number of recursive steps. Here, we also propose strategies to construct the corresponding virtual graph and for a single demand, we analytically compute the latency for distributing entangled link between any source-destination pair.

\begin{table*}[h!]
\centering
\begin{tabular}{| m{8em} | m{6cm}| m{6cm} | } 
 \hline
 Models &  Greedy Routing & NoN Routing \\ 
  \hline
Deterministically chosen virtual links & $O(\frac{n}{d_{\ths}} + \log d_{\ths})$ & $O(\frac{n}{d_{\ths}}+\log d_{\ths})$  \\
 \hline
Virtual links chosen with power-law distribution & $O\left(\frac{n}{d_{\ths}} + \log d_{\ths}\right)$ & $O\left(\frac{n}{d_{\ths}} + \frac{\log d_{\ths}}{ \log \log d_{\ths}}\right)$\\
 \hline
Virtual links chosen with uniform distribution & $O\left(\frac{n}{d_{\ths}} + \frac{d_{\ths}}{(\log d_{\ths})^2}\right)$ & $O\left(\frac{n}{d_{\ths}} + \frac{d_{\ths}}{(\log d_{\ths})^2}\right)$ \\
 \hline
\end{tabular}
\vspace{0.2in}
\caption{Expected number of entanglement swap for ring and grid network with single source-destination pair.}
\label{ent_swap}
\end{table*}

\section{Notations and Definitions}
\label{not_def}
In this paper we are interested in computing the latency of a demand. Let a routing algorithm $\mathcal{A}$ takes $T_{\mathcal{A},i,j}$ time steps to distribute $D_{i,j}$ number of entangled links between the source node $i$ and destination node $j$. If $|D|$ denotes the number of non-zero entries in $D$ then with respect to a routing algorithm $\mathcal{A}$ we define \textit{average latency (AL)} as follows,
\begin{equation}
\text{AL} = \frac{1}{|D|}\sum_{i,j}T_{\mathcal{A},i,j},
\end{equation}
where $T_{\mathcal{A},i,j}$ denotes the latency to distribute $D_{i,j}$ entangled link between the nodes $i$ and $j$. In this paper, for a graph $G$, we use $\dist_G(u,v)$, $\diam_G$, $\neigh_G(u)$ to denote hop distance between two nodes $u,v$ in $G$, diameter of $G$ and set of neighbours of $u$ respectively.

\section{Model}
\label{model}
\subsection{Discrete Time Model}

In this paper we consider a simplified discrete time model where each time step is equivalent to the communication time between two neighbouring nodes in the physical graph $G_{\phs} = (V,E_{\phs})$. Here we assume that the distance between any two physical neighbour nodes $u,v \in V$ is upper bounded by $\dist_{\phys}$, i.e., the maximum length of a physical communication link is $d_{\ths}$. This implies in the discrete time model each time step is equivalent to $\frac{\dist_{\phys}}{c}$ time units, where $c$ is the speed of light in the communication channel. 



\subsection{Quantum Network}
\label{qnet}
Each node in the quantum network has the following features.
\begin{enumerate}
\item Each node has a unique id, which carries the information about its location in the physical graph.
\item The nodes are capable of generating $cap$ number of EPR pairs in parallel with each of its neighbours.
\item Each node is capable of storing the created entangled state. However, we assume that the storage is noisy and the fidelity of the stored entangled state decays with each time step. See Section \ref{NQS} for more details. 
If any entangled link is not being used for $T_{\ths}$ time steps then the corresponding nodes won't use that link for entanglement swapping or for teleportation.
\item After $T_{\ths}$ time steps each node throws away the stored entangled state and starts generating a new entangled state. 
\item In this paper we assume that $T_{\ths} \gg \diam_{G_{\phs}}$, i.e, the storage time for quantum memory is much higher than the classical communication time between any two nodes in the network.
\item Each node stores the information about the physical graph topology.
\item Each node stores the information (if available) about its both virtual neighbours (if any) and physical neighbours.
\item The distance between two virtual neighbours in $G_{\phs} = (V,E_{\phs})$ is upper bounded by $d_{\ths}$.
\end{enumerate}

\subsection{Long-Distance Entanglement Creation}
\label{link_gen}

Long-distance entanglement creation is a well-studied subject \cite{CCGZ99,BK05,JKRK16,NTDS16,DHRW17}. For a detailed review, we refer to \cite{MATN15}. In this paper we are interested in a basic model. Due to the threshold time $T_{\ths}$ the entanglement distribution time scales exponentially with the distance between the source destination pair. However, with this type of entanglement generation scheme we can guarantee on the end to end fidelity of the shared entangled links. It is important to note that this model is a trivial one. One can reduce this exponential scaling to a polynomial one by other techniques like \textit{entanglement distillation} or \textit{quantum error correction} \cite{BDCZ98,IM10,MLKL16}. The concepts of these techniques are beyond the scope of this paper. In this paper our main focus is on the decision making procedure for routing, not the physical means of entanglement generation. However, our routing algorithms are designed in such a manner, so that it can work with any entanglement generation procedure. Moreover, the conclusions we draw in this paper about the routing algorithms, remain same with respect to other entanglement generation procedures.

The entanglement generation protocol between two nodes $s,e$ with $\dist_{G_{\phs}}(s,e)=d$, we use in this paper, can be subdivided into two parts,
\begin{itemize}
\item \textit{Elementary link creation}, where the nodes in $\path_{G_{\phs}}(s,e)$, which are connected directly by a physical link, create entangled links between themselves. Usually, the created links are stored in a quantum memory.
\item The next one is called \textit{longer link creation} where the intermediate nodes perform \textit{entanglement swap} operation and share an entangled link between the end nodes.
\end{itemize}

In practice, photon loss in the optical fibre and other imperfections of the network components make the entanglement generation procedure a probabilistic but heralded one. This implies that the elementary link creation procedure can produce a signal which certifies the successful creation of the entangled link. In order to model this here, we assume that \textit{elementary link creation} with probability $P_0 \in [0,1]$ it can create an entangled link within a single time step. For an example of $P_0$, that takes into account the repeater technology based on NV center, we refer to \cite{RYGR18}. 
In this paper, inspired from certain physical implementations \cite{PHHS14}, we assume that \textit{entanglement swap} operation is a deterministic one. We make the simplifying assumption that the time for this operation is negligible.
   
According to our simplified model, $s,e$ can share an entangled state if and only if all the elementary links are being created within the time step $T_{\ths}$. Of course, after the creation of the elementary entangled links the nodes need to communicate each other about it and it will cost some time steps. However, in the last section we assume that $T_{\ths}$ is much larger than this classical communication time. So, for the simplicity, we remove the classical communication time from our calculation. This implies, probability of creating an entangled link between $s,e$ within $T_{\ths}$ time step is at least $(1-(1-P_0)^{T_{\ths}})^d$. As a consequence, the expected entanglement distribution time increases exponentially with $d$. As mentioned before this is a simplified model for entanglement generation but it is useful to guarantee the end to end fidelity of the shared entangled state.

\vspace{-0.25cm}

\subsection{Noisy Quantum Storage}
\label{NQS}
The stored quantum state in quantum memory devices decoheres with time. Here we assume a simplified pessimistic model where the qubits decohere under the effect of symmetric depolarising noise with the parameter $p$. This type of quantum channel can be described as a completely positive trace preserving map $\epsilon: \mathcal{H}^{\otimes 2} \rightarrow \mathcal{H}^{\otimes 2}$.

The noise operator for a two qubits state $\rho_{0,se} = |\psi^+\rangle_{s,e}\langle\psi^+|$, shared between $s,e$, is denoted as $\epsilon$ and on this specific state (not in general) it acts as follows,

\vspace{-0.3cm}

\begin{align*}
&\epsilon(\rho_{0,se}) := p^2\rho_{0,se} + (1-p^2)\frac{\I_4}{4},
\end{align*}
where $\I_4$ is a $4\times 4$ identity matrix.

If at time step $t-1$ the stored state is $\rho_{{t-1},se}$ then at time step $t$ the stored state would be $\rho_{t,se} = \epsilon(\rho_{{t-1},se}).$ By solving this recursive relation on time steps $t$ we get,

\begin{equation}
\rho_{t,se} = p^{2t}\rho_{0,se} + (1-p^{2t})\frac{\I_4}{4}.
\end{equation}

We can rewrite the above expression as $\rho_{t,se} = (\frac{1}{4}+ \frac{3}{4}p^{2t})\rho_{0,se} + \frac{3}{4}(1-p^{2t})\rho^{\perp}_{0,se}$, where $\rho^{\perp}_{0,se}$ is orthogonal to $\rho_{0,se}$. This implies the fidelity of the stored state after $t$ time steps is $F(\rho_{t,se},\rho_{0,se}) = \frac{1}{4}+ \frac{3}{4}p^{2t}$.

\subsubsection{How to Choose $T_{\ths}$}
Entangled links with low fidelity make quantum communication very noisy. In order to protect the information from noise here we fix a threshold value, $F_{\ths}$, for the fidelity. According to our model, each entangled link can be stored inside the quantum memory for $T_{\ths}$ time. This implies $\rho_{T_{\ths},se}$ should satisfy the following condition, $$F(\rho_{T_{\ths},se},\rho_{0,se}) > F_{\ths}.$$ By substituting the value of $F(\rho_{T_{\ths},se},\rho_{0,se})$ from previous section we can get the following bound, up to time $T_{\ths}$.


\begin{align}
\label{ths}
T_{\ths} & \geq \frac{1}{2\log p}\log\left[\frac{4F_{\ths}}{3}-\frac{1}{3}\right].
\end{align}


\subsection{Continuous Network}
\label{cont_net}
In this model each of the nodes $u$ in $G_{\phs} = (V, E_{\phs})$ keeps the information about its $|\neigh_{G_{\phs}}(u)|$ physical neighbours and $O(k)$ virtual neighbours. According to the model, for any long distance neighbour $v$ of $u$, $\dist_{G_{\phs}}(u,v)\leq d_{\ths}$. Each of the nodes establishes and maintains virtual entangled links with all of its neighbours using entanglement generation procedure discussed in section \ref{link_gen}. After some fixed time interval ($T_{\ths}$) the nodes again start to generate all of the links, irrespective of whether there is any demand or not. Borrowing ideas from classical complex network theory here we give more precise description of the following families of continuous network models.

\begin{enumerate}
\item Deterministically chosen virtual neighbours: In the later sections we give a specific strategy for choosing virtual neighbours for ring, grid and recursively generated graph topologies. 
\item Randomly chosen virtual neighbours: 
\begin{enumerate}
\item virtual neighbours are chosen following a uniform distribution: In this network any node $u$ choses another node $v$ ($\dist_{G_{\phs}}(u,v)>1$) as a neighbour with probability $\pchose$, where

\begin{equation}
\pchose(u,v) := 
\begin{cases}
\frac{1}{N_{\leq d_{\ths}}(G_{\phs})},~\dist_{G_{\phs}}(u,v) \leq d_{\ths}\\
 0~~~\text{Otherwise},
\end{cases}
\end{equation}

where $N_{\leq d_{\ths}}(G)$ denotes the number of nodes at a distance at most $d_{\ths}$ from a node $u$.

\item virtual neighbours are chosen following a power-law distribution: In this network any node $u$ choses another node $v$ (such that $\dist_{G_{\phs}}(u,v)>1$) as a neighbour with probability $\pchose$, where

\begin{equation}
\label{pchose}
\pchose(u,v) := 
\begin{cases}
\frac{1}{\beta_{u}}\frac{1}{\dist^{\alpha}_G(u,v)},~\dist_G(u,v) \leq d_{\ths}\\
0~~~\text{Otherwise},
\end{cases}
\end{equation}
where $\beta_u = \sum_{v' \in V} \pchose(u,v')$ and $\alpha > 0$.
\end{enumerate}
\end{enumerate}

\subsection{On-demand Network}

In this model, each node $u$ in the physical graph $G_{\phs} = (V,E_{\phs})$ has only $|\neigh_{G_{\phs}}(u)|$ neighbours. In the on-demand network there are no pre-shared entangled links and each node has information about the entire physical network topology. This implies that if a demand comes, then the source node can compute the shortest path from a source to a destination and starts generating entangled links between the source and destination along that path.

\section{Routing Algorithms}
\label{routing}
In this section we propose three different kinds of distributed routing algorithms. The entire routing procedure between any two nodes in the quantum network can be subdivided into following three phases. 

\begin{enumerate}
\item Path discovery phase.
\item Entanglement reservation phase.
\item Entanglement distribution phase.
\end{enumerate}

One can find the steps of the routing procedure in algorithm \ref{algo_greed_dem}. Among the above three phases the first two phases all the routing decisions are made. In this paper in all of the routing algorithms we focus on different types of path discovery algorithms. Moreover, the algorithms reserve the entangled links while discovering the path. 
The reservation of the links is useful to prevent the utilisation of those links by other demands. After the path discovery and link reservation, if all of the entangled links are available along the path then, the intermediate nodes perform entanglement swap (in parallel) between the links. If there are not enough links available between two neighbours in the path, then the demand waits until all the missing links are being generate. If this missing link generation time is more than $T_{\ths}$ then all of the reserved links along the path expires. In this case the demand waits until the all of the links along the path is being generated. 


In all of the routing algorithms, we assume that each node has the complete information about the \emph{physical} network topology. However, due to the fragile nature of the entangled links, it is difficult to keep track of the current topology of the virtual graph. Here we assume that each node has all the information about the virtual links it shares with its neighbours. During the path discovery each node decides the next hop on the basis of the physical graph topology and the information it has about the shared entangled links with its neighbours. Both of the algorithms try to minimise the average waiting time. Designing such algorithms becomes challenging when there are not enough entangled links available between two neighbours $u,v$. At that moment, for a demand, if $v$ is an optimal neighbour of $u$ to reach a destination node $e$ then the routing algorithm has two options, (i) select the path via $v$ and generate sufficient amount of links between $u,v$, (ii) $u$ try to select another neighbour with whom it already shares enough amount of entangled links. 

The classical greedy routing algorithm \cite{klein99,klein00} always chooses the first option. Here we point out that this is not always a good option for routing in a quantum internet. However, a slight modification of the greedy algorithm can give better performance. In algorithm \ref{path_greed} we describe the modified version of the greedy routing algorithm in detail.
The path discovery algorithm in \ref{path_loc_best_eff} is a best effort algorithm, it chooses the second option. The aim of this algorithm is to utilise all the existing entangled links before generating a new one. In the next two sections, we give a detailed description of both of the algorithms.


\begin{algorithm}
\caption{\small{Distributed Routing Algorithm($s,e, \G, D, cap$)}}
 \label{algo_greed_dem}
\begin{algorithmic}[1]
\State $round = \lceil\frac{D_{s,e}}{cap}\rceil$
\State $i=1$ \Comment{Path discovery phase}
\While{$i \leq round$}
\State $ CommPath_{s,e} = PathDisc(s,e,\G,D_{s,e})$         \Comment{$CommPath_{s,e} = \{s= u_0, u_1, \ldots , u_{d-1}=e\}$}
\If{in $ CommPath_{s,e}$ some of the neighbour nodes do not have enough entangled links}
\State Generate all the links.
\Else
\State $j=1$ \Comment{Communication Phase}
\While{$j\leq d-2$} 
\State $EntSwap(u_0,u_j,u_{j+1})$
\State $j=j+1$
\EndWhile
\EndIf
\State $dem=dem-cap$
\State $i = i+1$
\EndWhile

\end{algorithmic}
\end{algorithm}

\subsection{Modified Greedy Routing}
\label{greed_route}

In the classical network theory, greedy routing is well studied for discovering near optimal paths using only local information \cite{klein99,klein00}. On the basis of the information of its neighbours, each node decides the next hop of a demand. Here, each node tries to jump to as close as possible to the destination node. The virtual neighbours are very useful for making such jumps. One big disadvantage of using the classical greedy routing algorithm for entanglement distribution is that while discovering the path, it doesn't take into account whether there exist enough entangled links available between two virtual neighbours. If the links are not present then the request or demand has to wait until all of the missing links are being generated. This might increase the latency. For example, for a demand, $u$ is the last node discovered in the path and $v$ is a neighbour of $u$ which is the closest to the destination $e$. If there is not enough entangled links available between $u,v$ and if $\dist_{G_{\phs}}(u,v) \geq \dist_{G_{\phs}}(u,e)$ then it is not a good idea for $u$ to put $v$ in the path. In this paper we use a slightly modified version of the algorithm proposed in \cite{klein99} so that it can handle this type of problem. In our modified Algorithm \ref{path_greed} a node $u$ puts the virtual neighbour $v$ in the path, which is the closest to $e$ if at least one of the following two cases are satisfied (See step $5$ of the algorithm \ref{path_greed}). 
\begin{itemize}
\item If the amount of available entangled links between them is greater than $D_{s,e}$. 
\item If $\dist_{G_{\phs}}(u,e) = \dist_{G_{\phs}}(u,v) + \dist_{G_{\phs}}(v,e)$, where $G_{\phs}=(V,E_{\phs})$ is a physical graph, $u$ is the current node and $v$ is its virtual neighbour which is the closest to the destination $e$.
\end{itemize}

If none of these two conditions were satisfied then $u$ chooses its physical neighbour, closest to $e$ as next hop. The details of the algorithm are described in algorithm \ref{path_greed}.


\begin{algorithm}
\caption{\small{Modified Greedy : PathDisc($s,e,\G,D_{s,e}$)}}
 \label{path_greed}
\begin{algorithmic}[1]
\State $u_{\curr}=s$ 
\State $CommPath_{s,e}=\{s\}$
\While{$u_{\curr} \neq e$}
\State $ u =  \argmin_{v \in \neigh_{\G}(u_{\curr})} \dist_{G_{\phs}}(v,e)$
\vspace{0.2cm}
\If{$\left(\substack{(i) \text{ There are not enough links available between $u,u_{\curr}$}\\ \text{\textbf{and}} \\ (ii)  ~\dist_{G_{\phs}}(e,u_{\curr}) < \dist_{G_{\phs}}(u,u_{curr}) + \dist_{G_{\phs}}(u,e)}\right)$}
\vspace{0.2cm}
\State $\small{u = \argmin_{\substack{v \in \neigh_{\G}(u_{\curr}) \\ \dist_{G_{\phs}}(v,u_{\curr}) = 1 }} \dist_{G_{\phs}}(v,e)}$
\EndIf
\State $CommPath_{s,e} = CommPath_{s,e} \cup \{u\}$
\State $u_{curr}=u$
\EndWhile

\State return($CommPath_{s,e}$)

\end{algorithmic}
\end{algorithm}

\subsection{Local Best Effort Routing}
\label{best_eff_route}

The local best effort routing algorithm is a distributed algorithm where each node tries to jump as close as possible to the destination node, using existing entangled links. Unlike the greedy algorithm in this one, a node $u$ first prepares a set of neighbours with whom it shares more than $D_{s,e}$ amount of entangled links. Then from that set, $u$ chooses the next hop $v$ such that $v$ is closest to $e$ and $\dist_{G_{\phs}}(u,e) < \dist_{G_{\phs}}(v,e)$ (See step $4$ of the algorithm \ref{path_loc_best_eff}). If no such neighbour exists, then it chooses its physical neighbour which is the closest to the destination (See step $6$ of the algorithm \ref{path_loc_best_eff}). Algorithm \ref{path_loc_best_eff} describes the procedure in more detail.

\begin{algorithm}
\caption{\small{Local Best Effort : PathDisc($s,e,\G,D_{s,e}$)}}
 \label{path_loc_best_eff}
\begin{algorithmic}[1]
\State $u_{\curr}=s$ 
\State $CommPath_{s,e}=\{s\}$
\While{$u_{\curr} \neq e$}
\State $ u =  \argmin_{\substack{v \in \neigh_{\G}(u_{\curr}) \\ \G(v,u_{\curr}) \geq D_{s,e} }} \dist_{G_{\phs}}(v,e)$
\vspace{0.2cm}
\If{$\left(\substack{(i) \text{ No such $u$ exists} \\ \text{\textbf{or}} \\ (ii)~\dist_{G_{\phs}}(u_{\curr},e) > \dist_{G_{\phs}}(u,e)}\right)$}
\vspace{0.2cm}
\State $\small{u = \argmin_{\substack{v \in \neigh_{\G}(u_{\curr}) \\ \dist_{G_{\phs}}(v,u_{\curr}) = 1 }} \dist_{G_{\phs}}(v,e)}$
\EndIf
\State $CommPath_{s,e} = CommPath_{s,e} \cup \{u\}$
\State $u_{curr}=u$
\EndWhile

\State return($CommPath_{s,e}$)

\end{algorithmic}
\end{algorithm}


\subsection{NoN Local Best Effort Routing}

In this section we are interested in studying the behaviour of the local best effort algorithm when the nodes have information about its neighbours as well as the neighbours of the neighbours in the virtual graph. In the literature of classical complex networks, this type of algorithm is known as NoN routing algorithm. In \cite{NW04,MNW04}, it has been showen that this type of algorithm gives better advantage for routing in certain kind of physical graphs. For example in ring ($C_n$) and grid network ($Grid_{n \times n}$) if virtual neighbours are chosen following a power law distribution then classical NoN routing algorithm can discover a path with expected length $\frac{\log_2 n}{\log_2 \log_2 n}$ \cite{CGS02,NW04,MNW04}. Here we adopt the idea of classical NoN routing algorithm into the local best effort algorithm. As it is a local best effort algorithm, so during the path discovery, the current node $u_{\curr}$ first prepares a set of neighbours $U'$ with whom it shares more than $D_{s,e}$ amount of entangled links (See step $5$ of the algorithm \ref{path_non_local_best_effort}). Then from that set, the algorithm finds a node $u'$ such that it has a neighbour $u$ which is the closest (compare to the neighbours of the nodes in $U'$) to the destination $e$ and the available entangled links between $u,u'$ is more than $D_{s,e}$. If no such $u$ exists then the algorithm just chooses $u$ from the physical neighbours of $u'$ (See steps $6$ to $12$ of the algorithm \ref{path_non_local_best_effort}). If the set $U'$ is empty, then the algorithm chooses the next hop among its physical neighbours (See the steps $14$ to $20$ of the algorithm \ref{path_non_local_best_effort}). For more details, we refer to Algorithm \ref{path_non_local_best_effort}.

\begin{algorithm}
\caption{\small{NoN Local Best Effort : PathDisc($s,e,\G,D_{s,e}$)}}
 \label{path_non_local_best_effort}
\begin{algorithmic}[1]
\State $u_{\curr}=s$ 
\State $CommPath_{s,e}=\{s\}$
\State $d=\dist_{G_{\phs}}(s,e)$
\While{$u_{\curr} \neq e$}
\State $U' = \{u' \in \neigh_{\G}(u): \G(u_{\curr},u') \geq D_{s,e}\}$
\For{all $u' \in U'$}
\State $ u =  \argmin_{\substack{v \in \neigh_{\G}(u') \\ \G(v,u') \geq D_{s,e} }} \dist_{G_{\phs}}(v,e)$
\vspace{0.2cm}
\If{$\left(\substack{(i)\text{ No such $u$ exists} \\ \text{\textbf{or}} \\ (ii)~\dist_{G_{\phs}}(u_{\curr},e) > \dist_{G_{\phs}}(u,e)}\right)$}
\vspace{0.2cm}
\State $u =$ $\argmin_{\substack{v \in \neigh_{\G}(u_{\curr}) \\ \dist_{G_{\phs}}(v,u_{\curr}) = 1 }} \dist_{G_{\phs}}(v,e)$
\EndIf
\If{$d<\dist_{G_{\phs}}(u,e)$}
\State $d= \dist_{G_{\phs}}(u,e)$
\State $u_{\mathrm{next}} = u'$
\EndIf
\EndFor
\If{$U' = \phi$}
\vspace{0.2cm}
\For{all $\left(\substack{u'' \in \neigh_{\G}(u) \\ \text{\textbf{and}} \\ \dist_{G_{\phs}}(u_{\curr},u'') =1}\right)$}
\vspace{0.2cm}
\State $u=\argmin_{\substack{v \in \neigh_{\G}(u'') \\ \G(v,u'') \geq D_{s,e} }}$ $ \dist_{G_{\phs}}(v,e)$
\vspace{0.2cm}
\If{$\left(\substack{(i) \text{ No such $u$ exists} \\ \text{\textbf{or}} \\ (ii)~\dist_{G_{\phs}}(u_{\curr},e) > \dist_{G_{\phs}}(u,e)}\right)$}
\vspace{0.2cm}
\State $u =$ $ \argmin_{\substack{v \in \neigh_{\G}(u_{\curr}) \\ \dist_{G_{\phs}}(v,u_{\curr}) = 1 }} \dist_{G_{\phs}}(v,e)$
\EndIf
\If{$d<\dist_{G_{\phs}}(u,e)$}
\State $d= \dist_{G_{\phs}}(u,e)$
\State $u_{\mathrm{next}} = u''$
\EndIf
\EndFor
\EndIf
\State $CommPath_{s,e} = CommPath_{s,e} \cup \{u_{\mathrm{next}}\}$
\State $u_{curr}=u_{\mathrm{next}}$
\EndWhile

\State return($CommPath_{s,e}$)

\end{algorithmic}
\end{algorithm}

\section{Properties of the Routing Algorithms in General Network}
\label{prop_net}
The performance analysis of the distributed routing algorithms in a general virtual graph is a challenging task. However, due to the greedy nature of the routing algorithms in next lemma we can give an upper bound on the latency of a demand for all of the proposed routing algorithms. 

\begin{lemma}
\label{routing_time}
For any physical graph $G_{\phs}$, and a demand matrix $D$ with $|D|=1$, for any source and destination pair $s,e$, with $D_{s,e}=1$, for both modified greedy and local best effort algorithms the expected waiting times are upper bounded by,
\begin{align}
&O(P^{-\diam_{G_{\phs}}})~~~~\text{On-demand model}\\
&O(d_{\ths}\diam_{G_{\phs}})~~~~\text{Continuous model},
\end{align}
where $P = (1-(1-P_0)^{T_{\ths}})$, i.e, the probability of creating an entangled link between two neighbour nodes (in $G$) within $T_{\ths}$ time step. 
\end{lemma}

\begin{proof}
In $G$, the maximum distance $d$ between any source and destination pair is $\diam_{G_{\phs}}$, i.e., $d\leq diam_{G_{\phs}}$. In the on-demand model, as there is no pre-shared entanglement so the proof directly follows from the fact that expected time to generate a $d$-distance entangled link is $(1-(1-P_0)^{T_{\ths}})^d$.

In the continuous model, each node shares the entangled links with its neighbours. As $|D|=1$ and $D_{s,e}=1$, so for the demand $s$ first discovers a path to $e$ and reserve an entangled link along that path. Then entanglement swapping is used along that path to distribute the entanglement. As the nodes do not need to perform any entanglement generation during the routing, and entanglement swap operation is deterministic, the worst case latency would be the time to discover the path from $s$ to $e$ and the time to inform all the nodes for performing the entanglement swap operation. In both of the routing algorithms if a node $u$ puts its neighbour $v$ in the path then $\dist_{G_{\phs}}(v,e) < \dist_{G_{\phs}}(u,e)$. This implies at each step of the path discovery phase the distance from the current node to the destination node is reduced by at least one. As the maximum length of any virtual link is $d_{\ths}$, this implies the total latency is upper bounded by $O(d_{\ths}\diam_G)$. 
\end{proof}
In the continuous model, the pre-shared entangled links are temporary in nature and one link can only be used for once. If $D_{s,e}> cap$ then in this model it is better idea to find different edge-disjoint paths between $s$ and $e$ and use each path to distribute $cap$ number of entangled links. In the next lemma we give an upper bound on $D_{s,e}$, such that the nodes do not need to wait for creating new entangled links.  
\begin{lemma}
\label{lem_cap}
In the continuous network model, for any virtual graph $\G$, and for a demand matrix $D$ with $|D|=1$, if $\mincut_{\G}$ denotes the minimal cut of the graph $\G$ then any two nodes $s,e$ in $\G$ can distribute at least $\mincut_{\G}.cap$ number of EPR pairs before regenerating any new entangled link.
\end{lemma}

\begin{proof}
In a virtual graph, an entangled link can only be used once for teleporting one qubit. In the continuous model we start with a connected virtual graph $\G$ and for sharing an entangled link between a source-destination pair $s,e$ we have to use one entangled link along for each edge along the path from $s$ to $e$. This implies $s,e$ can use the existing links for routing until the graph becomes disconnected. By Menger's Theorem this is equal to the minimal cut of an undirected graph.
\end{proof}

 The fragile nature of the entangled links affects the network topology. Due to this feature, if $|D|>1$ then the topology of the network changes rapidly and sometimes it might become disconnected. This triggers an interesting question regarding the performance of the routing algorithms. The next section focuses on this issue.

\section{Analysis of the Routing Algorithms on Ring and Grid Network Topology}

\label{ring_const}
In this section we consider two physical graphs, a ring network $C_{n}$ and a grid network $Grid_{n' \times n'}$. For the ring network each node in the network has a unique id from $\{0,1, \ldots , n-1\}$. For the grid network we assume that each node has node id $u=(u_a,u_b)$, where $u_a,u_b \in \{0,1, \ldots , n'-1\}$. For the simplicity of our calculation we assume for the ring network, $n=2^m$ for some positive integer $m$ and $d_{\ths} = 2^l$ for some positive integer $l \leq m$. 

 In the next subsections we define how to choose the virtual neighbours in the continuous model and the on demand model.
\subsection{Continuous Model}
\label{qig_ring}

\subsubsection{Virtual $\G$ with Deterministically Chosen entangled links}
\label{det_ring}
In this model, for the ring network, each node chooses its virtual neighbours with a deterministic strategy. If any node with node id $x$ is of the form $2^iy$, where $i,y$ are non-negative integers and $y$ is an odd number or $0$, then that node has $\min(2i,2\log_2 d_{\ths})$ virtual neighbours. Node id of those virtual neighbours are $x+2 (mod~n), x+2^2(mod~n),\ldots , x+2^{\min(i,\log_2 d_{\ths})}(mod~n)$ and $x+n-2 (mod~n), x+n-2^2(mod~n),\ldots , x+n-2^{\min(i,\log_2 d_{\ths})}(mod~n)$. One can also find this type of virtual network topology in \cite{SMIKW16}.  However, in \cite{SMIKW16} Schoute et al. constructed the network only for $d_{\ths}=\diam_{G_{\phs}}$. In figure \ref{chord_8} we give an example of such network.

Similarly, for grid network, each node with node id $u=(u_a,u_b)$ of the form $(2^ix,2^jy)$, where $i,j$ are non-negative integers and $x,y$ are positive odd numbers or zero, chooses 
$k=\min(2\min(i,j),2\log_2 d_{\ths})$ virtual neighbours. The node id of those virtual neighbours are $(u_a+2 (mod~n+1),u_b), (u_a+2^2(mod~n+1),u_b),\ldots , (u_a+2^{k}(mod~n+1),u_b)$ and $(u_a,u_b+2 (mod~n+1), (u_a,u_b+2^2(mod~n+1)),\ldots , (u_a,u_b+2^{k}(mod~n+1))$. In figure \ref{grid_5} we give an example of such network.

\begin{figure}[h!]
\centering
\includegraphics[scale=0.5]{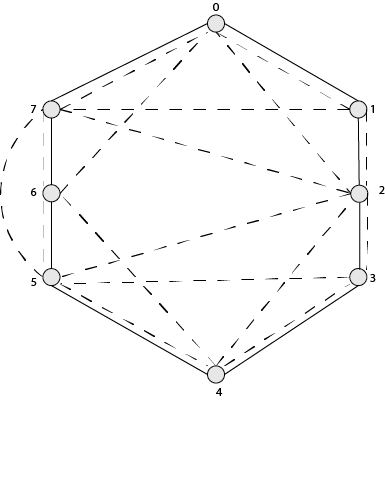}
\caption{Quantum Internet Graph $\G$ with Deterministically Chosen Virtual Links. Here each of the virtual links are denoted by a dotted line. In this example, $d_{\ths}= 2$. Nodes with even node id are of the form $2y$, so, each of such nodes have $2$ virtual neighbours and nodes with odd node id has no long distance neighbour.}
\label{chord_8}
\end{figure}

\begin{figure}[h!]
\centering
\includegraphics[scale=0.5]{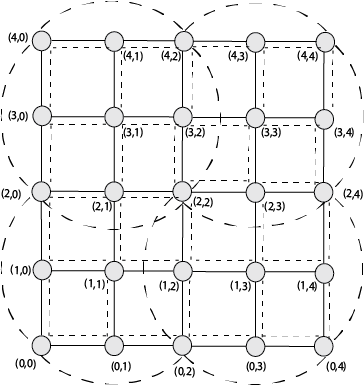}
\caption{Quantum Internet Graph $\G$ with Deterministically Chosen Virtual Links. Here each of the virtual links are denoted by a dotted line. In this example, $d_{\ths}= 2$. Nodes with even node id are of the form $2y$, so, each of such nodes have $2$ virtual neighbours and nodes with odd node id has no long distance neighbour.}
\label{grid_5}
\end{figure}


\subsubsection{Randomly Chosen virtual neighbours}
\label{power_ring}
In this model, for both of the network topologies, each node $u \in V$ has $k = \log_2 d_{\ths}$ virtual neighbours (other than the physical neighbours) and each of such neighbours $v$ is chosen independently without replacement, following the distribution $\pchose(u,v)$. For the ring network, if $\pchose(u,v)$ follows power law distribution then the diameter of the virtual graph is optimal $\alpha=1$. In this simulation we also assume $\alpha=1$. For the grid network  the diameter of the virtual graph is optimal for $\alpha =2$ \cite{klein00}. In this paper we assume $\alpha = 2$.

\subsection{On Demand Model}

In the on-demand model, the nodes do not have any virtual neighbours. So, for this model, we run our simulation with the physical graphs.

\subsection{Lower bound on the Fidelity for Single Source-Destination Pair}
In this section, we assume that for all types of virtual graphs, all of the pre-shared links are available during the routing. 
\begin{theorem}
\label{ent_swap_begin}
In the continuous model with $d_{\ths} > 2$, if the virtual graphs are constructed from a physical network ($G_{\phs}$) like a ring ($C_n$) or a grid ($Grid_{n \times n}$) topology, if $|D|=1$ and for all $i,j \in [0,n-1]$ and if $D_{i,j}\in \{0,1\}$ then for any source destination pair $s,e$, the expected number of required entangled swap operations for sharing an entangled link between $s,e$ is as given in table \ref{ent_swap}.
\end{theorem}

\textit{Proof Sketch :}
For the deterministic virtual graphs, the required number of swap operations, presented in table \ref{ent_swap}, remain same for all of the routing algorithms and the proofs directly follow from the adaptation of the proofs given in \cite{SMIKW16}. 

Here we give the idea about the proof techniques for the random virtual graphs. Please check appendix \ref{proof_thm1} for a detailed proof.

In all of the proofs, first we divide all of the nodes in the graph into $m'=\frac{2\dist_{G_{\phs}}(s,e)}{d_{\ths}}$ sets, $Z_0,Z_1, \ldots , Z_{m'-1}$, such that for all $0\leq i \leq m'-1$, $$Z_i = \left\{u : |u-e| \le |s-e| - \frac{id_{\ths}}{2}\right\}.$$ In the path discovery phase, all of the algorithms start from the node $s \in Z_0$ and then continue to visit other nodes $Z_1, Z_2,.., Z_{m'-1}$. Any node in the set $Z_{m'-1}$ is at most $\frac{d_{\ths}}{2}$ distance away from the destination $e$. In the proof we show that all of the algorithms discover a path from a node $u \in Z_i$ to another node $v \in Z_{i+1}$ (where $0 \leq i \leq m'-2$) using constant number of hops, i.e, $u$ can share an entangled link with $v$ using a constant number of entanglement swap operations. This proves that after $O(m')$ number of swaps $s$ can share an entangled link with a node $u$ such that $\dist_{G_{\phs}}(u,e) \leq \frac{d_{\ths}}{2}$. As $\dist_{G_{\phs}}(s,e) \leq \diam_{G_{\phs}} = O(n)$ for both $C_n$ and $Grid_{n \times n}$, this implies $m' = O(\frac{n}{d_{\ths}})$.

Next, we give upper bounds on the number of required entanglement swap operations for sharing an entangled link between a node $u\in Z_{m'-1}$ to the destination $e$.

For this, we partition all of the nodes in $Z_{m'-1}$ into $m$ sets, $X_0 \supset X_1 \supset \ldots \supset X_{m-1}$. For local best effort and modified greedy algorithm, for all $0\leq i \leq m-1$, 
\begin{align}
\label{greed}
X_i &= \left\{u : |u-e| \le \frac{d_{\ths}}{2^i}\right\}
\end{align}
For NoN local best effort algorithm, for all $0\leq i \leq m-1$,
\begin{align}
\label{non}
X_i & = \left\{u : |u-e| \le \frac{d_{\ths}}{(\log_2 d_{\ths})^i}\right\}.
\end{align}

\begin{itemize}
\item For the power-law virtual graphs, we show that, for all of the routing algorithms, a node $u' \in X_i$ reaches to another node $v' \in X_{i+1}$ (where $0 \leq i \leq m-2$) using constant number of hops, i.e, $u'$ can share an entangled link with $v'$ using a constant number of entanglement swap operations. This implies any node $u \in Z_{m'-1}$ can share an entangled link with a node $v' \in X_{m-1}$ using $O(m)$ number of entanglement swap operations. 
\begin{itemize}
\item For the local best effort and the modified greedy routing algorithms, if we choose $m = \log_2 d_{\ths}$, then according to the partitions in equation \ref{greed}, all the nodes $v' \in X_{m-1}$ are constant distance away from the destination $e$. This implies for both of the routing algorithms the total number of required entanglement swap is $O(\frac{n}{d_{\ths}}+ \log_2d_{\ths})$.
\item For the NoN local best effort routing algorithm, if we choose $m = \frac{\log_2 d_{\ths}}{\log_2 \log_2 d_{\ths}}$, then according to the partitions in equation \ref{non}, all of the nodes $v' \in X_{m-1}$ are constant distance away from the destination $e$. This implies for the NoN local best effort routing algorithms the total number of required entanglement swaps is $O(\frac{n}{d_{\ths}}+ \frac{\log_2 d_{\ths}}{\log_2 \log_2 d_{\ths}})$.
\end{itemize}   
\item For the uniform virtual graphs, for the local best effort and the modified greedy routing algorithms we show that a node $u' \in X_i$ reaches to another node $v' \in X_{i+1}$ (where $0 \leq i \leq m-2$) using $\frac{2^{i}}{\log_2 d_{\ths}}$ number of hops. If we choose $m=O\left(\log_2\left(\frac{d_{\ths}}{\log_2 d_{\ths}}\right)\right)$ then we get $\sum_{i=0}^{m-1} \frac{2^{i}}{\log_2 d_{\ths}} = O(\frac{d_{\ths}}{(\log_2d_{\ths})^2})$. This implies that both of the routing algorithms require $O(\frac{n}{d_{\ths}}+ \frac{d_{\ths}}{(\log_2d_{\ths})^2})$ entanglement swap operations for distributing entangled link between any source and destination pair.
\item For the uniform virtual graphs, for the NoN local best effort routing algorithm we show that a node $u' \in X_i$ reaches to another node $v' \in X_{i+1}$ (where $0 \leq i \leq m-2$) using $(\log_2 d_{\ths})^{i}$ number of hops. If we choose $m=O\left(\log_{\log_2 d_{\ths}}\left(\frac{d_{\ths}}{\log_2 d_{\ths}}\right)\right)$ then then we get $\sum_{i=0}^{m-1} \frac{2^{i}}{\log_2 d_{\ths}} = O(\frac{d_{\ths}}{(\log_2d_{\ths})^2})$. This implies that NoN local best effort routing algorithm require $O(\frac{n}{d_{\ths}}+ \frac{d_{\ths}}{(\log_2d_{\ths})^2})$ entanglement swap operations for distributing entangled link between any source and destination pair.
\end{itemize}

\qed

\begin{lemma}
\label{fid}
For a ring network $C_n$ and a grid network $Grid_{n\times n}$, in the continuous model, if $|D|=1$ and for some $s,e \in [0,n-1]$, if $D_{s,e}=1$ then for both local best effort and modified greedy routing algorithms the expected fidelity of the shared entangled link between $s,e$ is lower bounded by 
\begin{align}
\nonumber
&~~\text{Deterministic and Power Law Virtual Graphs}\\
&F^{O(\frac{n}{d_{\ths}} + \log d_{\ths})}\\ \nonumber
&~~\text{Uniform Virtual Graphs}\\
&F^{O(\frac{n}{d_{\ths}} + \frac{d_{\ths}}{(\log d_{\ths})^2})},
\end{align}
 where each of the pre-shared link has fidelity $F$. 
\end{lemma}
\begin{proof}
The proof of this lemma directly follows from the upper bounds on the number of entanglement swap operation, given in table \ref{ent_swap} and the fact that the fidelity of the entangled state, generated after performing one entanglement swap operation between two links with fidelity $F$, is $F^2$. 
\end{proof}
We can get the following corollary by replacing $d_{\ths}$ by the diameter of $C_n$ and $Grid_{n\times n}$.
\begin{corollary}
\label{fid_corr}
For a ring network $C_n$ and a grid network $Grid_{n\times n}$, in the continuous model with $d_{\ths}=\lceil\frac{n}{2}\rceil$, if $|D|=1$ and for some $s,e \in [0,n-1]$, if $D_{s,e}=1$ then for both local best effort and modified greedy routing algorithms the expected fidelity of the shared entangled link between $s,e$ is lower bounded by 
\begin{align}
\nonumber
&~~\text{Deterministic and Power Law Virtual Graphs}\\
&F^{O( \log n)}\\ \nonumber
&~~\text{Uniform Virtual Graphs}\\
&F^{O(\frac{n}{(\log n)^2})},
\end{align}
 where each of the pre-shared link has fidelity $F$. 
\end{corollary}
\subsection{Performance of Greedy Routing Algorithm with Multiple Source-Destination Pairs on Ring Network}

In this section, we study the performance of the greedy routing algorithm when we have multiple source-destination pairs. We evaluate the average entanglement distribution time for all of the models using MATLAB simulations. 

\subsubsection{Simulation Setup} 
For the simulation we choose $C_{32}$ and $Grid_{5 \times 5}$ as the physical graphs. We assume that the $\dist_{\phys}=10 km$, this implies, each time step is equivalent to $\frac{10}{c} \approx 0.00006$ seconds ($c$ denotes the speed of light in the optical fiber). In this paper we assume the threshold fidelity $F_{\ths}$ is $0.8$ and $p=0.9993$. If we plug in all the values of these parameters in equation \ref{ths} then we get $T_{\ths} \approx 1000$ steps, which is equivalent to $0.06$ seconds. Note that, the maximum classical communication time between any two nodes in both grid and ring network is of the order $10^{-5}$ seconds, which is much less than $T_{\ths}$. This makes the classical communication time negligible compare to the memory storage time. 

In present day practical implementations for basic heralded entanglement link generation the value of $P_0$ is more than of the order $10^{-4}$ \cite{SSMS14,RYGR18}. Here, for the analysis of worst case scenario we assume $P_0 = 0.0003$. If we plug in the values of all the parameters in the expression $(1-(1-P_0)^{T_{\ths}})$, then we get the probability of generating an entangled link between two physical neighbour nodes within the time window $T_{\ths}$ is at least $0.25$.

On top of the physical graph, we construct the virtual graph using the techniques proposed in section \ref{ring_const}. For each of the virtual graph, first we fix $|D|$ and then randomly generate the demand matrix $D$. Once we generate the demand matrix, then for each value of $D_{i,j}$ we compute the latency for entanglement distribution and update the topology of the virtual graph. For each value of $|D|$ we take $10000$ samples of $D$ and compute the \textit{AL} by taking the average over all the samples. In the simulation we also assume that $D_{i,j} \leq cap$. 

For routing in the continuous model, first the nodes discover a path from a source to a destination and reserve the entangled links along the path. If all of the required entangled links are available then the intermediate nodes just perform entanglement swap operation, otherwise the demand waits until all of the virtual links along the path are being generated. During the path discovery and entanglement reservation phase, we assume that there is no collision between two demands for reserving one link. One can easily overcome this assumption by adding some priority corresponding to each demand. The nodes distribute an entangled link between a source and a destination using the procedure proposed in section \ref{link_gen}. For each demand we compute the waiting time, which is composed of path discovery time, time for generating missing entangled links and communication time for informing all the nodes in the path for performing entanglement swap operation.

For the randomly chosen virtual neighbour models, we first sample the neighbours and construct the corresponding virtual graph from the physical graphs. We take $100$ such samples of virtual graphs and for each of such sample graphs, we compute the \textit{AL} by taking the average over all the samples of $|D|$ and virtual graphs. 
On the basis of this setup, in the next sections, we study the effect of different design parameters on the \textit{AL}. 

\subsubsection{Discussions on the Simulation Results}
\label{disc}
We begin our discussion with the simulations of the figures \ref {comp_ring1}, \ref {comp_ring2}, \ref {comp_ring3}. In all of the subfigures, we compare the proposed distributed routing algorithms. All of the figures basically show our intuition mentioned in section \ref{greed_route}, regarding the inefficiency of the classical greedy routing algorithms. The witness of inefficiency is more evident for $d_{\ths}=4$, as the penalty for an unwise decision of jumping increases with the distance. A closer observation of the figures also reveals another fact that local best effort performs better than the modified greedy algorithm when $|D|$ is smaller and the situation changes for the favour of modified greedy algorithms with an increase of $|D|$. The intuitive reason behind it is that the local best effort algorithm tries to use the existing entangled links as much as possible before generating a new one. This gives an advantage over the modified greedy algorithm. However, for large $|D|$ it quickly exhausts all of the pre-shared links and converges to the on-demand network model. 

Next, we focus on comparing the performance of both of the routing algorithms when the nodes pre-share four EPR pairs with each of its neighbours (see figure \ref{cap4dem2}) and each demand can ask for at most two EPR pairs. Intuitively, for these type of demands, continuous network model takes the advantage of the pre-shared entangled links and it outperforms the on-demand one.

In figure \ref{cap4dem4} we are interested in studying the behaviour of the routing algorithms when the nodes demand a high number (between two to four) of entangled links and each node pre-share four entangled links with each of its neighbours. For the ring network, in figure \ref{ringcap4dem4a}, \ref{ringcap4dem4c} we observe that for higher number of demands, deterministic virtual graphs are better compared to randomised one. However, figures \ref{gridcap4dem4a} and \ref{gridcap4dem4c} suggest that for the grid network it is wise to choose the virtual neighbours uniformly randomly. 

In the simulation corresponding to figure \ref{cap4largedist} all of the source destination pairs have distance at least $d_{\ths}$ and we observe that this scenario resembles with the simulations of figure \ref{cap4dem4}, where each source and destination pair asks for a high number of entangled links. We refer to such demands as long-distance demands. The intuitive reason for this similar behaviour is that here the source and destination pairs are generated randomly and as the size of the set of long-distance source and destination pairs are much larger, so most of the demands in the simulation of figure \ref{cap4dem4} are long distance demands. As the impact of the long-distance demands on latency is large compared to the short distance ones, so on average both of the simulations behave similarly. 

In all of the simulations we observe that for the grid network, the separation in AL between the on-demand and the continuous model is very high compare to the ring network. This can be explained by using the idea of lemma \ref{lem_cap}. Due to the higher connectivity among the nodes in the grid network it is easier to find edge disjoint paths for different demands.


\begin{figure*}[ht]
 \begin{subfigure}{4cm}
    \centering\includegraphics[width=8cm]{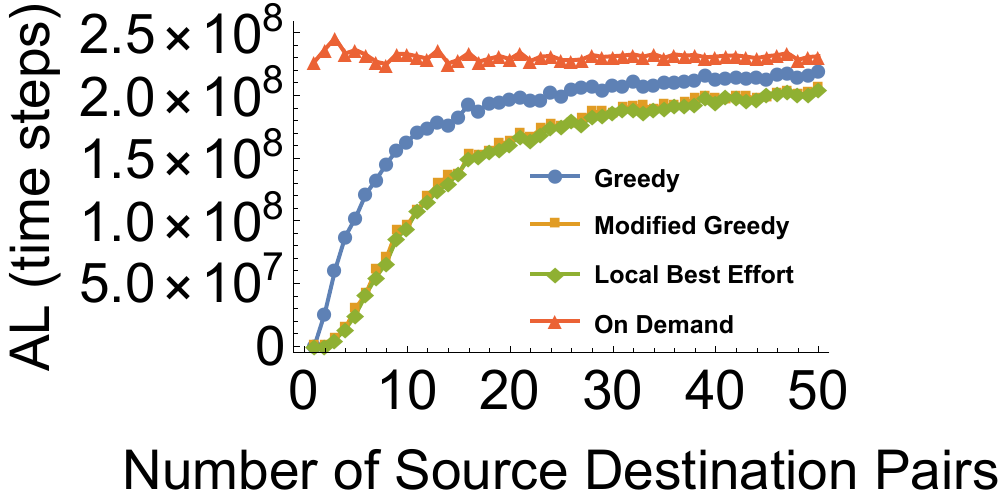}
    \caption{Ring,Deterministic $d_{\ths}=2$}
    \label{det_comp2}
  \end{subfigure}
  \hspace{5cm}
   \begin{subfigure}{4cm}
    \centering\includegraphics[width=8cm]{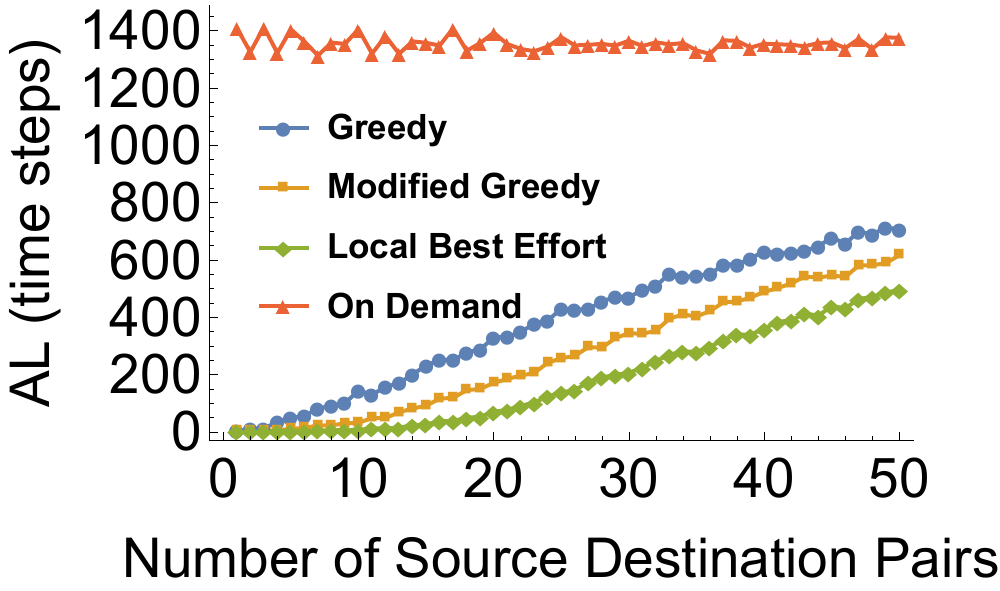}
    \caption{Grid, Deterministic $d_{\ths}=2$}
    \label{det_comp_grid2}
  \end{subfigure}

  \begin{subfigure}{4cm}  
   \centering \includegraphics[width=8cm]{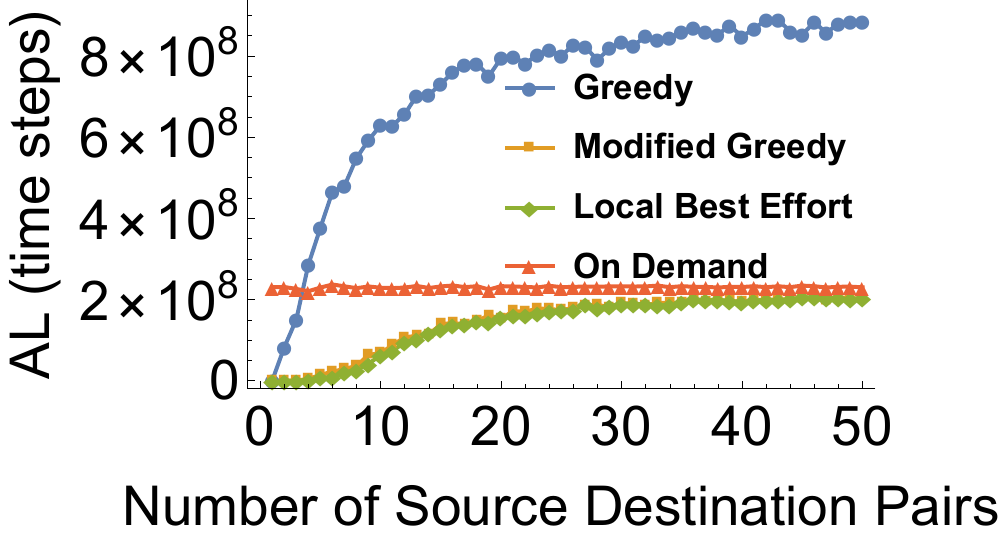}
    \caption{Ring,Deterministic $d_{\ths}=4$}
    \label{det_comp4}
  \end{subfigure}
 \hspace{5cm}
  \begin{subfigure}{4cm}
    \centering\includegraphics[width=8cm]{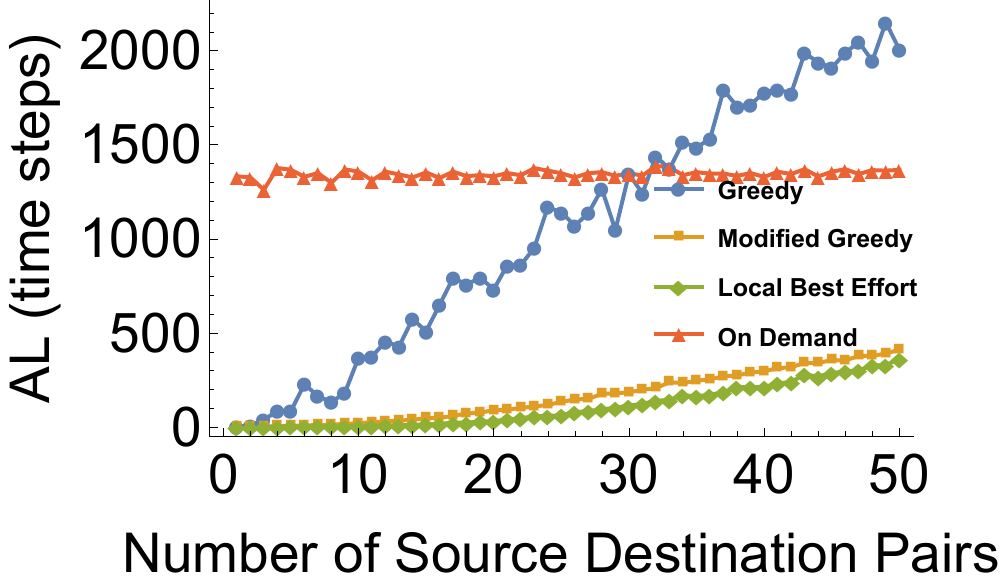}
    \caption{Grid, Deterministic $d_{\ths}=4$}
    \label{det_comp_grid4}
  \end{subfigure}
   \caption{Comparison of different routing algorithms on the deterministic virtual graphs, with $D_{i,j}=1$ and $cap=1$.}
\label{comp_ring1}
   \end{figure*}   
   
 \begin{figure*}[ht]
    \begin{subfigure}{4cm}
    \centering\includegraphics[width=8cm]{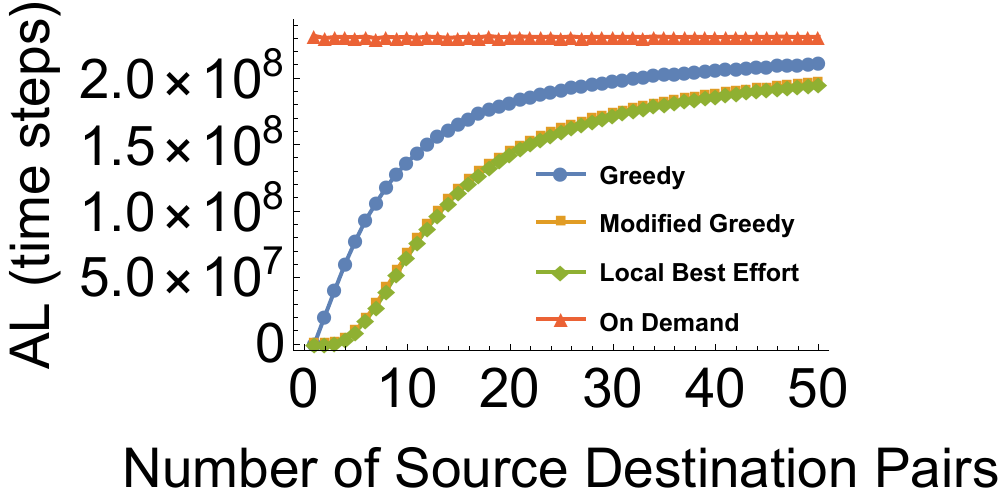}
    \caption{Ring, Power Law, $d_{\ths}=2$}
    \label{power_comp2}
  \end{subfigure}
   \hspace{5cm} 
  \begin{subfigure}{4cm}
    \centering\includegraphics[width=8cm]{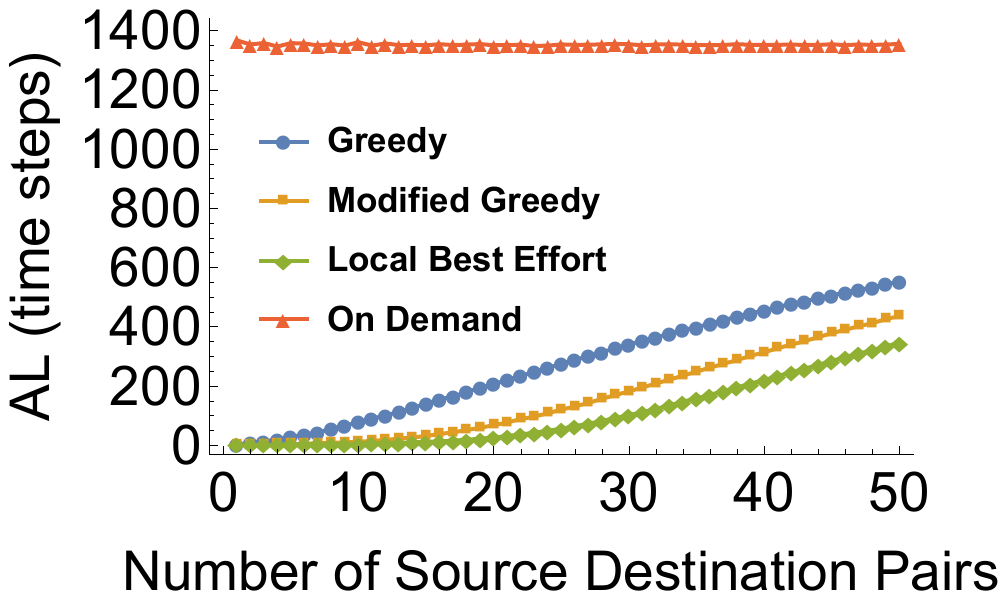}
    \caption{Grid, Power Law, $d_{\ths}=2$}
    \label{power_comp_grid2}
  \end{subfigure}

  \begin{subfigure}{4cm}
    \centering\includegraphics[width=8cm]{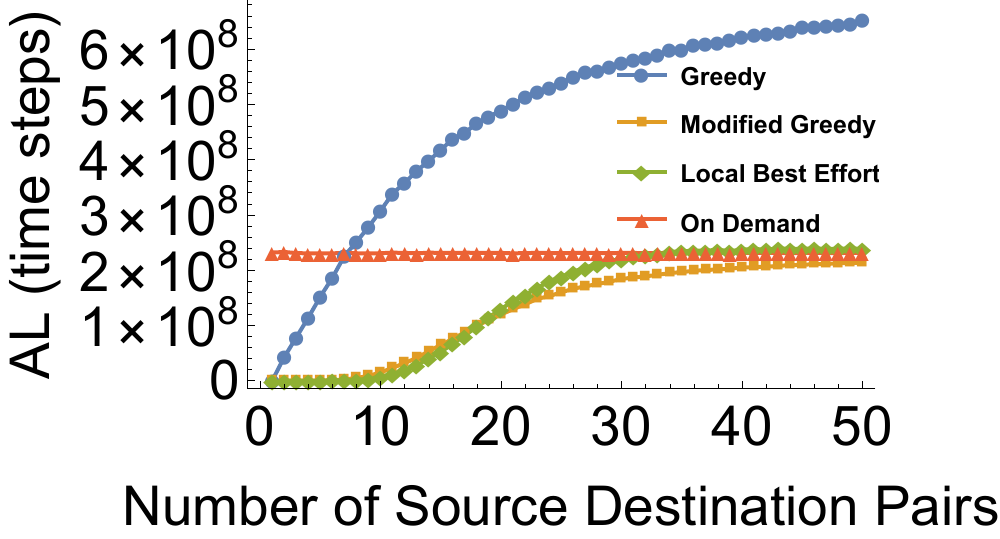}
    \caption{Ring, Power Law, $d_{\ths}=4$}
    \label{power_comp4}
  \end{subfigure}
   \hspace{5cm}
    \begin{subfigure}{4cm}
    \centering\includegraphics[width=8cm]{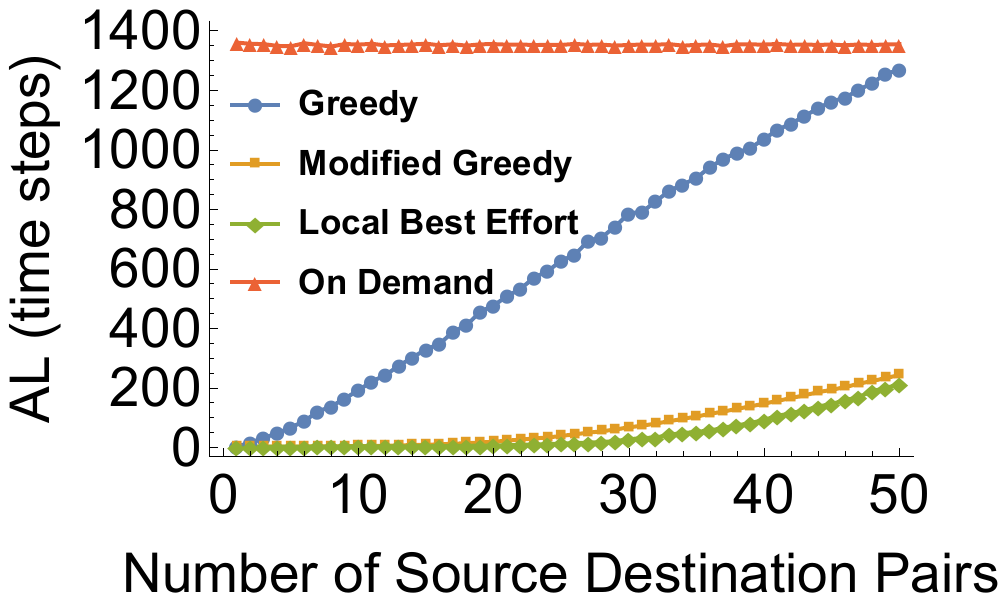}
    \caption{Grid, Power Law, $d_{\ths}=4$}
    \label{power_comp_grid4}
  \end{subfigure}
  \caption{Comparison of different routing algorithms on the power-law virtual graphs, with $D_{i,j}=1$ and $cap=1$.}
\label{comp_ring2}
 \end{figure*}
 
 In figure \ref{capnondist} we compare the performance of the NoN local best effort algorithm with other routing algorithms. In figure \ref{ringcap4nona} we can observe that NoN algorithms do not give any advantage in deterministic ring network. In table \ref{ent_swap} we also have this kind of observation. However, for random virtual graphs, in figure \ref{ringcap4nonb} and \ref{ringcap4nonc} we observe that this algorithm performs better than the other two algorithms for smaller number of demands. If the number of demands increases then we observe the worst performance for NoN local best effort algorithm and the best (among the three proposed algorithms) performance for the modified greedy routing algorithm. The main reason behind this is that the rate of consuming existing the existing entangled link for NoN local best effort algorithm is higher than the other ones. Due to this feature, for NoN local best effort algorithm the continuous model converges to the on-demand model very fast, hence after certain number of demands it shows worse performance.

 \begin{figure*}[ht]
   \begin{subfigure}{4cm}
    \centering\includegraphics[width=8cm]{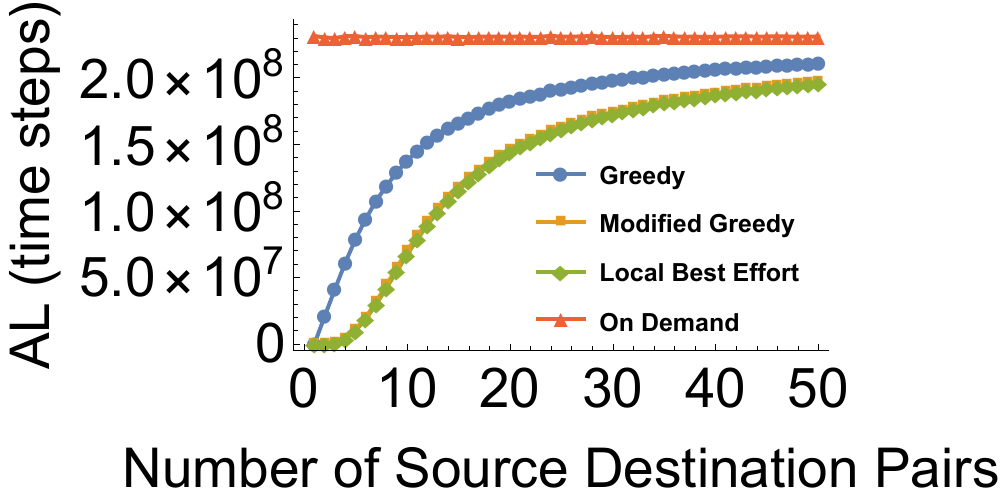}
    \caption{Ring, Uniform, $d_{\ths}=2$}
    \label{unif_comp2}
  \end{subfigure}
   \hspace{5cm}
   \begin{subfigure}{4cm}
    \centering\includegraphics[width=8cm]{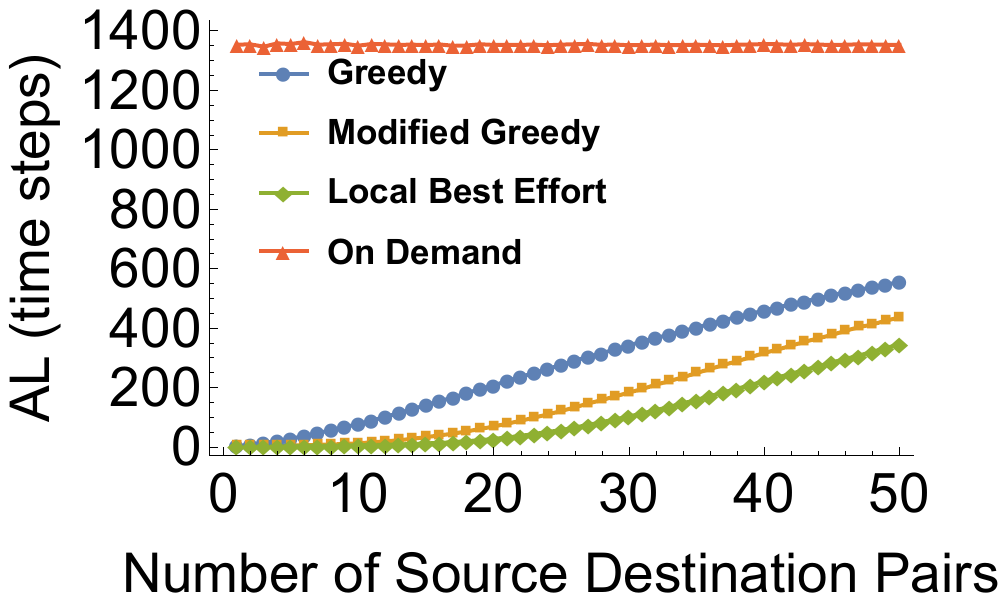}
    \caption{Grid, Uniform, $d_{\ths}=2$}
    \label{unif_comp_grid2}
  \end{subfigure}

 \begin{subfigure}{4cm}
    \centering\includegraphics[width=8cm]{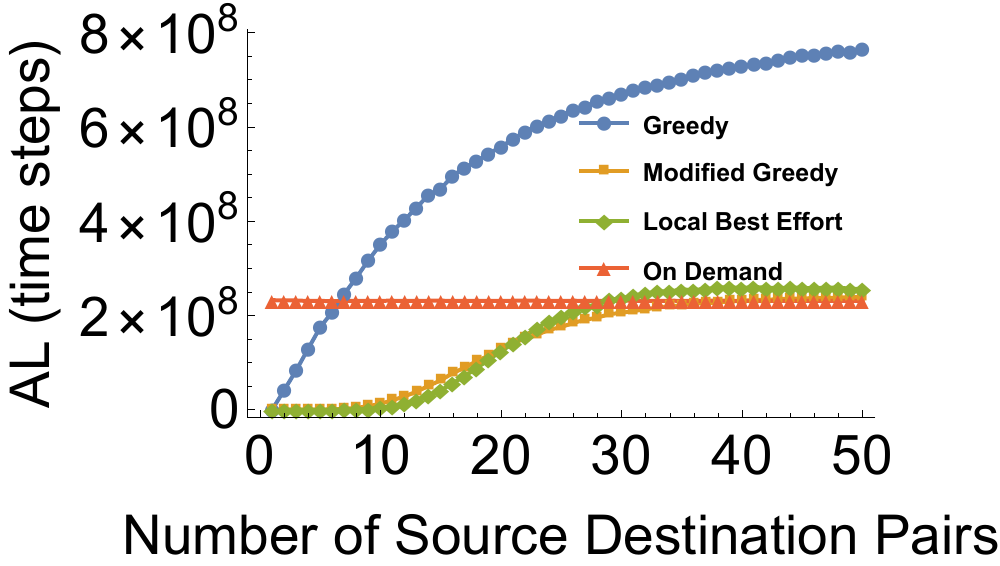}
    \caption{Ring, Uniform, $d_{\ths}=4$}
    \label{unif_comp4}
  \end{subfigure}
  \hspace{5cm}
  \begin{subfigure}{4cm}
    \centering\includegraphics[width=8cm]{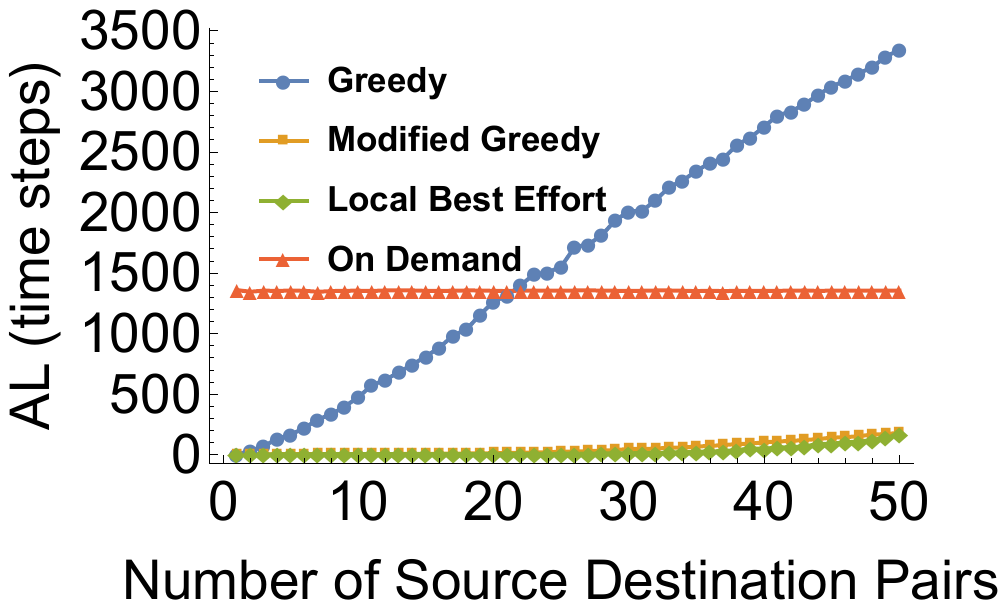}
    \caption{Grid, Uniform, $d_{\ths}=4$}
    \label{unif_comp_grid4}
  \end{subfigure}
\caption{Comparison of different routing algorithms on uniform virtual graphs, with $D_{i,j}=1$ and $cap=1$.}
\label{comp_ring3}
\end{figure*}  



\begin{figure*}[tp]
\begin{subfigure}{4cm}
    \centering
    \includegraphics[width=7cm]{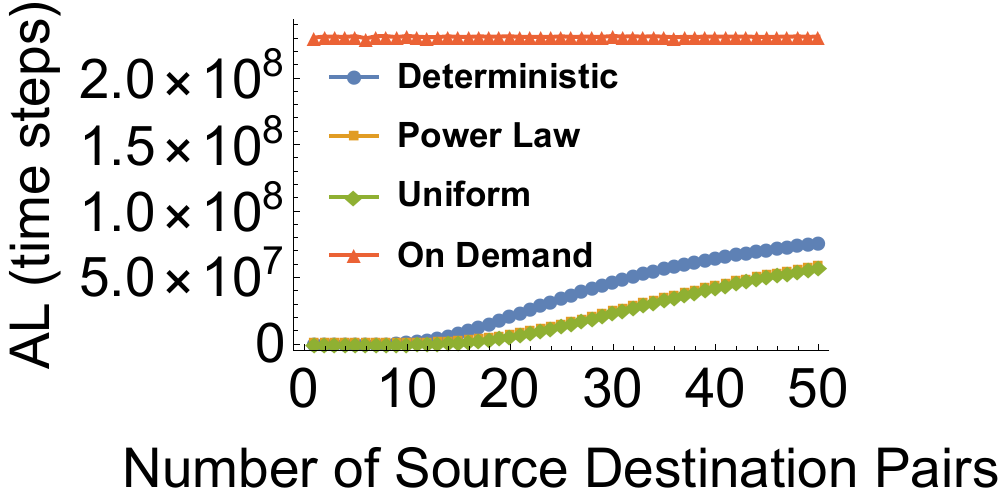}
    \caption{Ring,Modified Greedy, $d_{\ths}=2$}
 \label{ringcap4dem2a}
   \end{subfigure}
  \hspace{5cm}
  \begin{subfigure}{4cm}
    \centering
  \includegraphics[width=7cm]{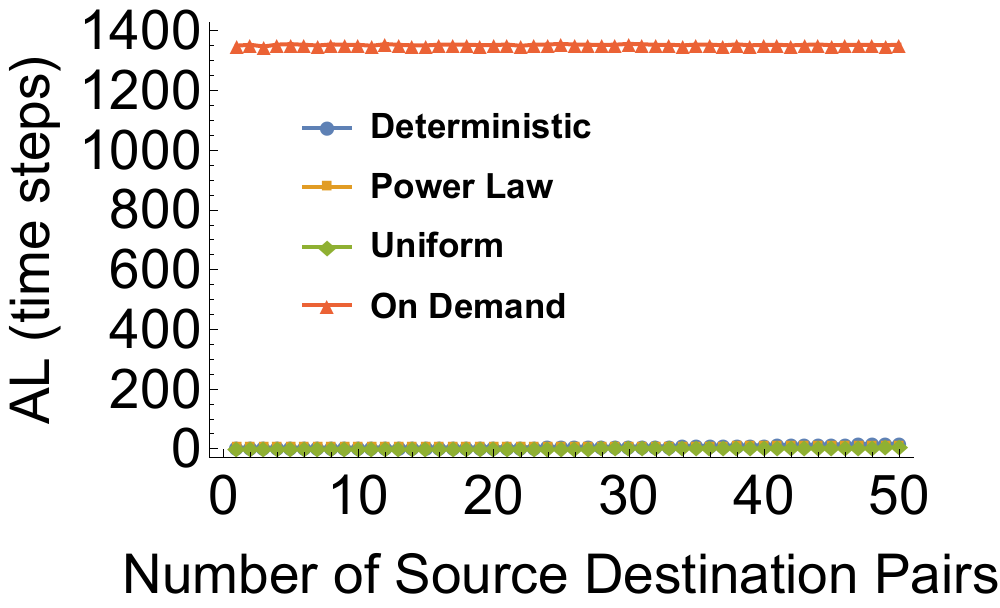}
    \caption{Grid, Modified Greedy, $d_{\ths}=2$}
 \label{gridcap4dem2a}
 \end{subfigure}

  \begin{subfigure}{4cm}
    \centering
    \includegraphics[width=7cm]{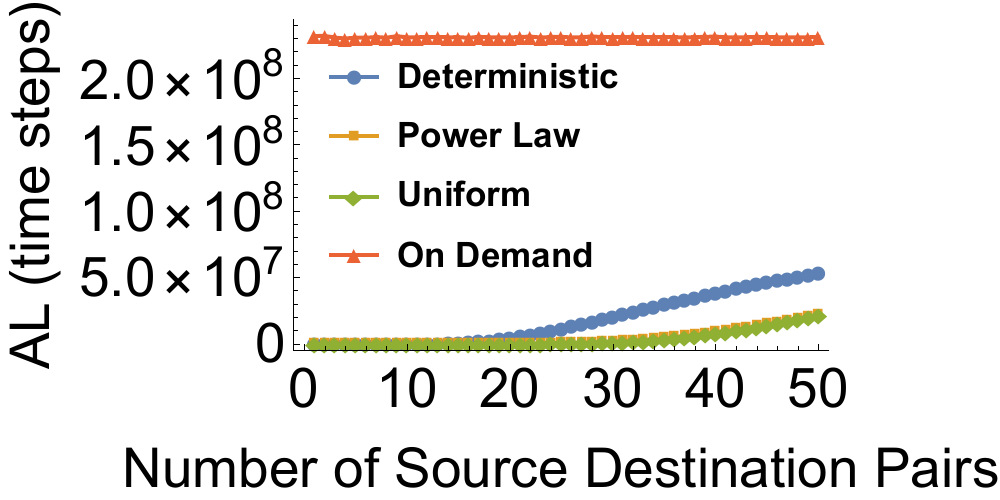}
     \caption{Ring, Modified Greedy, $d_{\ths}=4$}
 \label{ringcap4dem2b}
 \end{subfigure}
  \hspace{5cm}
\begin{subfigure}{4cm}
    \centering
   \includegraphics[width=7cm]{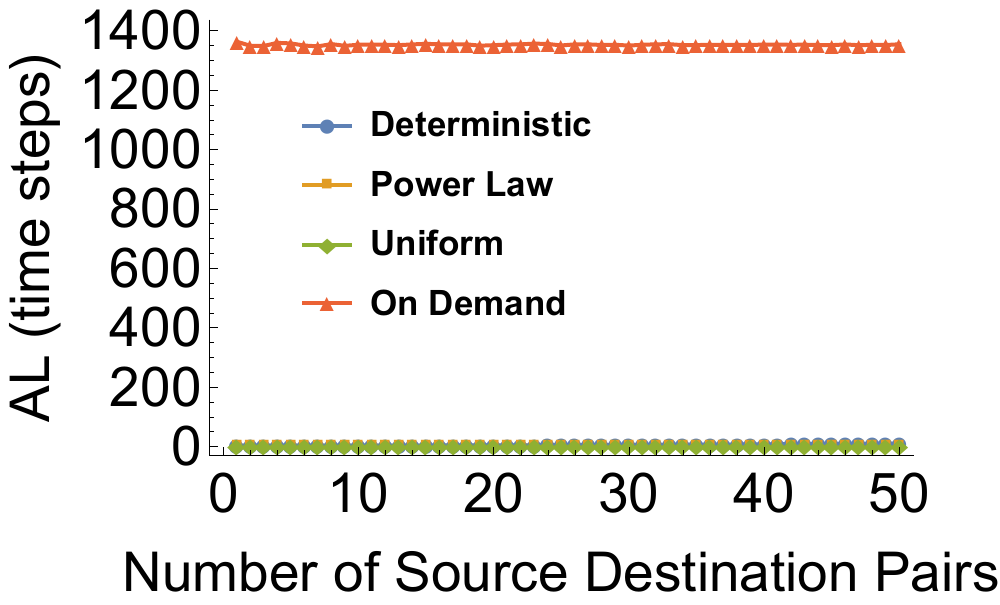}
      \caption{Grid, Modified Greedy, $d_{\ths}=4$}
 \label{gridcap4dem2b}
  \end{subfigure}

  \begin{subfigure}{4cm}
    \centering
    \includegraphics[width=7cm]{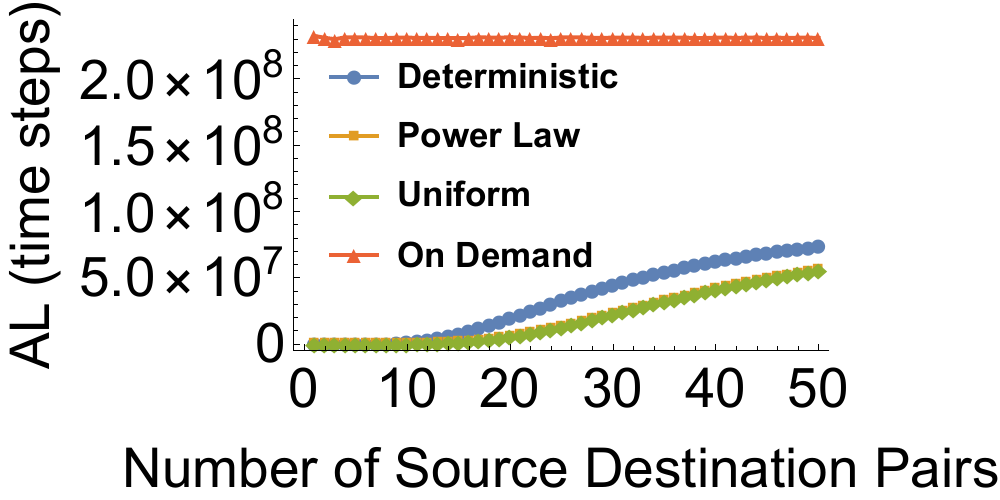}
    \caption{Ring, Local Best Effort, $d_{\ths}=2$}
 \label{ringcap4dem2c}
 \end{subfigure}
  \hspace{5cm}
  \begin{subfigure}{4cm}
    \centering
    \includegraphics[width=7cm]{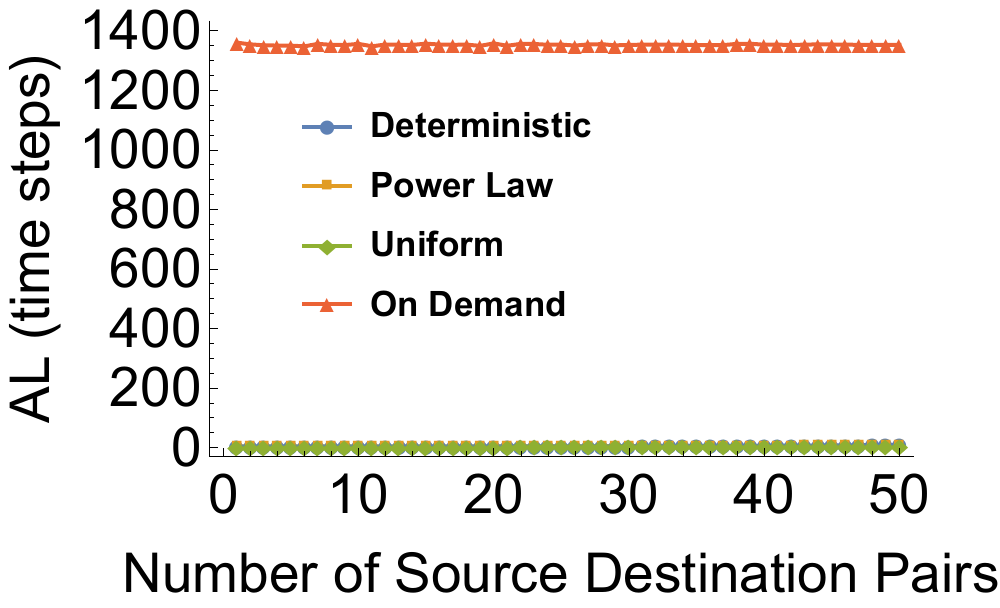}
    \caption{Grid, Local Best Effort, $d_{\ths}=2$}
 \label{gridcap4dem2c}
 \end{subfigure}

 \begin{subfigure}{4cm}
    \centering
    \includegraphics[width=7cm]{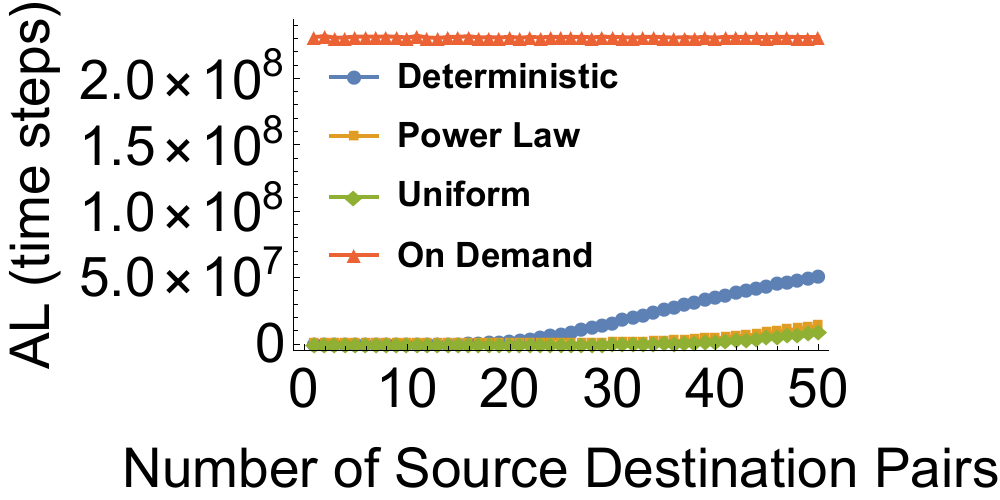}
     \caption{Ring, Local Best Effort, $d_{\ths}=4$}
 \label{ringcap4dem2d}
  \end{subfigure}
  \hspace{5cm}
 \begin{subfigure}{4cm}
    \centering
    \includegraphics[width=7cm]{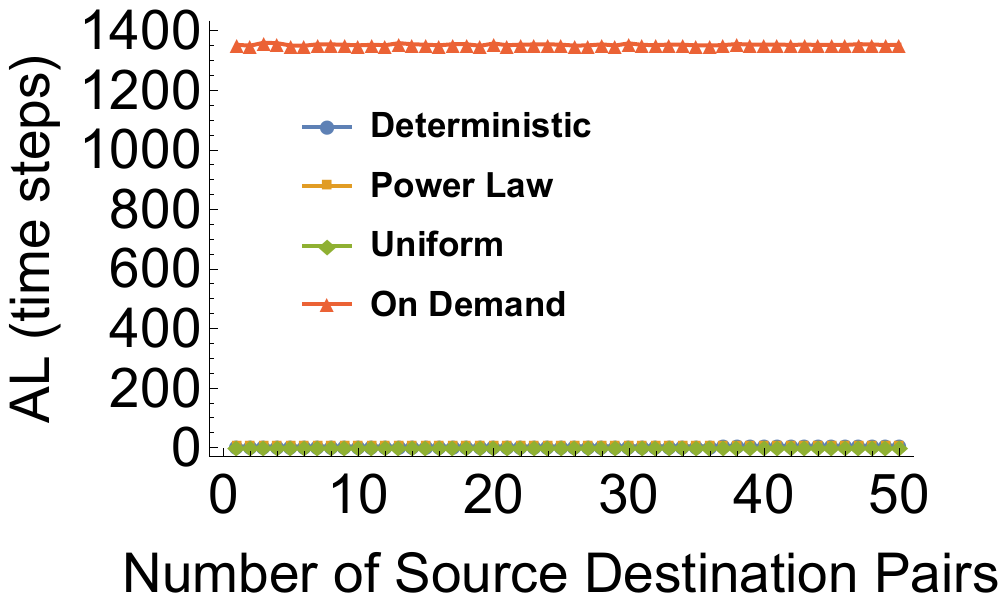}
     \caption{Grid, Local Best Effort, $d_{\ths}=4$}
 \label{compgridcap4dem2b}
 \end{subfigure}

\caption{Performance of the routing algorithms for different types of virtual graphs with $D_{i,j}\leq 2$ and $cap=4$.}
  \label{cap4dem2}
\end{figure*}


\begin{figure*}[tp]

 \begin{subfigure}{4cm}
    \centering
    \includegraphics[width=7cm]{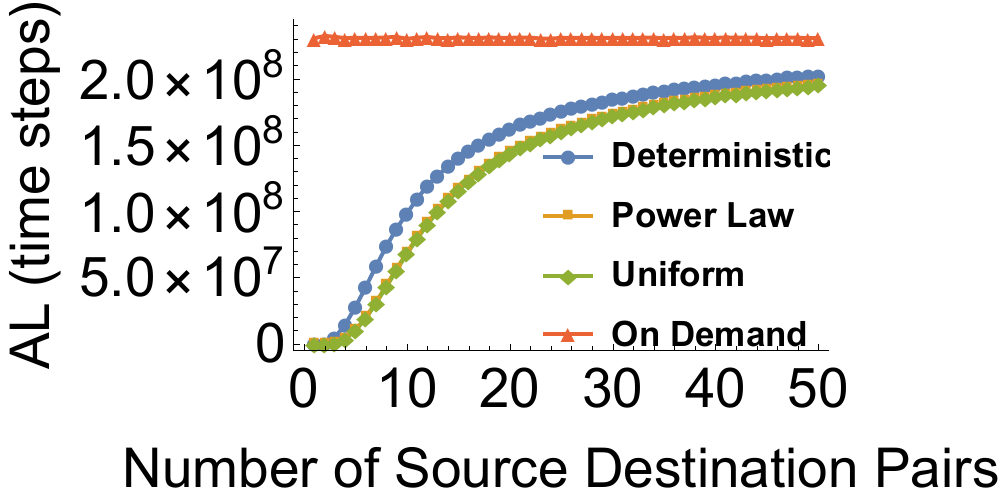}
    \caption{Ring, Modified Greedy, $d_{\ths}=2$}
 \label{ringcap4dem4a}
  \end{subfigure}
  \hspace{5cm}
 \begin{subfigure}{3.4cm}
    \centering
  \includegraphics[width=7cm]{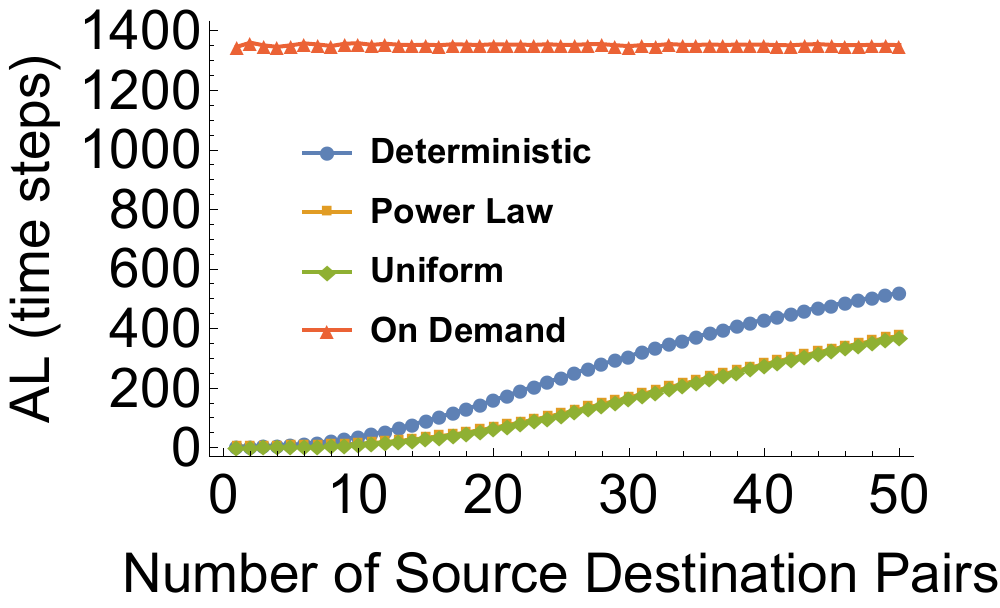}
    \caption{Grid, Modified Greedy, $d_{\ths}=2$}
 \label{gridcap4dem4a}
   \end{subfigure}
  \hspace{0.7cm}

 \begin{subfigure}{4cm}
    \centering
    \includegraphics[width=7cm]{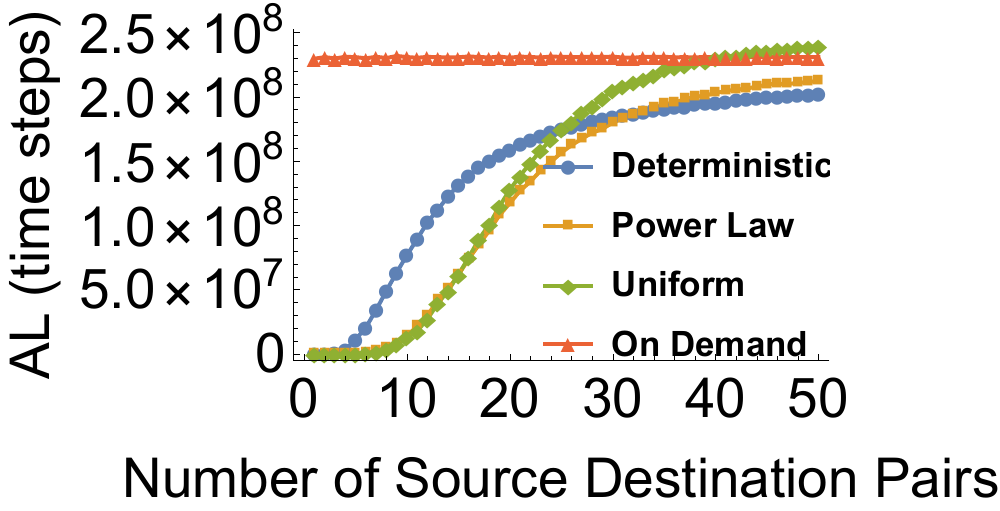}
     \caption{Ring, Modified Greedy, $d_{\ths}=4$}
 \label{ringcap4dem4b}
  \end{subfigure}
  \hspace{5cm}
 \begin{subfigure}{4cm}
    \centering
   \includegraphics[width=7cm]{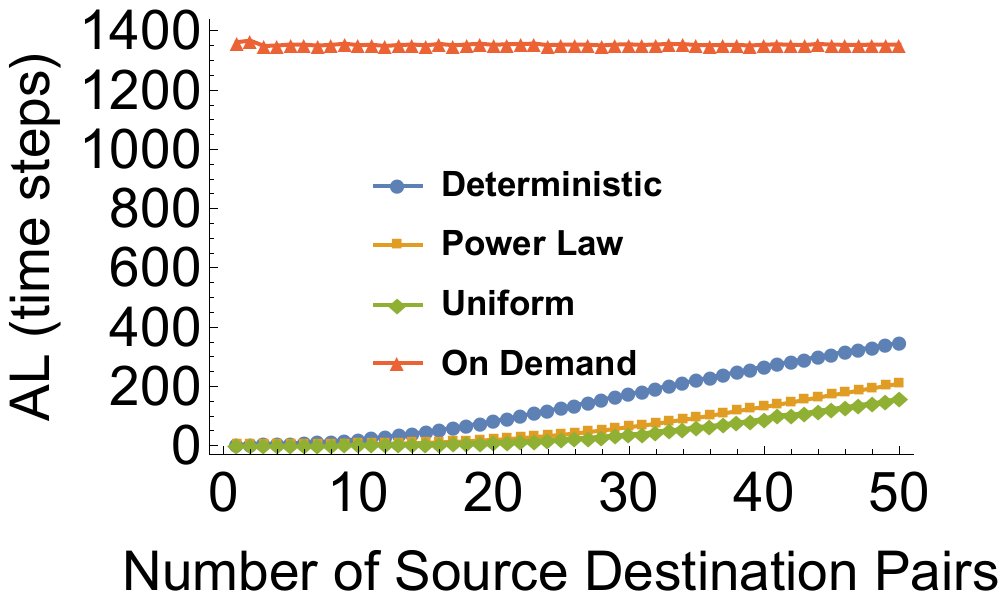}
      \caption{Grid, Modified Greedy, $d_{\ths}=4$}
 \label{gridcap4dem4b}
 \end{subfigure}

  \begin{subfigure}{4cm}
    \centering
    \includegraphics[width=7cm]{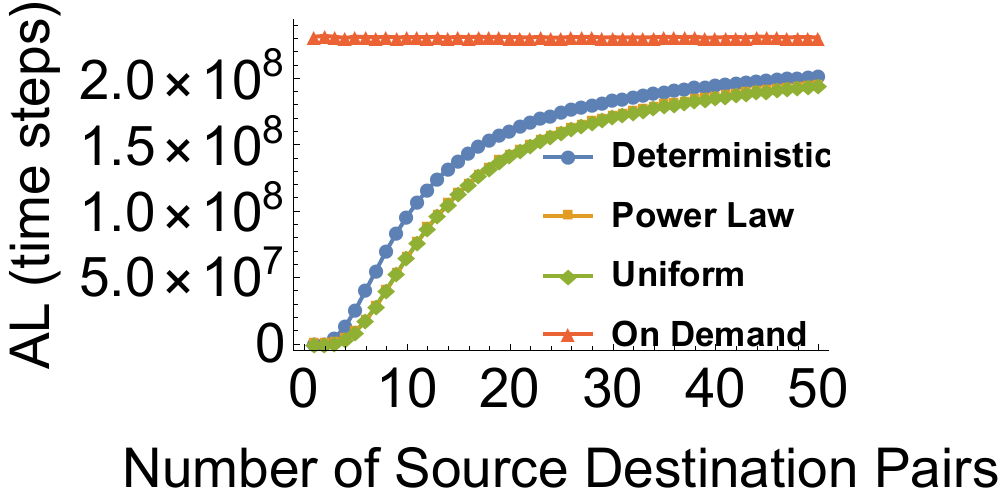}
    \subcaption{Ring, Local Best Effort, $d_{\ths}=2$}
 \label{ringcap4dem4c}
  \end{subfigure}
  \hspace{5cm}
   \begin{subfigure}{4cm}
    \centering
    \includegraphics[width=7cm]{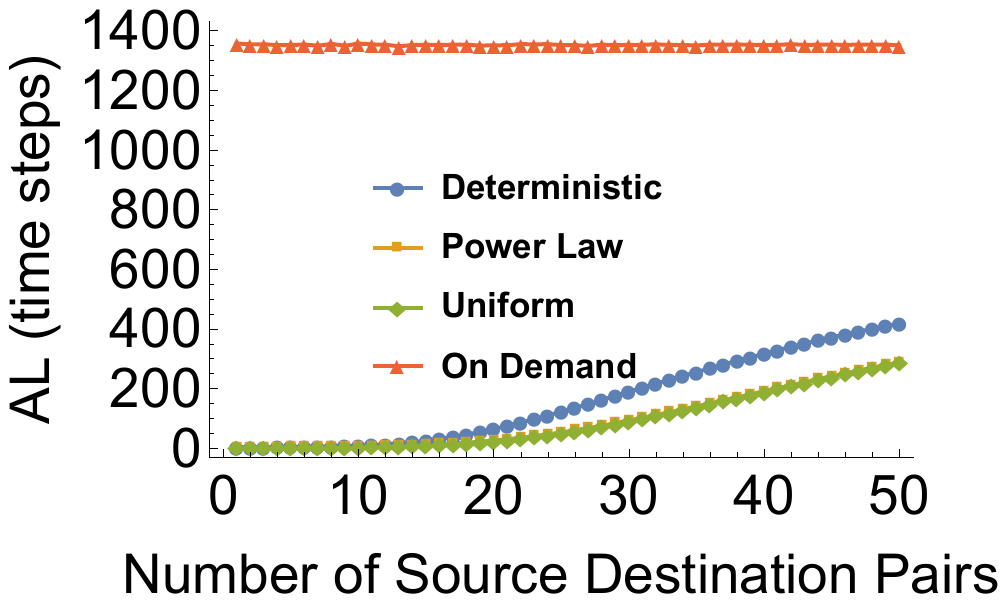}
    \caption{Grid, Local Best Effort, $d_{\ths}=2$}
 \label{gridcap4dem4c}
  \end{subfigure}
  \hspace{0.7cm}

   \begin{subfigure}{4cm}
    \centering
    \includegraphics[width=7cm]{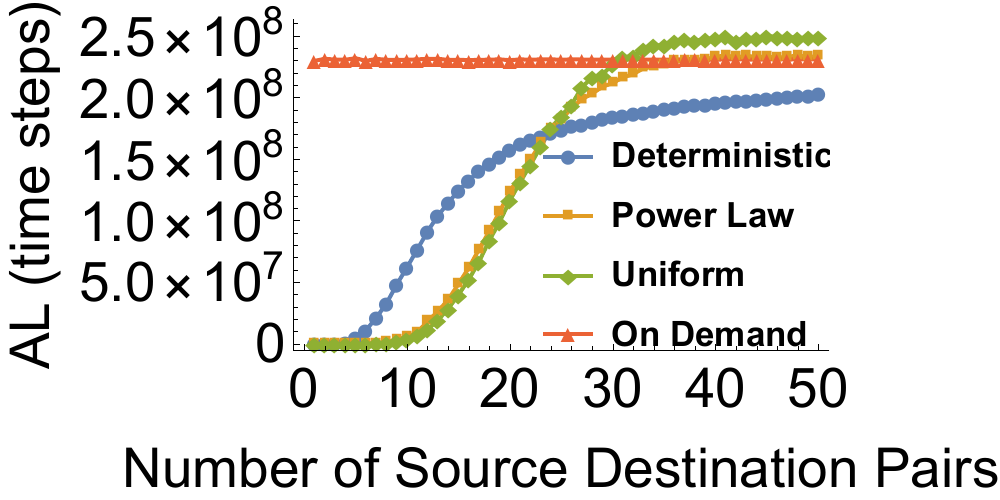}
     \caption{Ring, Local Best Effort, $d_{\ths}=4$}
 \label{ringcap4dem4d}
  \end{subfigure}
  \hspace{5cm}
 \begin{subfigure}{4cm}
    \centering
    \includegraphics[width=7cm]{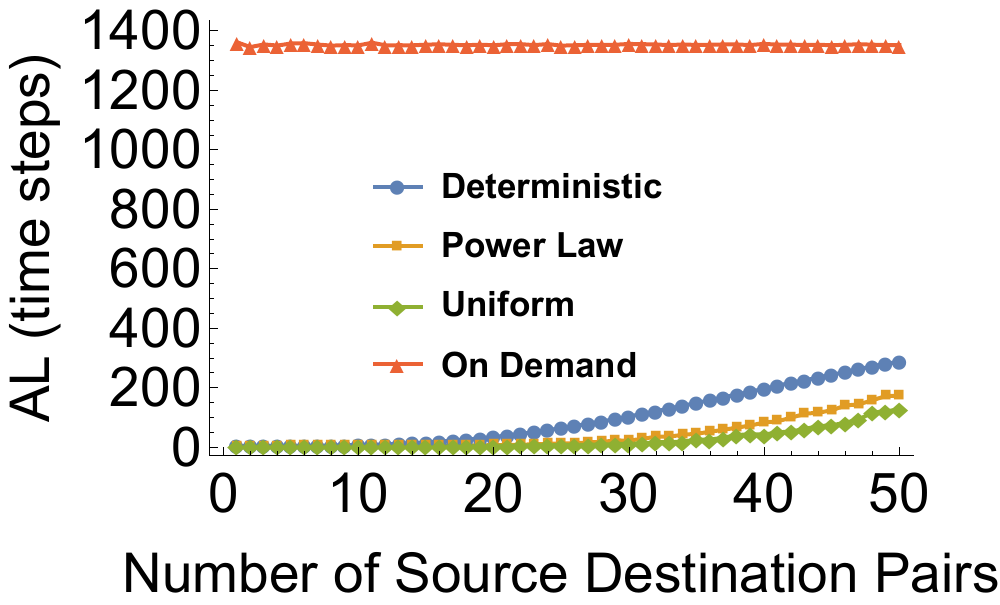}
     
 \caption{Grid, Local Best Effort, $d_{\ths}=4$}
 \label{gridcap4dem4d}
  \end{subfigure}

\caption{Performance of the different types of virtual graphs with $D_{i,j}\geq 2$ and $cap=4$.}

   \label{cap4dem4}
\end{figure*}


\begin{figure*}[tp]
\begin{subfigure}{4cm}
    \centering
    \includegraphics[width=7cm]{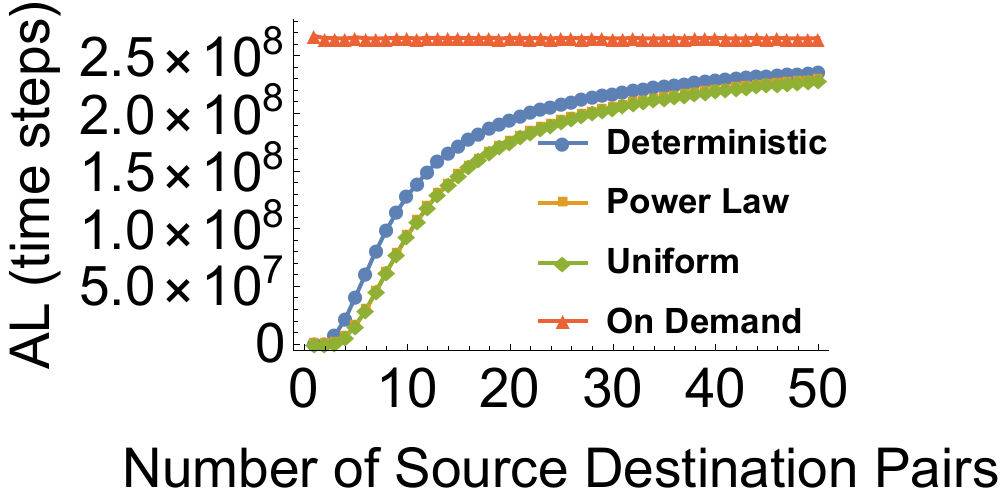}
    \caption{Ring, Modified Greedy, $d_{\ths}=2$}
 \label{ringcap4largedista}
 \end{subfigure}
  \hspace{5cm}
\begin{subfigure}{4cm}
    \centering
    \includegraphics[width=7cm]{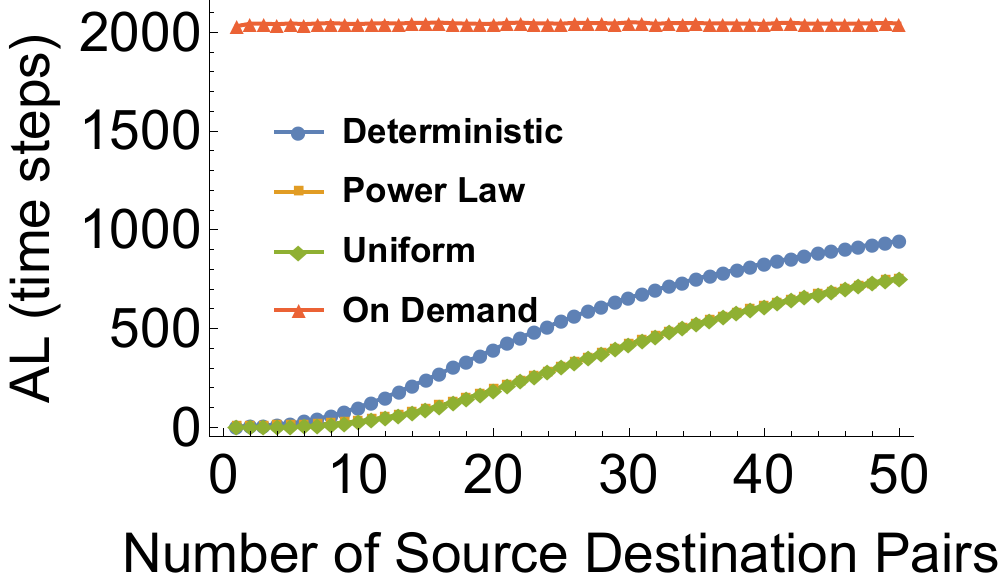}
    \caption{Grid, Modified Greedy, $d_{\ths}=2$}
 \label{gridcap4largeddista}
 \end{subfigure}

  \begin{subfigure}{4cm}
    \centering
   \includegraphics[width=7cm]{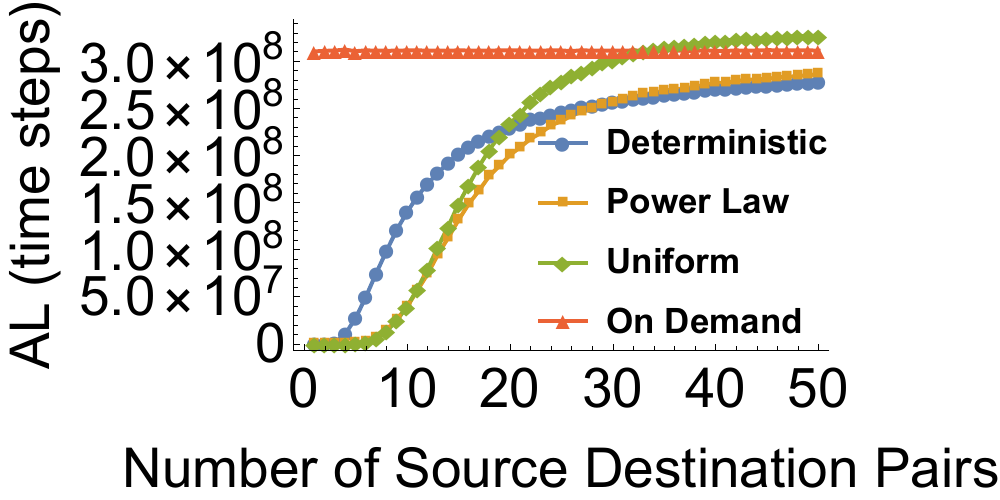}
  \caption{Ring, Modified Greedy, $d_{\ths}=4$}
 \label{ringcap4largedistb}
 \end{subfigure}
 \hspace{5cm}
  \begin{subfigure}{3.4cm}
    \centering
    \includegraphics[width=7cm]{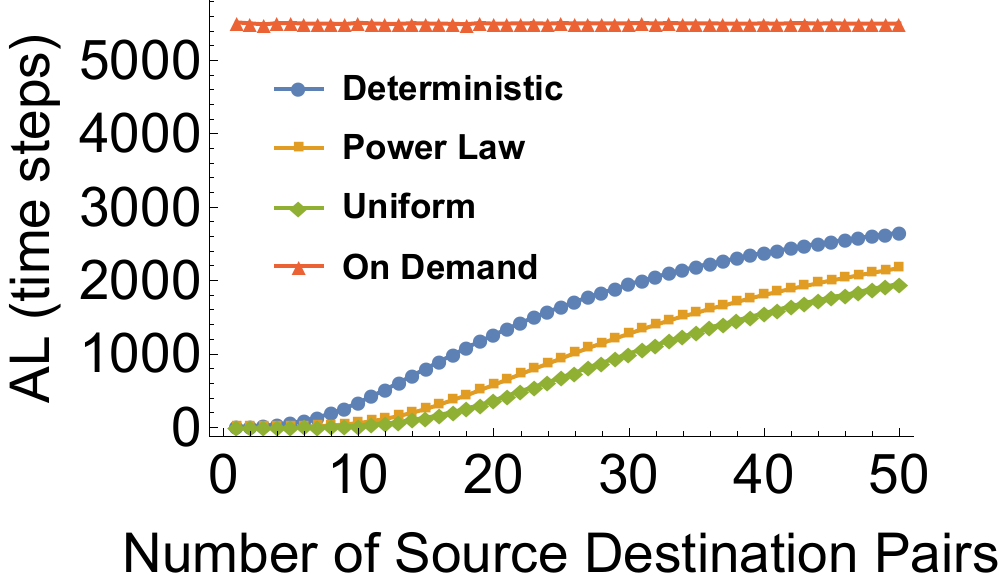}   
 \caption{Grid, Modified Greedy, $d_{\ths}=4$}
 \label{gridcap4largedistb}
  \end{subfigure}

 \begin{subfigure}{4cm}
    \centering
    \includegraphics[width=7cm]{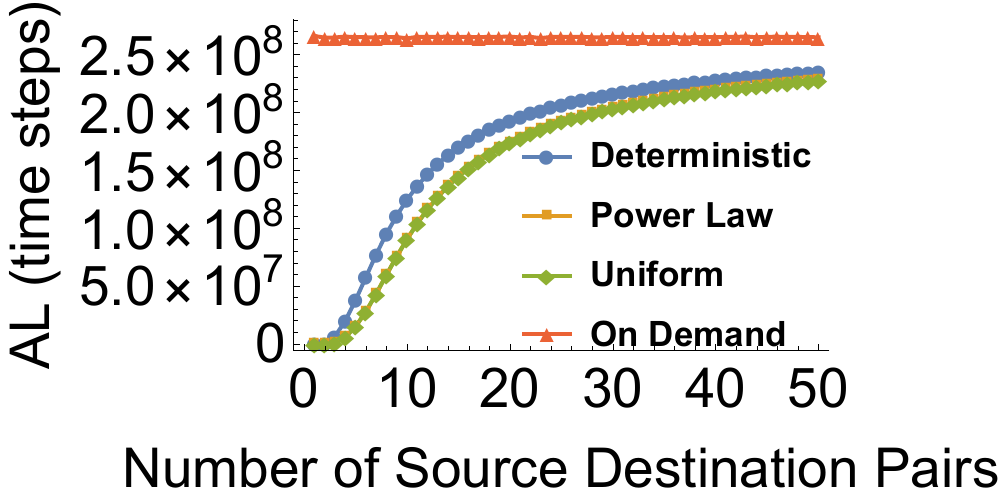}
    \caption{Ring, Local Best Effort, $d_{\ths}=2$}
 \label{ringcap4largedistc}
  \end{subfigure}
    \hspace{5cm}
    \begin{subfigure}{4cm}
    \centering
    \includegraphics[width=7cm]{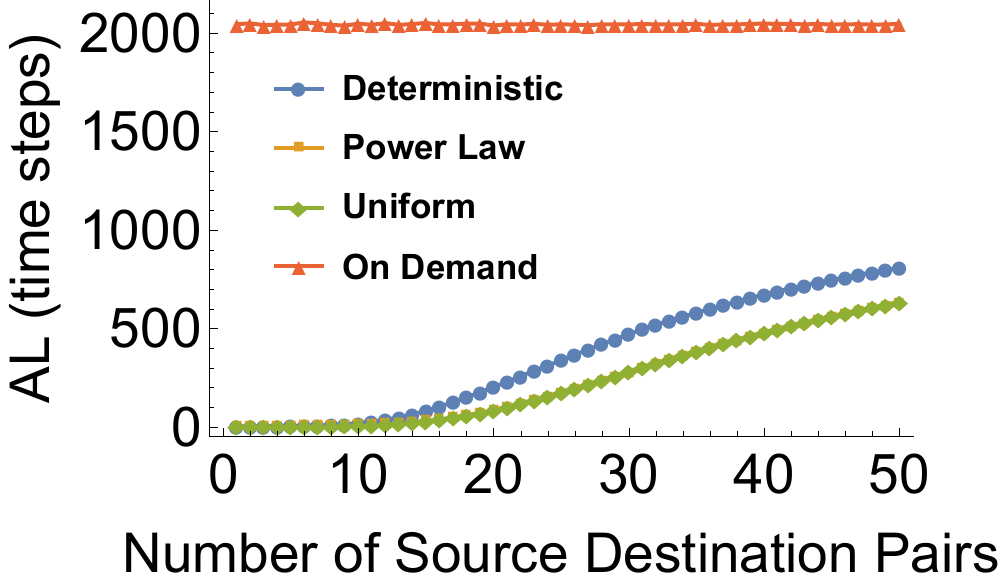}
    \caption{Grid, Local Best Effort, $d_{\ths}=2$}
 \label{gridcap4largedistc}
   \end{subfigure}

 \begin{subfigure}{4cm}
    \centering
    \includegraphics[width=7cm]{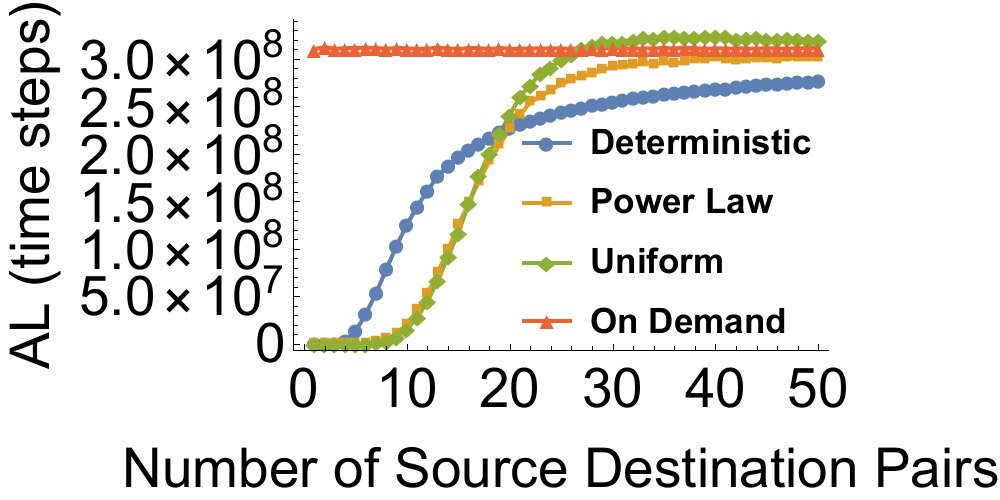}
     \caption{Ring, Local Best Effort, $d_{\ths}=4$}
 \label{ringcap4largedistd}
 \end{subfigure}
   \hspace{5cm}
  \begin{subfigure}{4cm}
    \centering
    \includegraphics[width=7cm]{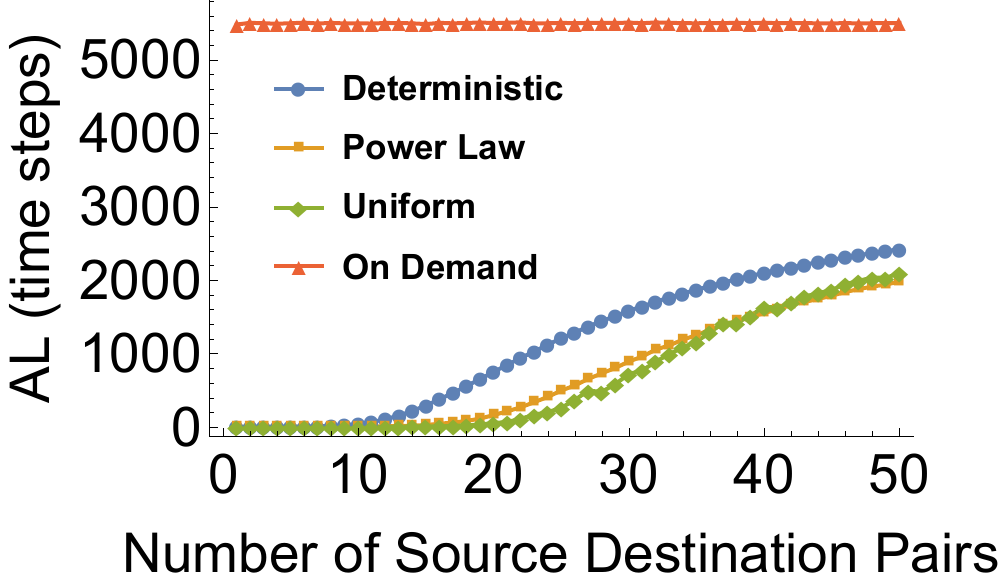}
     \caption{Grid, Local Best Effort, $d_{\ths}=4$}
 \label{gridcap4largeddistd}
 \end{subfigure}
  \caption{Performance of different types of virtual graphs with $D_{i,j}\geq 2$ and $\dist_{G_{\phs}}(i,j) \geq d_{\ths}$ and $cap=4$.}

   \label{cap4largedist}
\end{figure*}




\begin{figure*}[ht]
 \begin{subfigure}{5cm}
    \centering
    \includegraphics[width=7cm]{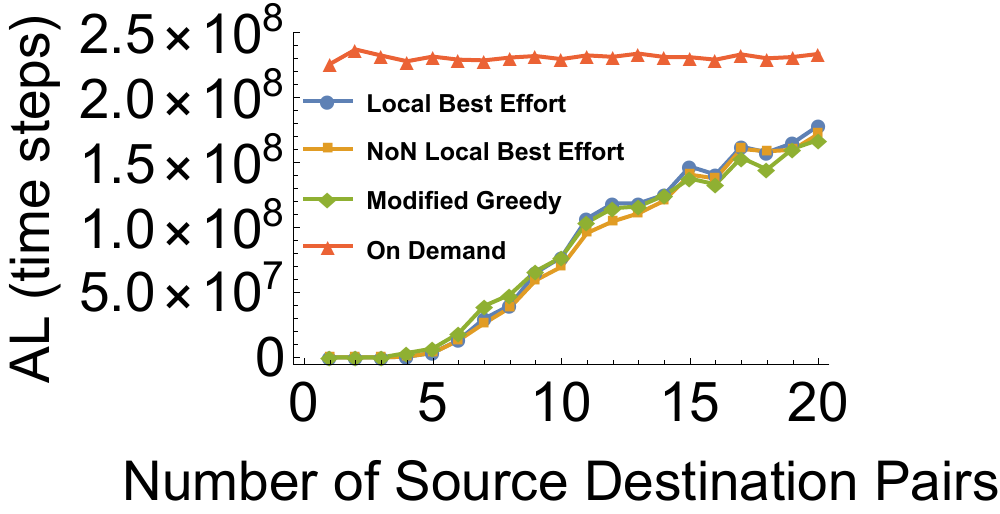}
    \caption{Ring, Deterministic, $d_{\ths}=4$}
 \label{ringcap4nona}
 \end{subfigure}
  \hspace{2cm}
  \begin{subfigure}{5cm}
    \centering
    \includegraphics[width=7cm]{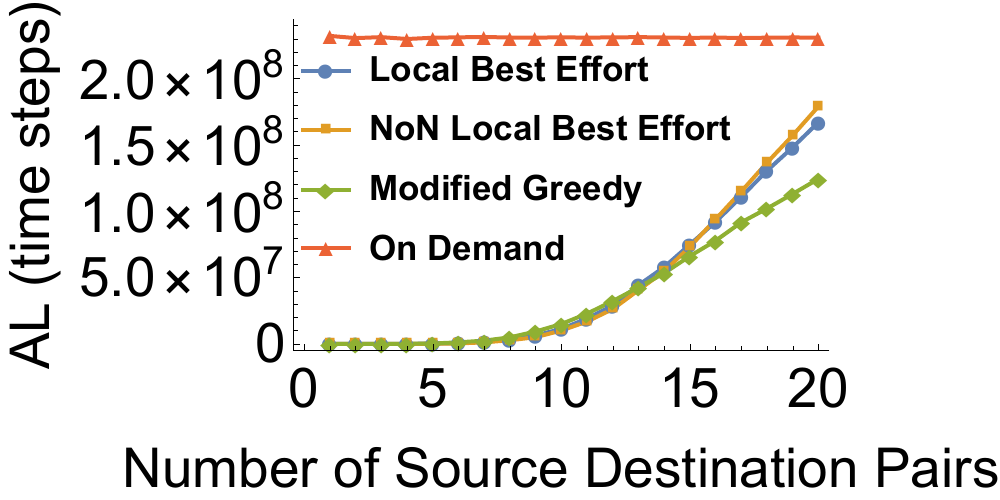}
    \caption{Ring, Power Law, $d_{\ths}=4$}
 \label{ringcap4nonb}
 \end{subfigure}

\begin{subfigure}{5cm}
    \centering
    \includegraphics[width=7cm]{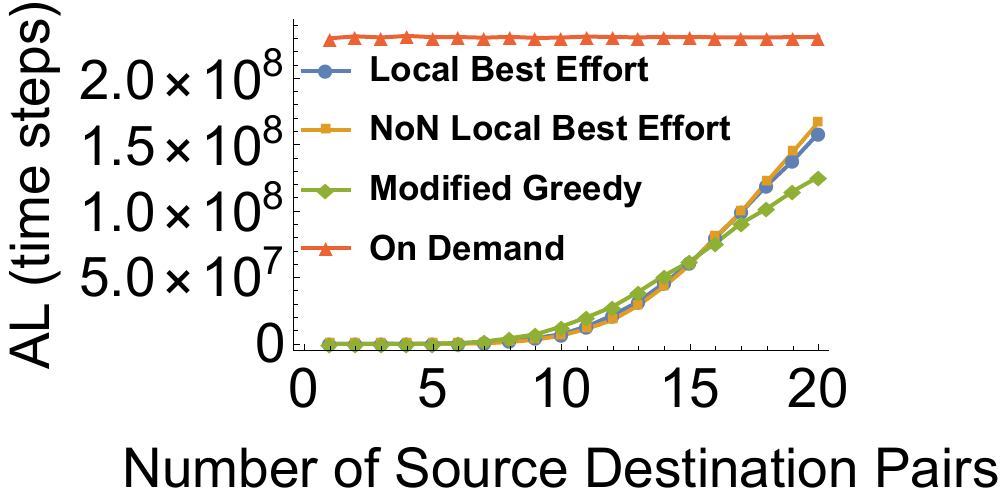}
    \caption{Ring, Uniform, $d_{\ths}=4$}
 \label{ringcap4nonc}
 \end{subfigure}

  \caption{Comparison between NoN best effort algorithm with other routing algorithms for different types of virtual graphs with $D_{i,j}= 1$ and $cap=1$.}

   \label{capnondist}
\end{figure*}

\section{Routing in more Generalised Graphs}
\label{rrgg}

\subsection{Recursively Generated Graphs (RGG)}
This section focuses on a further study of the proposed entanglement pre-sharing and routing models in more general graphs. The topology of a large real-life internet network is far more complex than a ring or grid and it always grows with time. However, in most of the cases, the global pattern of a complex system emerges from the patterns observed in smaller subsystems. These phenomena can also be observed in classical internet. The recursive graph is one of the best ways to model this kind of objects \cite{CFR04,HIK11}. In general, recursive graphs are being constructed from a base graph $G_0$ and at each recursive step, the current graph $G_t$ is being substituted by another graph $G_{t+1}$. The substitution rules remain the same for the entire evolution. Tensor product graphs are one example of such graphs \cite{HIK11}. In this paper, we study a model of the recursively generated graph, where a new graph is being generated by substituting the edges of the old graph. One can find an example of such graphs in \cite{SMIKW16}. Here, we use the term \textit{edge substitution} to describe this operation. It is defined as follows,

\begin{definition}[Edge Substitution of a Graph]
\label{edge_sub}
Let $G_0 = (V_0, E_0)$ and $H = (V_H, E_H)$ be two graphs. The edge substitution of a graph $G_0$ with respect to $H$ is an operation, where each edge $(u,v) \in E_0$ (where $u,v \in V_0$) is being substituted by another graph 
$(\{u,w_{uv}\},\{(u,w_{uv})\}) \cup G_{uv} \cup (\{v,w'_{uv}\},\{(v,w'_{uv})\})$, where $w_{uv},w'_{uv} \in G_{uv}$ and $G_{uv} = (V_{uv}, E_{uv})$ is isomorphic to $H$ and $\dist_{G_{uv}}(w_{uv},w'_{uv}) = \diam_H$. If $G_1 = (V_1, E_1)$ denotes the new substituted graph then,
\begin{align}
\label{edge_sub_graph}
V_1 &= V_0 \bigcup_{(u,v) \in E_0} V_{uv}\\ \nonumber
E_1 & = (E_0 \setminus (u,v)) \bigcup_{(u,v)\in E_0}[(u,w_{uv}) \cup E_{uv} \cup (v,w'_{uv})]. 
\end{align} 
 
\end{definition}

we construct the recursive graphs using following rules

\begin{enumerate}
\item The base graph $G_0= (V_0,E_0)$ represents quantum internet graph with physical links.
\item Let at recursive step $l$ the physical graph is $G_l = (V_l,E_l)$.
\item Suppose $\p$ be a set of graphs. If any subgraph $\tilde{G}_l \subseteq G_l$ is isomorphic to any graph in $\p$ then we perform edge substitution of $\tilde{G}_l$ with respect to $H = (V_H,E_H)$. Suppose $\tilde{G}_{l+1}= (\tilde{V}_{l+1},\tilde{E}_{l+1})$ denotes the substituted graph.
\item At recursive step $l+1$, if the physical graph with physical links is $G_{l+1} = (V_{l+1},E_{l+1})$ then
\begin{align*}
V_{l+1} &= V_l \bigcup_{\substack{\tilde{G}_l \subseteq G_l\setminus G_{l-1}\\ \tilde{G}_l \cong \p}} \tilde{V}_{l+1}.\\
E_{l+1} &= (E_{l}\setminus (\bigcup_{\substack{\tilde{G}_l \subseteq G_l\setminus G_{l-1}\\ \tilde{G}_l \cong \p}}\tilde{E}_l)) \bigcup_{\substack{\tilde{G}_l \subseteq G_l\setminus G_{l-1}\\ \tilde{G}_l \cong \p}} \tilde{E}_{l+1}.
\end{align*}
\end{enumerate}


\begin{definition}[Regular Recursively Generated Graphs (RRGG)]
A RGG with respect to an initial graph $G_0=(V_0,E_0)$, a set of graphs $\p$, a substitution graph $H = (V_H,E_H)$ is called regular if $G_0 = H$ is a regular graph and $\p = \{H\}$.
\end{definition}
\begin{figure*}[ht!]
\centering
 \begin{subfigure}{5cm}
    \centering
\includegraphics[width=4cm]{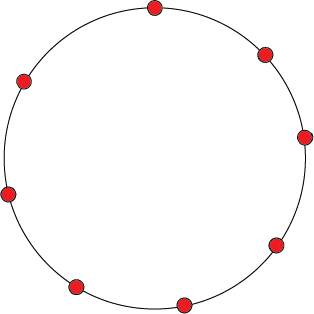}
\caption{RRGG with $G_0 = H = C_8$ and $\p = \{C_8\}$, at the $0$-th level of the recursive step.}
\label{rrgg_0}
\end{subfigure}
 \hspace{2cm}
 \begin{subfigure}{5cm}
\centering
\includegraphics[width=4cm]{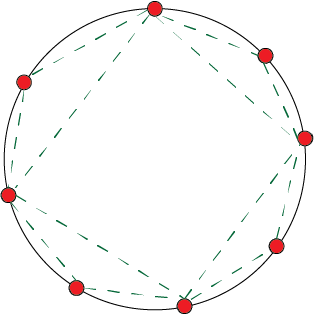}
\caption{Recursively generated deterministic virtual graph $\G_0$ with $d_{\ths}=2$ corresponding to the physical graph $G_0$ shown in figure \ref{rrgg_0}. Here green dotted lines are pre-shared entangled state.}
\label{rrgg_vir}
\end{subfigure}
\caption{Recursively generated physical graph and virtual graph at the $0$-th level of the recursive step.}
\label{rrgg_pv0}
\end{figure*}

\begin{figure*}[h!]
\centering
 \begin{subfigure}{5cm}
\centering
 \includegraphics[width=6cm]{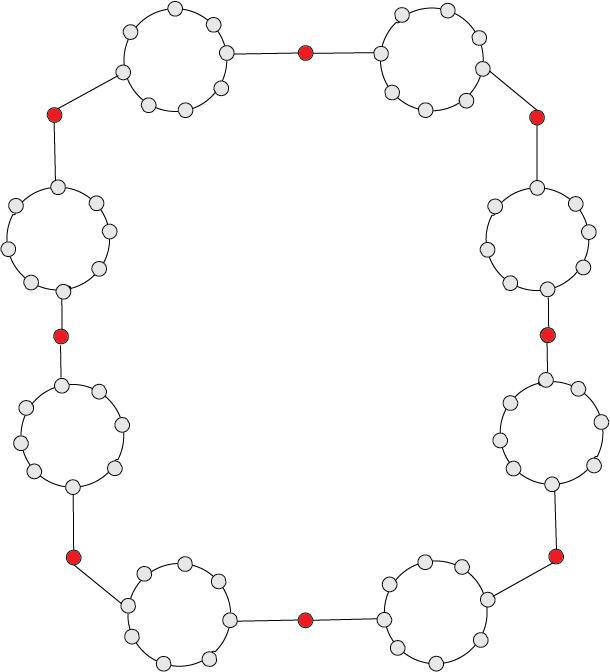}
\caption{RRGG with $G_0 =(V_0,E_0) = H = C_8$ and $\p = \{C_8\}$, at the $1$-st level of the recursive step. Here, red nodes correspond to the nodes from the graph $G_0$, in figure \ref{rrgg_pv0}. }
\label{rrgg_1}
\end{subfigure}
\hspace{2cm}
 \begin{subfigure}{5cm}
\centering
\includegraphics[width=6cm]{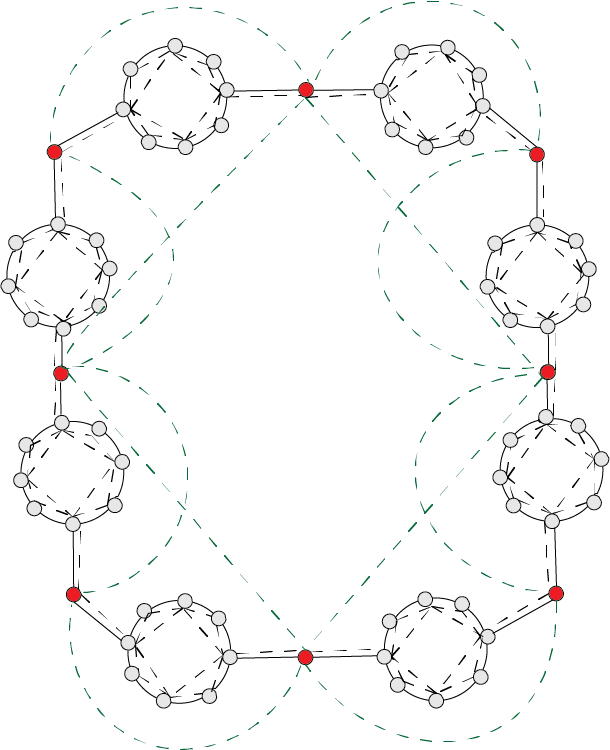}
\caption{Recursively generated deterministic virtual graph $\G_1$ with $d_{\ths}=2$ corresponding to the physical graph $G_0$ shown in figure \ref{rrgg_1}. Here green dotted lines denote the pre-shared entangled links from $\G_0$ and black dotted lines denote the newly generated pre-shared entangled links.}
\label{rrgg_vir1}
\end{subfigure}
\caption{Recursively generated physical graph and virtual graph at the $1$-st level of the recursive step.}
\label{rrgg_pv1}
\end{figure*}

In this paper we focus on studying RRGG with $G_0 = H = C_n$ and $\p = \{C_n\}$.

\subsection{Recursively Generated Virtual Graph}

The main motivation of studying recursive graph is to understand the behaviour of a hierarchical topology for a quantum network. In this paper we are also interested in studying the quantum internet if the nodes from different levels of the hierarchy have different entanglement generation capabilities. We assume that for any RRGG $G_l = (V_l,E_l)$ if any node $v \in V_{l'}\setminus V_{l'-1}$, for some $0\leq l' \leq l-1$, then it can pre-share entangled link with any node within $d_{\ths}(\diam_{G_{l-l'-1}}+2)$ distance from itself. The nodes $v \in G_l\setminus G_{l-1}$ can create entangled links within $d_{\ths}$ distance from itself.

The detailed construction of the virtual graph $\G_l = (V_l,\E_l)$ is given below. 
\begin{enumerate}
\item Let $G_0 = (V_0,E_0)$ be a physical graph, which is isomorphic to $C_n$. We construct the corresponding virtual graph $\G_0 = (V_0,\E_0)$ using the same procedure, proposed in section \ref{qig_ring}.
\item Let $G_{l-1} = (V_{l-1},E_{l-1})$ denote the physical graph at $(l-1)$-th level of recursion.
\item At $l$-th level, if any subgraph $\tilde{G}_{l-1}$ of $G_{l-1}$ is isomorphic to $C_n$ then for all of the edges of $\tilde{G}_{l-1}$ we perform edge substitution with respect to the graph $\G_0$. Let $\tilde{\G}_{l} = (V_{l},\tilde{\E}_{l})$ denotes the substituted graph. 
\item This implies, if at $l$-th level the virtual graph is denoted by $\G_l = (V_l,\E_l)$ then
\begin{align*}
\E_{l} &= \E_{l-1} \bigcup_{\substack{\tilde{G}_{l-1} \subseteq G_{l-1}\setminus G_{l-2}\\ \tilde{G}_{l-1} \cong \p}}\tilde{\E}_{l}.
\end{align*}
 \end{enumerate}
 
\textbf{Example :} In figure \ref{rrgg_pv0} and \ref{rrgg_pv1} we give examples of the RRGG with $G_0 = H = C_8$ and $\p=\{C_8\}$. The graph corresponds to the figure \ref{rrgg_0} denote the physical graph at the recursive step $0$. From there we construct the deterministic virtual graph $\G_0$ with $d_{\ths}=2$ using the techniques proposed in section \ref{qig_ring}. Later in figure \ref{rrgg_1} we construct the physical graph by performing edge substitution corresponding to each of the edges in $G_0$. For a better understanding, in figure \ref{rrgg_1} we colour the nodes at $G_0$ red. One can note that, as a result of the edge substitution, between two red nodes there is one $C_8$. The virtual graph $\G_1$ in figure \ref{rrgg_vir1} is being constructed from $\G_0$ and $G_1$. In figure \ref{rrgg_vir1} all of the edges of $G_0$ is being substituted by a graph isomorphic to $\G_0$. Here also we put red colour on the nodes from $\G_0$. Beside this, in $\G_1$, we keep the virtual links between the red nodes from $\G_0$. In figure \ref{rrgg_vir1} we denote these links by green dotted lines.

\section{Properties of the RRGG}
Here we discuss the properties of this type of network which are useful to get an upper bound on the latency for distributing an entangled link between a source and destination pair using all of the routing algorithms, described in the section \ref{routing}.
In the previous section we assume that for any RRGG $G_l = (V_l,E_l)$ if any node $v \in V_{l'}\setminus V_{l'-1}$, for some $0\leq l' \leq l-1$, then it can pre-share entangled link with any node within $d_{\ths}(\diam_{G_{l-l'-1}}+2)$ distance from itself. In the next lemma we show that our assumption holds for the proposed recursively generated virtual graphs. 

\begin{lemma}
\label{dist_node}
In the recursively generated virtual graph $\G_l = (V_l,\E_l)$, generated from a RRGG $G_l = (V_l,E_l)$, if for any two nodes $u,v \in V_l$, $\dist_{\G_l}(u,v)=1$ then,
\begin{align*}
\dist_{G_l}(u,v) &\leq d_{\ths}(\diam_{G_{l-l'-1}}+2) ~\text{if }  \\
& \forall 1\leq l' \leq l-1,  v \in V_{l'}\setminus V_{l'-1} \\
\dist_{G_l}(u,v) & \leq d_{\ths}~~\text{if } v \in V_{l}\setminus V_{l-1}.
\end{align*} 
\end{lemma}

\begin{proof}
 For the case $l'=0$ we have $v \in V_l\setminus V_{l-1}$. All such nodes are being generated after the edge substitution operation with respect to $\G_0$. This implies that all the virtual neighbours of such nodes $v$ are within $d_{\ths}$ perimeter. 
 
In case of $l' > 0$ all the nodes $v \in V_{l-l'-1}\setminus V_{l-l'-2}$ are being generated after the edge substitution operation with respect to $\G_0$. This implies, if any two nodes $u,v \in V_{l-l'-1}\setminus V_{l-l'-2}$ have distance $\dist_{G_{l-l'-1}}(u,v)$, then after the edge substitution their distance would become $\dist_{G_{l-l'}}(u,v)=\dist_{G_{l-l'-1}}(\diam_{G_0} +2)$ and $\dist_{G_{l}}(u,v)=\dist_{G_{l-l'-1}}(\diam_{G_{l-l'-1}} +2)$. If two nodes $u,v \in V_{l-l'-1}\setminus V_{l-l'-2}$ are neighbours in $\G_{l-l'-1}$ then $\dist_{G_{l-l'-1}}(u,v) \leq d_{\ths}$. According to the construction, those two nodes would still be neighbours in $\G_l$. This implies their distances in $G_l$ would be upper bounded by $d_{\ths}(\diam_{G_{l-l'-1}}+2)$.

\end{proof}

According to lemma \ref{routing_time}, in the continuous model, if a single source and destination pair would like to generate a single entangled link then the latency is the order of the the diameter of the physical graph. In the next lemma we are interested in computing the diameter of a RRGG at the $l$-th step of the recursion. 
\begin{lemma}
\label{diam_rrgg}
For any RRGG $G_l = (V_l,E_l)$ where $l \geq 1$ and $G_0 = (V_0,E_0) = C_n$ is the base graph and $\p = \{C_n\}$, 
\begin{equation}
\diam_{G_l} = \diam_{G_{0}}(\diam_{G_{l-1}}+2).
\end{equation}
\end{lemma}
\begin{proof}
We prove this lemma using induction. For the base case, $l=1$ we construct $G_1$ by performing the edge substitution operations with respect to $C_n$ on all the edges of $G_0$. Each of the edge substitution operations increase the distance between two neighbour nodes in $G_0$ by $\diam_{G_0}+2$. For any two nodes $u,v \in V_1$ there are three following cases.

\begin{itemize}
\item \textbf{Case 1 : If $u,v \in V_0$ }then $\dist_{G_{0}}(u,v) \leq \diam_{G_0}$. This implies, $\dist_{G_{1}}(u,v) \leq \diam_{G_0}(\diam_{G_{0}}+2)$. The equality holds if $\dist_{G_0}(u,v) = \diam_{G_0}$. 

\item \textbf{Case 2 : If $u,v \in V_1\setminus V_0$ }then they are part of $G_1\setminus G_0$. Suppose $u$ has been generated after the edge substitution of $(u',u'') \in E_0$ and $v$ has been generated after the edge substitution of $(v',v'') \in E_0$. Without any loss of generality we can assume that $\dist_{G_0}(u',v') = \min\{\dist_{G_0}(u',v'), \dist_{G_0}(u',v''),$ $\dist_{G_0}(u'',v'),\dist_{G_0}(u'',v'')\}.$ This implies, $\dist_{G_0}(u'',v'') = \dist_{G_0}(u',v') + 2$ and $\dist_{G_0}(u',v') \leq \diam_{G_0}-2$. Hence, $\dist_{G_1}(u,v) = \dist_{G_1}(u,u') + \dist_{G_1}(v,v') + \dist_{G_1}(u',v') \leq 2\diam_{G_0} + \dist_{G_0}(u',v')(\diam_{G_0}+2) \leq \diam_{G_0}(\diam_{G_{0}}+2)$.

\item \textbf{Case 3 : If $u \in V_0$ and $v \in V_1 \setminus V_0$ }then let $v$ has been generated after the edge substitution of the edge $(v',v'') \in E_0$. Without any loss of generality we can assume that $\dist_{G_0}(u,v') = \min\{\dist_{G_0}(u,v') , \dist_{G_0}(u,v'')\}$. This implies $\dist_{G_0}(u,v'') = \dist_{G_0}(u,v') + 1$ and $\dist_{G_0}(u,v') \leq \diam_{G_0}-1$. Hence, $\dist_{G_1}(u,v) = \dist_{G_1}(v,v') + \dist_{G_1}(u,v') \leq \diam_{G_0} + \dist_{G_0}(u,v')(\diam_{G_0}+2) \leq \diam_{G_0}(\diam_{G_{0}}+2)$.
\end{itemize}

The base case of the inductive hypothesis is true. Suppose the inductive hypothesis is true for some $l=l'$. This implies $\diam_{G_{l'}} = \diam_{G_0}(\diam_{G_{l'-1}}+2)$. For $l=l'+1$, the graph is $G_{l'+1}$. It has been generated by performing the edge substitution operations on all of the edges of the subgraphs $H_{l'} \subseteq G_{l'}$ which are isomorphic to $C_n$. One can also construct $G_{l'+1}$ by replacing each edge of $G_0$ by a $G_{l'}$. For any two nodes in $G_{l'+1}$ we have following three cases.

\begin{itemize}
\item \textbf{Case 1 : If $u,v \in V_{l'}$ }then $\dist_{G_{l'}}(u,v) \leq \diam_{G_{l'}}$. This implies, $\dist_{G_{l'+1}}(u,v) \leq \diam_{G_0}(\diam_{G_{l'}}+2)$. The equality holds if $\dist_{G_{l'}}(u,v) = \diam_{G_{l'}}$. 

\item \textbf{Case 2 : If $u,v \in V_{l'+1}\setminus V_{l'}$ }then they are part of $G_{l'+1}\setminus G_{l'}$. Suppose $u$ has been generated after the edge substitution of $(u',u'') \in E_{l'}$ and $v$ has been generated after the edge substitution of $(v',v'') \in E_{l'}$. Without any loss of generality we can assume that $\dist_{G_{l'}}(u',v') = \min\{\dist_{G_{l'}}(u',v'), \dist_{G_{l'}}(u',v''),$ $\dist_{G_{l'}}(u'',v'),\dist_{G_{l'}}(u'',v'')\}.$ This implies, $\dist_{G_{l'}}(u'',v'') = \dist_{G_{l'}}(u',v') + 2$ and $\dist_{G_{l'}}(u',v') \leq \diam_{G_{l'}}-2$. Hence, $\dist_{G_{l'+1}}(u,v) = \dist_{G_{l'+1}}(u,u') + \dist_{G_{l'+1}}(v,v') + \dist_{G_{l'+1}}(u',v') \leq 2\diam_{G_0} + \dist_{G_{l'}}(u',v')(\diam_{G_0}+2) \leq \diam_{G_0}(\diam_{G_{l'}}+2)$.

\item \textbf{Case 3 : If $u \in V_{l'}$ and $v \in V_{l'+1} \setminus V_{l'}$ }then let $v$ has been generated after the edge substitution of the edge $(v',v'') \in E_{l'}$. Without any loss of generality we can assume that $\dist_{G_{l'}}(u,v') = \min\{\dist_{G_{l'}}(u,v') , \dist_{G_{l'}}(u,v'')\}$. This implies $\dist_{G_{l'}}(u,v'') = \dist_{G_{l'}}(u,v') + 1$ and $\dist_{G_{l'}}(u,v') \leq \diam_{G_{l'}}-1$. Hence, $\dist_{G_{l'+1}}(u,v) = \dist_{G_{l'+1}}(v,v') + \dist_{G_{l'+1}}(u,v') \leq \diam_{G_{l'}} + \dist_{G_{l'}}(u,v')(\diam_{G_{l'}}+2) \leq \diam_{G_0}(\diam_{G_{l'}}+2)$.

\end{itemize}
This implies, if the inductive hypothesis is true for any $l=l' \geq 1$ then it is also true for $l=l'+1$. From the principle of mathematical induction one concludes that $\diam_{G_l} = \diam_{G_{0}}(\diam_{G_{l-1}}+2)$ for any integer $l \geq 1$.
\end{proof}

By using the recursive relation in the above lemma we get the following upper bound on the latency for entanglement distribution in RRGG. 

\begin{corollary}
\label{}
In the continuous model, for a virtual graph $\G_l = (V_l,\E_l)$ corresponding to any RRGG $G_l = (V_l,E_l)$ with $G_0 = (V_0,E_0) = C_n$ as base graph and $\p = \{C_n\}$, if $|D|=1$ and for all $i,j \in V_l$ and if $D_{i,j}\in \{0,1\}$ then for any source destination pair $s,e$, the latencies of the routing algorithms \ref{path_greed}, \ref{path_loc_best_eff} and \ref{path_non_local_best_effort} are upper bounded by,
\begin{align}
&O\left(\left(\frac{n}{2}\right)^{l+1}\right).
\end{align}
\end{corollary}


\section{Conclusion and Open Problems}
\label{concl}

In this paper, by combining ideas from complex network theory and quantum information theory, we propose a special kind of quantum network with pre-shared entangled links, which can reduce the latency of the network. As expected, we observe that the network with pre-shared entanglement is useful when the network has a small number of demands. If the network has a very high number of demands then the continuous model does not give many advantages compared to the on-demand model. Here we study the performance of our routing algorithms on some specific network topologies. We leave the investigation of the performance of the proposed entanglement pre-sharing strategies and routing algorithms in more generalised networks (Especially, the networks where the physical neighbour nodes are not deployed with equal distances) as future work. 

 For the continuous model, here we construct the virtual graph using the ideas from complex networks. This type of approaches reduce the diameter of the virtual graph and the routing algorithms require a lesser number of entanglement swap operations. This implies if each of the pre-shared links has very high fidelity then after the routing, the source and destination can share a very high fidelity entangled state. In table \ref{ent_swap} we give upper bounds on the number of required entanglement swap operations for a different type of virtual graphs. This upper bound depends on the threshold distance of $d_{\ths}$. Usually, higher $d_{\ths}$ helps to reduce the network latency but also decreases network throughput. It is always interesting to find a tradeoff between these two quantities. We leave this as future research. Other than this, in the continuous model, all of the proposed entanglement pre-sharing schemes are independent of the network traffic. One possible future research direction would be to use the tools from traffic engineering for predicting the behaviour of the network and study how one can choose the virtual neighbours according to the traffic. 
 
In this paper, we consider a very simple model for entanglement generation. For example, here we consider that due to the presence of noise in the quantum memory the fidelity of the stored state decays with time. As a result, one can store the state for a limited number of time steps ($T_{\ths}$) steps. This limited storage time can increase the latency of the network in the continuous model. One can increase the value of $T_{\ths}$ using the concept of entanglement distillation. Entanglement distillation process consumes two or more low fidelity entangled states and produces a new high fidelity entangled link. This might give some advantage in both reducing the latency and increasing the end to end fidelity at the expense of a higher setup cost. We leave the study of the advantages and disadvantages of entanglement distillation in the continuous model for future work.  

In our model, during the path discovery, first, we reserve the existing entangled links along the path. The nodes start performing the entanglement swap operations when the entire path has been established. This makes the routing algorithm stateful, i.e, each of the nodes needs to keep track of all of the paths which are passing through it. Besides this, it might increase the latency of the network, as the existing links might become useless while discovering the path. In our next approach, we would like to come up with a more dynamic version of the routing algorithm, where the nodes can perform the swap operation while discovering the path.

We also show that existing classical greedy routing algorithms are not capable of taking the full advantage of the pre-shared entangled links. In this paper we propose two new distributed routing algorithms and using numerical simulations on ring and grid topology we show that both of them perform better than the classical one. Simulations in figure \ref{comp_ring1}, \ref{comp_ring2}, \ref{comp_ring3} and \ref{capnondist} show that the routing algorithms which utilises the existing entangled links as much as possible, performs better for small number of demands. However, this way they might exhaust the pre-shared links and the network might converge to on-demand model very fast. For a high number of demands among the three routing algorithms modified greedy performs better than both local best effort and NoN local best effort algorithm. This observation is quite counterintuitive compared to the classical routing algorithms. The reason is that usually if the nodes have more information about the updated network topology, then the routing algorithm performs better than the case where the nodes have less information about updated network topology. It gives an insight for designing new routing algorithms.

In this paper, we study the advantages and disadvantages of the three types of virtual networks. From the simulation results, we observe that in the ring network, if we use local best effort and modified greedy routing algorithms then deterministic virtual graphs have low latency compared to random virtual graphs. However, for the grid network, the uniform and power-law virtual graphs show better performance. On the other hand, if we consider the end to end fidelity of the shared state, then table \ref{ent_swap} suggests to use the power-law virtual graphs.

 In the end, we propose another type of continuous model for a hierarchical network with a small diameter. Here, we show that all of the proposed routing algorithms can distribute an entangled link between a source and a destination using a very small number of entanglement swap operations. We leave the performance analysis of the routing algorithms on such networks for multiple source, destination pairs for future work.
 
 \section*{Acknowledgement}
We would like to acknowledge S. Bauml, K. Goodenough, A. Szava for many stimulating discussions on routing in a quantum internet. We would also like to thank ERC starting grant, NWO VIDI grant and Quantum internet Alliance (QIA) for providing the financial supports.
\bibliographystyle{IEEEtran}
\bibliography{IEEEabrv,mybibnet}

\appendix
\appendices
\label{proof_thm1}
The entire appendix is focused on giving the detailed proof of theorem \ref{ent_swap_begin}. Before going to the detailed proof, in the next appendix first we define some of the used notations again. One can find the more details about theorem \ref{ent_swap_begin} and the outline of the proofs in appendix \ref{over}.

\section{Notations}
\label{notation}
In this section, for the reader's convenience we define the some of used notations again.

\begin{itemize}
\item $G_{\phs} = (V,E_{\phs})$ : Physical graph, where $V$ denotes the set of nodes and $E_{\phs}$ denotes the set of edges in the physical graph.
\item $C_n$ : Ring network with $n$ nodes.
\item $\G$ : Virtual Graph.
\item $\dist_{G_{\phs}}(u,v)$ : The hop distance between two nodes $u,v$ in a physical graph $G_{\phs}$. 
\item $\diam_{G_{\phs}}$ : The diameter of the physical graph $G_{\phs}$. 
\item $d_{\ths}$ : The maximum threshold distance (in the physical graph) between two virtual neighbour.
\item $T_{\ths}$ : The threshold storage time. 
\item $\pchose(u,v)$ : Probability that two nodes $u,v$ are virtual neighbours of each other.
\begin{itemize}
\item For the power-law virtual graphs,
\begin{equation}
\label{pchoseapp}
\pchose(u,v) := 
\begin{cases}
\frac{1}{\beta_{u}}\frac{1}{\dist_{C_n}(u,v)},~\dist_{C_n}(u,v) \leq d_{\ths}\\
0~~~\text{Otherwise},
\end{cases}
\end{equation}
where $\beta_u = \sum_{\substack{v' \in V \\ 0< \dist(u,v') \leq d_{\ths}}} \frac{1}{\dist_{C_n}(u,v)}$.
\item For the uniform virtual graphs,
\begin{equation}
\label{pchoseunifapp}
\pchose(u,v) := 
\begin{cases}
\frac{1}{N_{\leq d_{\ths}}(G_{\phs})},~\dist_{G_{\phs}}(u,v) \leq d_{\ths}\\
 0~~~\text{Otherwise},
\end{cases}
\end{equation}
where $N_{\leq d_{\ths}}(G)$ denotes the number of nodes at a distance at most $d_{\ths}$ from a node $u$.
\end{itemize}
\item $s$ : Source node.
\item $e$ : Destination node.
\item $D$ : The demand matrix.
\item $D_{i,j}$ : (i,j)-th entry of the demand matrix $D$. It signifies, how many EPR pairs the source $i$ and the destination $j$ wants to create.
\item $|D|$ : The number of non zero entries in the demand matrix $D$.
\item $PathDisc(s,e,\G,D_{s,e})$ : The subroutine, used by all of the proposed routing algorithms (modified greedy (\ref{path_greed}), local best effort (\ref{path_loc_best_eff}) and NoN local best effort (\ref{path_non_local_best_effort})) for discovering paths between a source and a destination pair. 
\item $CommPath_{s,e}$ : The path, returned by the subroutine $PathDisc(s,e,\G,D_{s,e})$.
\end{itemize}

\section{Overview}
\label{over}
In the main paper, in theorem \ref{ent_swap_begin} we give the upper bound on the number of required entanglement swap operations for distributing entanglement between any two nodes in a ring and a grid network. For the ease for the reader's reading we restate the theorem again.

\textbf{Theorem 1:} \textit{In the continuous model with $d_{\ths} > 2$, if the virtual graphs are constructed from a physical network ($G_{\phs}$) like a ring ($C_n$) or a grid ($Grid_{n \times n}$) topology, if $|D|=1$ and for all $i,j \in [0,n-1]$ and if $D_{i,j}\in \{0,1\}$ then for any source destination pair $s,e$, the expected number of required entangled swap operations for sharing an entangled link between $s,e$ is as given in table \ref{ent_swap_app}.}

\begin{table*}[h!]
\centering
\begin{tabular}{| m{8em} | m{6cm}| m{6cm} | } 
 \hline
 Models &  Modified Greedy and Local-Best Effort Routing & NoN Local-Best Effort Routing \\ 
  \hline
Deterministically chosen virtual links & $O(\frac{n}{d_{\ths}} + \log d_{\ths})$ & $O(\frac{n}{d_{\ths}}+\log d_{\ths})$  \\
 \hline
Virtual links chosen with power-law distribution & $O\left(\frac{n}{d_{\ths}} + \log d_{\ths}\right)$ & $O\left(\frac{n}{d_{\ths}} + \frac{\log d_{\ths}}{ \log \log d_{\ths}}\right)$\\
 \hline
Virtual links chosen with uniform distribution & $O\left(\frac{n}{d_{\ths}} + \frac{d_{\ths}}{(\log d_{\ths})^2}\right)$ & $O\left(\frac{n}{d_{\ths}} + \frac{d_{\ths}}{(\log d_{\ths})^2}\right)$ \\
 \hline
\end{tabular}
\vspace{0.2in}
\caption{Expected number of entanglement swap for ring and grid network with single source-destination pair.}
\label{ent_swap_app}
\end{table*}

In this paper, we only give the detailed proofs of the theorem \ref{ent_swap_begin} for the ring network with $n$ nodes ($C_n$). The same proof technique can be used for proving the theorem for the grid networks. Note that, in the table \ref{ent_swap_app} the number of required entanglement swap operations changes with the type of the virtual graphs. In this paper, we study deterministic, power-law and uniform virtual graphs. The structure of the proofs for all of these graphs is organised in the following manner.

\begin{enumerate}

\item In the first row of the table \ref{ent_swap_app} we have the upper bound on the number of swap operations for the deterministic virtual graphs. For these type of virtual graphs, the required number of swap operations, presented in the table \ref{ent_swap_app}, remains the same for all of the routing algorithms and the proofs directly follow from \cite{SMIKW16}. Hence, we do not include the proofs for such virtual graphs.

\item In the second row of the table \ref{ent_swap_app}, we have the upper bound on the number of swap operations for the power-law virtual graphs. One can find the detailed proofs of the bounds in the second row of the table \ref{ent_swap_app} in appendix \ref{proof_power}. More precisely, using lemma \ref{ndth} and lemma \ref{y_greed} we prove the upper bound on the number of entanglement swap operations that are required by the modified greedy and the local best effort routing algorithms for distributing entanglement in a power law virtual graph $\left(\text{The upper bound is }O\left(\frac{n}{d_{\ths}} + \log d_{\ths}\right)\right)$. Similarly, using lemma \ref{ndth} and lemma \ref{y_non_greed} we prove the upper bound on how many entanglement swap operations the NoN local best effort algorithm takes for distributing an entangled link between any two nodes in a power law virtual graph $\left(\text{The upper bound is }O\left(\frac{n}{d_{\ths}} + \frac{\log d_{\ths}}{ \log \log d_{\ths}}\right)\right)$. One can find lemma \ref{ndth} , \ref{y_greed} and \ref{y_non_greed} in appendix \ref{proof_power}.

\item In the third row of the table \ref{ent_swap_app}, we have the upper bound on the number of swap operations for the uniform virtual graphs. One can find the detailed proofs of the bounds in the third row of the table \ref{ent_swap_app} in appendix \ref{proof_unif}. Note that, in the table \ref{ent_swap_app}, both of the upper bounds in the third row are the same. However, due to the different behaviour of the routing algorithms, we need to use different proof techniques. For the uniform virtual graphs, by combining the results from lemma \ref{ndth_unif} and lemma \ref{y_greed_unif} we give the proof on the upper bound for the modified greedy and local best effort routing algorithm $\left(\text{The upper bound is }O\left(\frac{n}{d_{\ths}} + \frac{d_{\ths}}{(\log d_{\ths})^2}\right)\right)$. One can get the proof of the upper bound on the entanglement swap operations for the NoN local best effort algorithm by combining the results from lemma \ref{ndth_unif} and lemma \ref{y_non_unif} $\left(\text{The upper bound is }O\left(\frac{n}{d_{\ths}} + \frac{d_{\ths}}{(\log d_{\ths})^2}\right)\right)$.
 \end{enumerate}

\subsection{\textbf{Detailed Outline of the Proofs}}
In the routing algorithm \ref{algo_greed_dem}, the total number of required entanglement swap operations to share an entangled link between a source node $s$ and a destination node $e$ is the same as the size of the set $CommPath_{s,e}$ (i.e., $|CommPath_{s,e}|$). In theorem \ref{ent_swap_begin} we would like to give an upper bound on the following quantity, 

\begin{equation}
\label{target}
\max_{s,e \in C_n}E[|CommPath_{s,e}|].
\end{equation}

Here, the maximum is taken over all possible source destination pairs $(s,e)$ in $C_n$.

In algorithm \ref{algo_greed_dem}, the subroutine $PathDisc(s,e,\G,D_{s,e})$ computes the routing path $CommPath_{s,e}$. In this paper we consider three different types of $PathDisc(s,e,\G,D_{s,e})$ subroutines (See algorithms \ref{path_greed}, \ref{path_loc_best_eff} and \ref{path_non_local_best_effort}) and each of them are greedy and distributed in nature. This implies, the path has been discovered in a hop by hop fashion. In this discovery process the current node $u$ chooses one of its neighbours $v$ as the next hop of the path if $\dist_{C_n}(u,e) > \dist_{C_n}(v,e)$. In the physical graph $C_n$, for two nodes $u,v$, $\dist_{C_n}(u,v) = \min(|u-v|, n-|u-v|)$. For simplicity we assume that $\min(|s-e|, n-|s-e|) = |s-e|$. 

In order to give an upper bound on the equation \ref{target}, first we partition the set of all nodes in $C_n$ into $m'$ nonempty sets, $Z^{(s,e)}_0, \ldots , Z^{(s,e)}_{(m'-1)}$, where $m' = \left\lceil\frac{2|s-e|}{d_{\ths}}\right\rceil$. Here for all $0\leq i \leq m'-1$, $$Z^{(s,e)}_i := \left\{u : |u-e| \le |s-e| - \frac{id_{\ths}}{2}\right\}.$$ As all of the $PathDisc(s,e,\G,D_{s,e})$ subroutines are greedy in nature, so they start by constructing the set $CommPath_{s,e}$ from the source node $s \in Z^{(s,e)}_0$ and they discover the path through the nodes in the sets $Z^{(s,e)}_1, \ldots , Z^{(s,e)}_{m'-1}$. Each node $u \in Z^{(s,e)}_{i}$ is connected to a node in $v \in Z^{(s,e)}_{i+1}$ with probability $\pchose(u,v)$. Here note that if the algorithms discover a node $u\in Z^{(s,e)}_{m'-1}$ then $|u-e| \leq \frac{d_{\ths}}{2}$. Let $\Z^{(s,e)}_i$ denotes the number of nodes the path discovery algorithm visits to discover a node $v$ in the set $Z^{(s,e)}_{i+1}$ from a node $u$ in the set $Z^{(s,e)}_i$.


 
 According to the algorithms \ref{path_greed}, \ref{path_loc_best_eff} and \ref{path_non_local_best_effort} the total number of required entanglement swap is same as the length of the discovered path ($|CommPath_{s,e}|$). This implies,
\begin{equation}
\label{comm_len}
|CommPath_{s,e}| = \sum_{i=0}^{m'-2} \Z^{(s,e)}_i + Y_{s,e}.
\end{equation}
where $Y_{s,e}$ denotes the length of the discovered path from a node $u \in \Z^{(s,e)}_{m'-2}$ to $e$. By taking expectation on both sides of equation \ref{comm_len} we get,
\begin{align*}
E[|CommPath_{s,e}|] &=   E\left[\sum_{i=0}^{m'-2} \Z^{(s,e)}_i + Y_{s,e}\right]\\
& = \sum_{i=0}^{m'-2} E[\Z^{(s,e)}_i] + E[Y_{s,e}],
\end{align*}
In the theorem \ref{ent_swap_begin} we prove the upper bound on the following quantity,

\begin{align*}
\max_{s,e \in C_n}E[|CommPath_{s,e}|]  =  \max_{s,e}&\left[\sum_{i=0}^{m'-2} E[\Z^{(s,e)}_i] + E[Y_{s,e}]\right]\\
\leq  \max_{s,e \in C_n}\left[\sum_{i=0}^{m'-2} E[\Z^{(s,e)}_i]\right]& +  \max_{s,e}E[Y_{s,e}].
\end{align*} 

In the appendix \ref{proof_power} we give the detailed proof of the upper bound on $\max_{s,e}E[|CommPath_{s,e}|]$ for the power-law network. The proof has been subdivided into two parts. In the first part we focus on giving an upper bound on $\sum_{i=0}^{m'-2} E[\Z^{(s,e)}_i]$ (See lemma \ref{ndth}). The proof for this part is the same for all of the proposed $PathDisc(s,e,\G,D_{s,e})$ subroutines. However, in the second part, the proof of the upper bound on $E[Y_{s,e}]$ changes with the $PathDisc(s,e,\G,D_{s,e})$ subroutines. We prove the upper bounds on $E[Y_{s,e}]$ for algorithms \ref{path_greed} and \ref{path_loc_best_eff} in lemma \ref{y_greed}. We give the proof of the upper bound on $E[Y_{s,e}]$ for algorithm \ref{path_non_local_best_effort} in lemma \ref{y_non_greed}.

Similarly, appendix \ref{proof_unif} contains the detailed proof of the upper bound on $\max_{s,e}E[|CommPath_{s,e}|]$ for the uniform virtual graph. Lemma \ref{ndth_unif} contains the proof of the upper bound on $\sum_{i=0}^{m'-2} E[\Z^{(s,e)}_i]$. The proof of the upper bound on $E[Y_{s,e}]$ for the $PathDisc(s,e,\G,D_{s,e})$ subroutines \ref{path_greed} and \ref{path_loc_best_eff} is given in lemma \ref{y_greed_unif}. For the proof of the upper bound on $E[Y_{s,e}]$ for the path discovery algorithm \ref{path_non_local_best_effort} we refer to lemma \ref{y_non_unif}.



\section{Upper Bound on the number of Entanglement Swap Operations for Power Law Virtual Graphs}
\label{proof_power}

In the continuous model, the maximum distance (distance in the physical graph) between two virtual neighbour nodes is $d_{\ths}$. Due to this upper bound, in the worst case two nodes need to perform at least $\frac{\diam_{G_{\phs}}}{d_{\ths}}$ swap operations for sharing an entangled state. In lemma \ref{ndth} we show that, in worst case (when the distance between a source $s$ and a destination $e$ is $\diam_{G_{\phs}}$), all of the proposed routing algorithms take $O(\frac{n}{d_{\ths}})$ number of swap operations for distributing an entangled link between, a source $s$ and a node $u$, such that $\dist_{G_{\phs}}(u,e) \leq \frac{d_{\ths}}{2}$. Later, in lemma \ref{y_greed} we show that, both local best effort and modified greedy routing algorithms take $O(\log_2 d_{\ths})$ number of swap operations to create an entangled state between $u$ and the destination $e$. By combining the results of lemma \ref{ndth} and \ref{y_greed} we prove the following upper bound.

\begin{equation}
\max_{s,e \in C_n}E[|CommPath_{s,e}|] \leq O\left(\frac{n}{d_{\ths}} + \log_2 d_{\ths}\right)
\end{equation}

In lemma \ref{y_non_greed} we show that using the NoN greedy routing algorithm, the node $u$ can share an entanglement link with $e$ using $O(\frac{\log_2 d_{\ths}}{\log_2 \log_2 d_{\ths}})$ number of swap operations. Using the results of lemma \ref{ndth} and \ref{y_non_greed} we prove the following upper bound,

\begin{equation}
\max_{s,e\in C_n}E[|CommPath_{s,e}|] \leq O\left(\frac{n}{d_{\ths}} +\frac{\log_2 d_{\ths}}{\log_2 \log_2 d_{\ths}}\right)
\end{equation}


For the power-law virtual graphs, a node $u$ choses another node $v$ as virtual neighbour with probability $\pchose(u,v)$. According to the definition of $\pchose$ (see equation \ref{pchoseapp}) we have,

\begin{equation}
\label{pchoseapp_again}
\pchose(u,v) := 
\begin{cases}
\frac{1}{\beta_{u}}\frac{1}{\dist_{C_n}(u,v)},~\dist_{C_n}(u,v) \leq d_{\ths}\\
0~~~\text{Otherwise},
\end{cases}
\end{equation}
where $\beta_u = \sum_{\substack{v' \in V \\ 0< \dist(u,v') \leq d_{\ths}}} \frac{1}{\dist_{C_n}(u,v)}$.

In the next lemma, we give an upper bound on $\beta_u$ for the physical graph $C_n$.

\begin{lemma}
\label{upp_beta}
In a power law virtual graph, which has been constructed from a physical graph $C_n$, if any two nodes $u,v$ are virtual neighbours of each other with probability, $\pchose(u,v)$, such that,

\begin{equation}
\label{pchoseappagg}
\pchose(u,v) := 
\begin{cases}
\frac{1}{\beta_{u}}\frac{1}{\dist_{C_n}(u,v)},~\dist_{C_n}(u,v) \leq d_{\ths}\\
0~~~\text{Otherwise},
\end{cases}
\end{equation}
where $\beta_u = \sum_{\substack{v' \in V \\ 0< \dist(u,v') \leq d_{\ths}}} \frac{1}{\dist_{C_n}(u,v)}$, then for any $u$, the value of $\beta_u \leq 4\log_2 d_{\ths}$.

\end{lemma}
\begin{proof}
According to the definition of $\beta_u$ we have,
\begin{align}
 \nonumber
\beta_u &= \sum_{\substack{v' \in V \\ 0< \dist(u,v') \leq d_{\ths}}} \frac{1}{\dist_{C_n}(u,v')}\\ \nonumber
& = \sum_{j=1}^{d_{\ths}} \sum_{\substack{v' \in V\\\dist_{C_n}(u,v') = j}}  \frac{1}{j}\\ \nonumber
& \leq \sum_{j=1}^{d_{\ths}} \frac{2}{j}\\ \nonumber
&\text{As }\sum_{j=1}^{d_{\ths}}\frac{1}{j} \leq \log_{2}d_{\ths} +1 < 2\log_{2}d_{\ths}, \text{this implies}\\  
\label{upper_beta}
\beta_u & \leq 4\log_{2}d_{\ths}.
\end{align} 

\end{proof}

\begin{lemma}
\label{ndth}

For the power-law virtual graphs, constructed from the physical graph $C_n$, if $|D|=1$ and for all $i,j \in [0,n-1]$ if $D_{i,j}\in \{0,1\}$, then for any source destination pair $(s,e)$, all of the algorithms \ref{path_greed}, \ref{path_loc_best_eff} and \ref{path_non_local_best_effort} take $O(\frac{n}{d_{\ths}})$ number of entanglement swap operations for distributing an entangled link between $s$ and a node $u$ such that $\dist_{C_n}(u,e) \leq \frac{d_{\ths}}{2}$. 
\end{lemma}


\textit{Proof Outline :} For the proof of this lemma, we divide the set of nodes of $C_n$ into $m'$ nonempty sets, $Z^{(s,e)}_0, \ldots , Z^{(s,e)}_{m'-1}$, where 
\begin{equation}
\label{eqm}
m' := \left\lceil\frac{2|s-e|}{d_{\ths}}\right\rceil.
\end{equation}
 Here for all $0\leq i \leq m'-1$, 
 \begin{equation}
 \label{eq_zpower}
 Z^{(s,e)}_i := \left\{w : |w-e| \leq |s-e|-\frac{id_{\ths}}{2}\right\}.
 \end{equation}
  The greedy path discovery subroutines, start constructing the path from $s \in Z^{(s,e)}_0$. In this process if a node $w \in Z^{(s,e)}_i$, chooses $v$ as a next hop in the path, then either $v \in Z^{(s,e)}_i$ or $v \in Z^{(s,e)}_{i+1}$ ($0 \leq i \leq m'-2$). Note that all the nodes $u \in Z^{(s,e)}_{m'-1}$ are at most $\frac{d_{\ths}}{2}$ distance away from the destination $e$. for all $0\leq i \leq m'-2$, $\Z^{(s,e)}_i$ denotes the number of nodes the path discovery algorithms visit to discover a node in the set $Z^{(s,e)}_{i+1}$ from a node in the set $Z^{(s,e)}_i$, then in this lemma we would like to prove,

\begin{equation}
\sum_{i=0}^{m'-2} E[\Z^{(s,e)}_i] \leq O\left(\frac{n}{d_{\ths}}\right).
\end{equation}

In the proof, first we show that each of $E[\Z^{(s,e)}_i]$ is upper bounded by a constant number, which is independent of $n$ and $d_{\ths}$. As a result, we get $\sum_{i=0}^{m'-2} E[\Z^{(s,e)}_i] \leq O(m')$. Substituting the value of $m'$ from equation \ref{eqm} we get $\sum_{i=0}^{m'-2} E[\Z^{(s,e)}_i] \leq O(m') = O\left(\left\lceil\frac{2|s-e|}{d_{\ths}}\right\rceil \right)$. As, $|s-e| \leq \diam_{C_n} = \lceil\frac{n}{2}\rceil$. This implies, $\sum_{i=0}^{m'-2} E[\Z^{(s,e)}_i] \leq O(m') \leq O\left(\frac{n}{d_{\ths}}\right)$. The detailed proof is given below. 

\begin{proof}   

We consider the situation where, for any $0\leq i \leq m'-1$ the path discovery algorithms discover the path from $s$ to a node $w \in Z^{(s,e)}_i$. In the next step of the path discovery algorithm, $w$ choses the next hop $v$ from his neighbour nodes, such that $\dist_{C_n}(v,e)< \dist_{C_n}(w,e)$. We are interested in computing the total number of required entanglement swap operations for creating an entangled link with a node $v \in Z^{(s,e)}_{i+1}$ from any node $w \in Z^{(s,e)}_{i}$. Let $\Pr[w \rightarrow Z^{(s,e)}_{i+1}]$ denotes the probability that $w$ has at least one neighbour in the set $Z^{(s,e)}_{i+1}$. From the definition of $\pchose$ (see equation \ref{pchoseapp}) we have,
\begin{align*}
\pchose(w,v)  &= \frac{1}{\beta_{w}\dist_{C_n}(w,v)}~~\text{if}~|w-v|\leq d_{\ths}\\
& \geq \frac{1}{\beta_{w}d_{\ths}},~~\text{As}~|w-v|\leq d_{\ths}\\
& = 0 ~~~~\text{Otherwise}.
\end{align*}

Let $Z'^{(s,e)}_{i+1} \subseteq Z^{(s,e)}_{i+1}$ denotes the set of nodes $v$ such that $\forall v \in Z'^{(s,e)}_{i+1}$, $|w-v| \leq d_{\ths}$. Each of the nodes in $Z'^{(s,e)}_{i+1}$ choose $w$ as a virtual neighbour with probability $\pchose(w,v)$. This implies $w$ has at least one virtual neighbour in $Z^{(s,e)}_{i+1}$ with probability $|Z'^{(s,e)}_{i+1}|\pchose(w,v)$. In section \ref{power_ring} we mention that in the virtual graph, $w$ has at least $\log_2 d_{\ths}$ such virtual neighbours and each of them belongs to the set $Z^{(s,e)}_{i+1}$ with probability $|Z'^{(s,e)}_{i+1}|\pchose(w,v)$. This implies,
\begin{equation}
\label{temp_w1}
\Pr[w \rightarrow Z^{(s,e)}_{i+1}]  \geq 1- (1-|Z'^{(s,e)}_{i+1}|\pchose(w,v))^{\log_2 d_{\ths}},
\end{equation}
For all real $x$ and $r> 0 $, we have $(1-x)^r \leq \exp(-xr)$. In the above equation if we consider $x = |Z'^{(s,e)}_{i+1}|\pchose(w,v)$ and $r = \log_2 d_{\ths}$, then we get, $(1-|Z'^{(s,e)}_{i+1}|\pchose(w,v))^{\log_2 d_{\ths}} \leq \exp(-|Z'^{(s,e)}_{i+1}|\pchose(w,v)\log_2 d_{\ths})$. Substituting this inequality in equation \ref{temp_w1} we get,
\begin{align}
\nonumber 
\Pr[w \rightarrow Z^{(s,e)}_{i+1}]& \geq 1- \exp(-|Z'^{(s,e)}_{i+1}|\pchose(w,v)\log_2 d_{\ths}).
\end{align}
Substituting the value of $\pchose(w,v)$ we get,
\begin{align}
\nonumber 
\Pr[w \rightarrow Z^{(s,e)}_{i+1}]& \geq 1- \exp\left(-\frac{|Z'^{(s,e)}_{i+1}|\log_2 d_{\ths}}{\beta_w d_{\ths}}\right).
\end{align}
Substituting the value of $\beta_w$ from lemma \ref{upp_beta} we get,
\begin{align}
\nonumber
\Pr[w \rightarrow Z^{(s,e)}_{i+1}]& \geq 1- \exp\left(-\frac{|Z'^{(s,e)}_{i+1}|\log_2 d_{\ths}}{4d_{\ths}\log_2 d_{\ths} }\right)\\
\label{eqmain}
\Pr[w \rightarrow Z^{(s,e)}_{i+1}]& \geq 1- \exp\left(-\frac{|Z'^{(s,e)}_{i+1}|}{4d_{\ths}}\right).
\end{align}

Now we focus on giving a lower bound on $|Z'^{(s,e)}_{i+1}|$. For the lower bound, we consider the nodes in $Z^{(s,e)}_{i}$ which are the furthest from the set $Z'^{(s,e)}_{i+1}$. According to the definition (see equation \ref{eq_zpower}), all of the nodes in $Z'^{(s,e)}_i$ are at most $|s-e| - \frac{id_{\ths}}{2}$ distance away from the destination node $e$ and all of the nodes in $Z'^{(s,e)}_{i+1}$ are at most $|s-e| - \frac{(i+1)d_{\ths}}{2}$ away from the destination $e$. The maximum distance between a node $\omega \in Z^{(s,e)}_{i}$ and the set $Z'^{(s,e)_{i+1}}$ is computed in the following equation, 

\begin{equation}
\max_{\omega \in Z^{(s,e)}_i} \min_{v' \in Z^{(s,e)}_{i+1}} |\omega - v'|.
\end{equation}
From the triangle inequality we have, for any three nodes $\omega, v', e$, $|\omega - v'| \geq |\omega - e| - |v' - e|$. This implies, 
\begin{align}
\nonumber
\max_{\omega \in Z^{(s,e)}_i} \min_{v' \in Z^{(s,e)}_{i+1}} |\omega - v'| &\geq \max_{\omega \in Z^{(s,e)}_i} \min_{v' \in Z^{(s,e)}_{i+1}} |\omega - e| - |v' - e|\\ \label{temp_z}
& = \max_{\omega \in Z^{(s,e)}_i} |\omega - e| - \min_{v' \in Z^{(s,e)}_{i+1}}   |v' - e|
\end{align}
As $ \omega \in Z^{(s,e)}_i$ this implies, $|\omega - e| \geq |s-e| - \frac{id_{\ths}}{2}$. Similarly, as $v' \in Z^{(s,e)}_{i+1}$, this implies $|v'-e| \leq |s-e| - \frac{(i+1)d_{\ths}}{2}$. Substituting the values of $|\omega - e|$ and $|v' - e|$ in equation \ref{temp_z} we get, 
\begin{align*}
\max_{\omega \in Z^{(s,e)}_i} \min_{v' \in Z^{(s,e)}_{i+1}} |\omega - v'| & \geq |s-e| - \frac{id_{\ths}}{2} - |s-e| + \frac{(i+1)d_{\ths}}{2}\\
&= \frac{d_{\ths}}{2}.
\end{align*}

The above derivation implies that in the worst case all of the nodes $Z'^{(s,e)}_{i+1}$ are at least $\frac{d_{\ths}}{2}$ distance away from $w$. According to the definition of $Z'^{(s,e)}_{i+1}$, each of the nodes in this set are at most $d_{\ths}$ distance away from the node $w$. This implies, the maximum distance between any two nodes in $Z'^{(s,e)}_{i+1}$ is at least $\frac{d_{\ths}}{2}$. As, in $C_n$, there are at least $d$ number of nodes within $d$ distance from any node $w$. This implies, 

\begin{equation}
\label{set_size}
|Z'^{(s,e)}_{i+1}| \geq \frac{d_{\ths}}{2}. 
\end{equation}
Substituting the value of $|Z'^{(s,e)}_{i+1}|$ in equation \ref{eqmain} we get,
\begin{align}
\nonumber
\Pr[w \rightarrow Z^{(s,e)}_{i+1}] &\geq   1- \exp\left(-\frac{|Z'^{(s,e)}_{i+1}|}{8d_{\ths}}\right)\\
\label{low_bound_power}
& \geq 1- \exp\left( -\frac{1}{8}\right).
\end{align}

Here, $\Z^{(s,e)}_{i+1}$ denotes the number of required swap operations for distributing an entangled link between a node in $ Z^{(s,e)}_{i+1}$ from a node in $ Z^{(s,e)}_{i}$. From the equation \ref{low_bound_power} we have that each of such node $w \in Z^{(s,e)}_{i}$ is connected to a link $v \in  Z^{(s,e)}_{i+1}$ with probability at least $1- \exp\left( -\frac{1}{8}\right)$. This implies, each swap operation manages to create an entangled link with a node $v \in Z^{(s,e)}_{i+1}$ with probability at least $1- \exp\left( -\frac{1}{8}\right)$. The phenomenon of this entanglement can be modelled as a sequence of independent trials, where each trial succeeds with probability at least $1- \exp\left( -\frac{1}{8}\right)$. This implies, $\Z^{(s,e)}_{i+1}$ follows a geometric distribution with parameter $1- \exp\left( -\frac{1}{8}\right)$. So, we have, for all $0\leq i \leq m'-2$, $E[\Z^{(s,e)}_{i+1}]\leq \left(1- \exp\left( \frac{1}{8}\right)\right)^{-1} = O(1)$. This implies, 

\begin{align}
\nonumber
\sum_{i=0}^{m'-2} E[\Z^{(s,e)}_i] &\leq  \sum_{i=0}^{m'-2} O(1) .
\end{align}
Substituting the value of $m'$ from equation \ref{eqm} in the above equation we get,

\begin{align}
\label{temp_upp1}
\sum_{i=0}^{m'-2} E[\Z^{(s,e)}_i]& \leq O\left(\left\lceil\frac{2|s-e|}{d_{\ths}}\right\rceil \right).
\end{align}
As in a ring network with $n$ nodes ($C_n$) the maximum distance between any two nodes is at most $\frac{n}{2}$, this implies, for any $s,e \in V$, $|s-e| \leq \frac{n}{2}$. Substituting the upper bound on $|s-e|$ in equation \ref{temp_upp1} we get,
\begin{align}
\label{upp_stand}
\sum_{i=0}^{m'-2} E[\Z^{(s,e)}_i] \leq O\left(\frac{n}{d_{\ths}}\right).
\end{align}
This concludes the proof.
\end{proof}
Substituting the upper bound on $\sum_{i=0}^{m'-2} E[\Z^{(s,e)}_i]$ to equation \ref{target} we get,
\begin{equation}
\label{target_sofar}
\max_{s,e}E[|CommPath_{s,e}|] \leq O\left(\frac{n}{d_{\ths}} \right) + \max_{s,e}E[Y_{s,e}].
\end{equation}

Up to this part, the proofs are the same for all of the proposed algorithms. The analysis changes when we compute the upper bounds of $E[Y_{s,e}]$.

\subsection{\textbf{Proof of the upper bound on the swap operations for both local best effort and modified greedy algorithm}}

In equation \ref{target_sofar}, the term $Y_{s,e}$ is related to the term $\max_{s,e}E[|CommPath_{s,e}|]$. For a source destination pair $(s,e)$, the term $Y_{s,e}$ denotes how many more entanglement swap operations are required for creating an entangled link between $s$ and $e$, given that $s$ already has created an entangled link with a node $u$ such that $\dist_{C_n}(u,e) \leq \frac{d_{\ths}}{2}$. In lemma \ref{y_greed} we show that the modified greedy and the local best effort algorithm can share an entangled state between any two nodes $u,e$ such that $\dist_{C_n}(u,e) \leq \frac{d_{\ths}}{2}$, using only $ O(\log_2 d_{\ths})$ swap operations. This gives us the required bound for $\max_{s,e}E[Y_{s,e}] \leq O(\log_2 d_{\ths})$. The proof technique is just a simple adaptation of the proof given in \cite{klein99}. For completeness, here we include the proof.

\begin{lemma}
\label{y_greed}
For the power-law virtual graphs, constructed from the physical graph $C_n$, if $|D|=1$ and for some $u,e \in [0,n-1]$ if $D_{u,e} = 1$ and if $\dist_{C_n}(u,e) \leq \frac{d_{\ths}}{2}$ then the algorithms \ref{path_greed}, \ref{path_loc_best_eff} take $O(\log_2 d_{\ths})$ number of entanglement swap operations for distributing an entangled link between $u$ and $e$.
\end{lemma}

\begin{proof}
Let us consider a hierarchy of of $m+1$ nonempty sets, $X^{(u,e)}_0 \supset X^{(u,e)}_1 \supset X^{(u,e)}_2 \supset \ldots \supset X^{(u,e)}_m$, where 
\begin{equation}
\label{upp_m}
m := \log_2 d_{\ths},
\end{equation}
and for all $0\leq i \leq m$, $$X^{(u,e)}_i := \left\{u' : |u'-e| \le \frac{d_{\ths}}{2^i}\right\}.$$
Note that, all the nodes in $X^{(u,e)}_m$ are at most one distance (in the physical graph) away from the destination node $e$. This implies, all of the nodes in $X^{(s,e)}_m$ are physical neighbours of $e$.

In a ring network, $C_n$, one can verify that for all $0 \leq i \leq m$,

\begin{equation}
\label{eqsize_power1}
|X^{(u,e)}_i| \geq \frac{d_{\ths}}{2^i}.
\end{equation}

The algorithms start to discover a path from a node in $u \in X^{(u,e)}_0$ to a node in the set $X^{(u,e)}_{m}$ through the nodes in the sets $X^{(u,e)}_{1}, \ldots , X^{(u,e)}_{m-1}$. The algorithms stop when it discovers the path up to the destination node in $X^{(u,e)}_m$. Let the algorithm spends $Y^{(u,e)}_i$ iterations in the set $X^{(u,e)}_i$. According to Algorithm \ref{path_greed} and \ref{path_loc_best_eff} the length of the path is,
\begin{equation}
\label{sumY_greed}
 \sum_{i=0}^{m} Y^{(u,e)}_i.
\end{equation}



Here, first we show that for each $0 \leq i \leq m$, $E[Y^{(u,e)}_i]$ is upper bounded by a constant number. This gives us the proof that, $\sum_{i=0}^m E[Y^{(u,e)}_i] \leq O(m)$. Substituting the value of $m$ from equation \ref{upp_m} we have $\sum_{i=0}^m E[Y^{(u,e)}_i] \leq O(m) \leq O(\log_2 d_{\ths})$. The detailed proof is given below.

For any $0 \leq i \leq m-1$, suppose the path discovery algorithm has discovered a path from $u \in X^{(u,e)}_0$ to a node $u'$ such that $u' \in X^{(u,e)}_i$ but $u' \not\in X^{(u,e)}_{i+1}$. Let $u' \rightarrow X^{(u,e)}_{i+1}$ denotes the event that $u'$ has a virtual neighbour in the set $X^{(u,e)}_{i+1}$. In a similar way, by the notation $u' \not\rightarrow X^{(u,e)}_{i+1}$ we denote the event that $u'$ has no virtual neighbour in $X^{(u,e)}_{i+1}$. Let $v \in X^{(u,e)}_{i+1}$, then from equation \ref{pchoseapp} we have $\pchose(u',v) \geq \frac{1}{\beta_v} \frac{1}{|u'-v|}$. From lemma \ref{upp_beta} we have, $\beta_v \leq 4\log_2 d_{\ths}$. This implies 
\begin{equation}
\label{pchoselem}
\pchose(u',v) \geq \frac{1}{4\log_2 d_{\ths}} \frac{1}{|u'-v|}. 
\end{equation}
As $v\in X^{(u,e)}_{i+1}$ so $|v-e| \leq \frac{d_{\ths}}{2^{i+1}}$. From the triangle inequality of the distance function we have,

\begin{align*}
|u'-v| &\leq |u'-e| + |v-e|,
\end{align*}
As, $u' \in X^{(s,e)}_i$, this implies $|u'-e| \leq \frac{d_{\ths}}{2^{i}}$. Similarly, as $v\in  X'^{(s,e)}_{i+1}$, this implies, $|v-e| \leq \frac{d_{\ths}}{2^{i}}$. Substituting $|u'-e|$ and $|v-e|$ in the above equation we get,
\begin{align*}
|u'-v|& \leq \frac{2d_{\ths}}{2^{i}}.
\end{align*}

By substituting the upper bound on of $|u'-v|$ in the expression of $\pchose(u',v)$ in equation \ref{pchoselem} we get,

\begin{equation}
\label{low_p}
\pchose(u',v) \geq  \frac{1}{8\log_2 d_{\ths}} \frac{2^i}{d_{\ths}}.
\end{equation}

In the power-law virtual graph, each of the nodes $v \in X^{(s,e)}_{i+1}$ is a virtual neighbour of $u'$ with probability $\pchose(u',v)$. So, the probability that $u'$ has a virtual neighbour in $X^{(s,e)}_{i+1}$ is, $$\sum_{v \in X^{(u,e)}_{i+1}} \pchose(u',v).$$
In the above expression, substituting the value of $\pchose(u',v)$ from equation \ref{low_p} we get,

\begin{align*}
\sum_{v \in X^{(u,e)}_{i+1}} \pchose(u',v) &\geq \sum_{v \in X^{(u,e)}_{i+1}}\frac{1}{8\log_2 d_{\ths}} \frac{2^i}{d_{\ths}} \\
&\geq |X^{(u,e)}_{i+1}|\frac{1}{8\log_2 d_{\ths}} \frac{2^i}{d_{\ths}}.
\end{align*}
Substituting $|X^{(u,e)}_{i+1}|$ from equation \ref{eqsize_power1} we get,
\begin{align*}
\sum_{v \in X^{(u,e)}_{i+1}} \pchose(u',v) &\geq \frac{d_{\ths}}{2^{i+1}}\frac{1}{16\log_2 d_{\ths}} \frac{2^i}{d_{\ths}}\\
& = \frac{1}{16\log_2 d_{\ths}}.
\end{align*}

As $u'$ has $k=\log_2 d_{\ths}$ such virtual neighbours and each of them is identical and independently distributed according to $\pchose$. This implies,  

\begin{align*}
\Pr[u' \not\rightarrow X^{(u,e)}_{i+1}] &\leq (1 - |X^{(u,e)}_{i+1}|\pchose(u',v))^k
\end{align*}
As for all real $x$ and $r>0$ we have $(1-x)^r \leq \exp(-xr)$, this implies

\begin{align*}
\Pr[u' \not\rightarrow X^{(u,e)}_{i+1}] & \leq \exp(-k|X_{i+1}|\pchose(u',v)).
\end{align*}

Substituting the values of $\pchose$ and $|X_{i+1}|$ from equation \ref{low_p} and  \ref{eqsize_power1} we get,

\begin{align*}
\Pr[u' \not\rightarrow X^{(u,e)}_{i+1}] & \leq \exp(-\frac{k}{16\log_2 d_{\ths}}).
\end{align*}
Substituting $k = \log_2 d_{\ths}$ we get,
\begin{align*}
\Pr[u' \not\rightarrow X^{(u,e)}_{i+1}]  & \leq \exp(-\frac{1}{16})\\
1- \Pr[u' \not\rightarrow X^{(u,e)}_{i+1}] &\geq 1- \exp(-\frac{1}{16})
\end{align*}
As for any $0\leq x \leq 1$, $1-e^{-x} \geq \frac{x}{2}$, this implies,
\begin{align*}
\Pr[u' \rightarrow X^{(u,e)}_{i+1}] & \geq \frac{1}{32}.
\end{align*}

If the algorithms discover the path to a node $u' \in X^{(u,e)}_i$, then with probability at least $\frac{1}{32}$, the algorithms discover another node $v \in X^{(u,e)}_{i+1}$. This implies, each step of the algorithms manages to find a neighbour node $v \in X^{(s,e)}_{i+1}$ with probability at least $\frac{1}{32}$. The phenomenon of this finding a node $v\in X^{(s,e)}_{i+1}$ can be modelled as a sequence of independent trials, where each trial has two outcomes (success and failure) and each trial succeeds with probability at least $\frac{1}{32}$. This implies, $Y^{(u,e)}_i$ follows a geometric distribution with parameter $\frac{1}{32}$. This implies,  

\begin{equation}
E[Y^{(u,e)}_i] \leq 32.
\end{equation}

Substituting the value of $E[Y^{(u,e)}_i]$ in Equation \ref{sumY_greed} we get,

\begin{align}
\nonumber
\sum_{i=0}^m E[Y^{(u,e)}_i] & \leq \sum_{i=0}^m 32\\\nonumber
& = O(m)\\\nonumber
\text{Substituting the value of } & m \text{ from equation \ref{upp_m} we get }\\ 
\label{upp_y_greed}
\sum_{i=0}^m E[Y^{(u,e)}_i]&\leq O\left(\log_2 d_{\ths}\right).
\end{align} 
This concludes the proof.
\end{proof}

As the result of lemma \ref{y_greed} holds for all possible source destination pair $(u,e)$, such that $\dist_{C_n}(u,e) \leq \frac{d_{\ths}}{2}$, so we can use this result to get the following upper bound on $\max_{s,e}E[Y_{s,e}]$,

\begin{equation}
E[Y_{s,e}] \leq O\left(\log_2 d_{\ths}\right).
\end{equation}

Substituting the upper bound on $\max_{s,e}E[Y_{s,e}]$ in equation \ref{target_sofar} we get,
\begin{equation}
\max_{s,e \in C_n}E[|CommPath_{s,e}|] \leq O\left(\frac{n}{d_{\ths}} + \log_2 d_{\ths}\right)
\end{equation}

This concludes the proof of the upper bound on the total number of required entanglement swap operations for local best effort and modified greedy routing algorithms. For the ease of the reader's understanding, we restate the part of theorem \ref{ent_swap_begin}. 


\textbf{Power-law Part of Theorem \ref{ent_swap_begin}: }
\textit{In the continuous model, for a power-law virtual graph, constructed from the physical graph $C_n$, if $|D|=1$ and for all $i,j \in [0,n-1]$ if $D_{i,j}\in \{0,1\}$ then for sharing an entangled link between any source destination pair $s,e$, the expected number of required entangled swap operations by the Local best effort and the modified greedy routing algorithms ($\max_{s,e}E[|CommPath_{s,e}|]$), is upper bounded by}
\begin{align}
\label{greed_power}
&O\left(\frac{n}{d_{\ths}}+ \log_2 d_{\ths}\right).\\ \nonumber
\end{align}

\subsection{\textbf{Proof of the upper bound on the swap operations for NoN local best effort algorithm}}

Using a simple adaptation of the proof given in \cite{MNW04} one can show $E[Y_{s,e}] \leq O(\frac{\log_2 d_{\ths}}{\log_2 \log_2 d_{\ths}})$. For completeness, here we include the proof.

\begin{lemma}
\label{y_non_greed}
For the power-law virtual graphs, constructed from the physical graph $C_n$, if $|D|=1$ and for some $u,e \in [0,n-1]$ if $D_{u,e} = 1$ and if $\dist_{C_n}(u,e) \leq \frac{d_{\ths}}{2}$ then the algorithm \ref{path_non_local_best_effort} take $O(\frac{\log_2 d_{\ths}}{\log_2 \log_2 d_{\ths}})$ number of entanglement swap operations for distributing an entangled link between $u$ and $e$.
\end{lemma}

\begin{proof}
Let us consider a hierarchy of $m+1$ nonempty sets, $X^{(u,e)}_0 \supset X^{(u,e)}_1 \supset X^{(u,e)}_2 \supset \ldots \supset X^{(u,e)}_m$, where 
\begin{equation}
\label{m_non}
m := \frac{\log_2 d_{\ths}}{\log_2 \log_2 d_{\ths}},
\end{equation}
and for all $0\leq i \leq m$, $$X^{(u,e)}_i := \left\{u' : |u'-e| \le \frac{d_{\ths}}{(\log_2 d_{\ths})^i}\right\}.$$
Note that, all the nodes in $X^{(u,e)}_m$ are at most one distance away from the destination node $e$. This implies all of the nodes in the set $X^{(s,e)}_m$ are physical neighbours of the destination node $e$.

In a ring network, $C_n$, one can verify that for all $0 \leq i \leq m$,

\begin{equation}
\label{eqsize_power}
|X^{(u,e)}_i| \geq \frac{d_{\ths}}{(\log_2 d_{\ths})^i}.
\end{equation}

The algorithm starts to discover a path from a node in $u \in X^{(u,e)}_0$ to a node in the set $X^{(u,e)}_{m}$ through the nodes in the sets $X^{(u,e)}_{1}, \ldots , X^{(u,e)}_{m-1}$. The algorithm stops when it discovers the path up to the destination node in $X^{(u,e)}_m$. Let the algorithm spends $Y^{(u,e)}_i$ iterations in the set $X^{(u,e)}_i$. According to Algorithm \ref{path_non_local_best_effort} the length of the path is,
\begin{equation}
\label{eqY_non}
\sum_{i=0}^{m} Y^{(u,e)}_i.
\end{equation}

Here, first we show that for each $0 \leq i \leq m$, $E[Y^{(u,e)}_i]$ is upper bounded by a constant. This gives us the proof that, $E[Y_{s,e}] \leq O(m) \leq O\left(\frac{\log_2 d_{\ths}}{\log_2 \log_2 d_{\ths}}\right)$. The detailed proof is given below.

For any $0 \leq i \leq m-1$, suppose the path discovery algorithm discovers a path from $u \in X^{(u,e)}_0$ to a node $u'$ such that $u' \in X^{(u,e)}_i$ but $u' \not\in X^{(u,e)}_{i+1}$. Let $u' \rightarrow_1 X^{(u,e)}_{i+1}$ denotes the event that $u'$ has a virtual neighbour in the set $X^{(u,e)}_{i+1}$. In Algorithm \ref{path_non_local_best_effort} we are interested in the fact that whether $u'$ is connected to $X^{(u,e)}_{i+1}$ via a path of length two or not (neighbour of the neighbour). So, we use the notation $u' \rightarrow_2 X^{(u,e)}_{i+1}$ to denote the event $u'$ is connected to the set $X^{(u,e)}_{i+1}$ via a path of length two. Similarly, the complement of this event is denoted by $u' \not\rightarrow_2 X^{(u,e)}_{i+1}$.  If $v \in X^{(u,e)}_{i+1}$, then from equation \ref{pchoseapp} we have $\pchose(u',v) = \frac{1}{\beta_{u'}|u'-v|}$. Substituting the value of $\beta_{u'}$ from lemma \ref{upp_beta} we get $\pchose(u',v) \geq \frac{1}{4\log_2 d_{\ths}} \frac{1}{|u'-v|}$. As $v \in X^{(u,e)}_{i+1}$, this implies $|v-e| \leq \frac{d_{\ths}}{(\log_2 d_{\ths})^{i+1}}$, similarly as $u' \in X^{(s,e)}_i$, this implies $|u-e| \leq \frac{d_{\ths}}{(\log_2 d_{\ths})^{i}}$. By combining these two inequalities we get,

\begin{align}
\nonumber
|u'-e| + |v-e| & \leq \frac{2d_{\ths}}{(\log_2 d_{\ths})^{i}}.
\end{align}
From the triangle inequality we have, 
\begin{align}
\nonumber
|u'-e| + |v-e| & \geq |u'-v|.
\end{align}
This implies,
\begin{align}
\label{up_uv}
|u'-v| &\leq \frac{2d_{\ths}}{(\log_2 d_{\ths})^{i}}.
\end{align}
By substituting the upper bound of $|u'-v|$ from equation \ref{up_uv} and $\beta_{u'}$ from lemma \ref{upp_beta} in the expression of $\pchose(u',v)$ we get,

\begin{equation}
\label{low_non_p}
\pchose(u',v) \geq  \frac{1}{4\log_2 d_{\ths}} \frac{(\log_2 d_{\ths})^i}{d_{\ths}}.
\end{equation}

In algorithm \ref{path_non_local_best_effort} we are interested in finding whether the node $u'$ is connected to any node in $X^{(u,e)}_{i+1}$ via path of length at most $2$. In the set $X^{(u,e)}_{i+1}$, there are $|X^{(u,e)}_{i+1}|$ such $v$'s and each of them is a long distance neighbour of $u'$ with probability $\pchose(u',v)$. So, the probability that $u$ is connected with one of them is,
\begin{align*}
&\sum_{v \in X^{(u,e)}_{i+1}} \pchose(u',v) \\
& \text{Substituting the value of } \pchose(u',v) \text{ from equation \ref{low_non_p}}\\
& \geq \sum_{v \in X^{(u,e)}_{i+1}}\frac{1}{4\log_2 d_{\ths}} \frac{(\log_2 d_{\ths})^i}{d_{\ths}}\\
&= |X^{(u,e)}_{i+1}|\frac{1}{4\log_2 d_{\ths}} \frac{(\log_2 d_{\ths})^i}{d_{\ths}}\\
&\text{Substituting $|X^{(u,e)}_{i+1}|$ from equation \ref{eqsize_power}}\\
&\geq \frac{d_{\ths}}{(\log_2 d_{\ths})^{i+1}}\frac{1}{4\log_2 d_{\ths}} \frac{(\log_2 d_{\ths})^i}{d_{\ths}}~~~~\\
& = \frac{1}{4(\log_2 d_{\ths})^2}.
\end{align*}

$u'$ has $k = \log_2 d_{\ths}$ such virtual neighbours and each of them is identical and independently distributed according to $\pchose$. This implies,  

\begin{align*}
\Pr[u' \not\rightarrow_1 X^{(u,e)}_{i+1}] &\leq (1 - |X^{(u,e)}_{i+1}|\pchose(u',v))^k\\
\end{align*}
Each of the neighbours of $u'$ has again $k$ such virtual neighbours and each of those virtual neighbours are identically and independently distributed according to $\pchose$. This implies,

\begin{align*}
\Pr[u' \not\rightarrow_2 X^{(u,e)}_{i+1}] &\leq (1 - |X^{(u,e)}_{i+1}|\pchose(u',v))^{k^2}.
\end{align*}
As $(1-x)^r \leq  \exp(-xr)$, this implies,
\begin{align*}
\Pr[u' \not\rightarrow_2 X^{(u,e)}_{i+1}] & \leq \exp(-k^2|X^{(u,e)}_{i+1}|\pchose(u,v)).
\end{align*}
Substituting the value of $\pchose(u',v)$ and $X^{(u,e)}_{i+1}$ rom equation \ref{low_non_p} and \ref{eqsize_power} we get,
\begin{align*}
\Pr[u' \not\rightarrow_2 X^{(u,e)}_{i+1}]& \leq \exp\left(-\frac{k^2}{4(\log_2 d_{\ths})^2}\right).
\end{align*}
Substituting $k= \log_2 d_{\ths}$ we get,
\begin{align*}
\Pr[u' \not\rightarrow_2 X^{(u,e)}_{i+1}]& \leq \exp\left(-\frac{1}{4}\right)\\
1- \Pr[u' \not\rightarrow_2 X^{(u,e)}_{i+1}] &\geq 1- \exp\left(-\frac{1}{4}\right)
\end{align*}

As for any $0\leq x \leq 1$,$1-e^{-x} \geq \frac{x}{2}$. This implies,
\begin{align*}
\Pr[u' \rightarrow_2 X^{(u,e)}_{i+1}] & \geq \frac{1}{8}.
\end{align*}

If the algorithms discover the path to a node $u' \in X^{(u,e)}_i$, then with probability at least $\frac{1}{8}$, the algorithms discover another node $v \in X^{(u,e)}_{i+1}$. This implies, each step of the algorithms manages to find a neighbour node $v \in X^{(s,e)}_{i+1}$ with probability at least $\frac{1}{8}$. The phenomenon of this finding a node $v\in X^{(s,e)}_{i+1}$ can be modelled as a sequence of independent trials, where each trial has two outcomes (success and failure) and each trial succeeds with probability at least $\frac{1}{8}$.This implies the random variable $Y^{(u,e)}_i$ follows geometric distribution with parameter $\frac{1}{8}$. 

\begin{equation}
E[Y^{(u,e)}_i] \leq 8.
\end{equation}

Substituting the value of $E[Y^{(u,e)}_i]$ in Equation \ref{eqY_non} we get,

\begin{equation}
\sum_{i=0}^{m}E[Y^{(u,e)}_i] \leq O\left(\frac{\log_2 d_{\ths}}{\log_2 \log_2 d_{\ths}}\right).
\end{equation} 
This concludes the proof.
\end{proof}
As the result of lemma \ref{y_non_greed} holds for all possible source destination pair $(u,e)$, such that $\dist_{C_n}(u,e) \leq \frac{d_{\ths}}{2}$, so we can use this result to get the following upper bound on $\max_{s,e \in C_n}E[Y_{s,e}]$,

\begin{equation}
\max_{s,e\in{C_n}}E[Y_{s,e}] \leq O\left(\frac{\log_2 d_{\ths}}{\log_2 \log_2 d_{\ths}}\right).
\end{equation}
Substituting the value of $E[Y_{s,e}]$ in equation \ref{target_sofar} we get,
\begin{equation}
\max_{s,e\in {C_n}}E[|CommPath_{s,e}|] = O\left(\frac{n}{d_{\ths}} + \frac{\log_2 d_{\ths}}{\log_2 \log_2 d_{\ths}}\right)
\end{equation}

This concludes the proof of the upper bound on the total number of required entanglement swap operations for NoN local best effort algorithm. For the ease of the reader's understanding, we restate the part of theorem \ref{ent_swap_begin}.

\textbf{Power-law Part of Theorem \ref{ent_swap_begin}: }
\textit{In the continuous model, for a power-law virtual graph, constructed from the physical graph $C_n$, if $|D|=1$ and for all $i,j \in [0,n-1]$ if $D_{i,j}\in \{0,1\}$ then for sharing an entangled link between any source destination pair $s,e$, the expected number of required entangled swap operations by the NoN Local best effort algorithm ($\max_{s,e}E[|CommPath_{s,e}|]$), is upper bounded by}
\begin{align}
\label{non_power}
&O\left(\frac{n}{d_{\ths}}+ \frac{\log_2 d_{\ths}}{\log_2\log_2 d_{\ths}}\right).
\end{align}

\section{Upper Bound on the number of Entanglement Swap Operations for the Uniform Virtual Graphs}
\label{proof_unif}

In the continuous model, the maximum distance (distance in the physical graph) between two virtual neighbour nodes is $d_{\ths}$. Due to this upper bound, in the worst case any two nodes need to perform at least $\frac{\diam_{G_{\phs}}}{d_{\ths}}$ swap operations for sharing an entangled state. In lemma \ref{ndth_unif} we show that, in the worst case (when the distance between a source $s$ and a destination $e$ is $\diam_{G_{\phs}}$), for all of the proposed routing algorithms, after $O(\frac{n}{d_{\ths}})$ number of swap operations, the source $s$ can share an entangled state with a node $u$, such that $\dist_{G_{\phs}}(u,e) \leq \frac{d_{\ths}}{2}$. Later, in lemma \ref{y_greed_unif} we show that, both local best effort and modified greedy routing algorithms take $O\left(\frac{d_{\ths}}{(\log_2 d_{\ths})^2}\right)$ number of swap operations to share an entangled state from $u$ to the destination $e$. By combining the results of lemma \ref{ndth_unif} and \ref{y_greed_unif} we get the following upper bound,

\begin{equation}
\max_{s,e\in C_n}E[|CommPath_{s,e}|] \leq O\left(\frac{n}{d_{\ths}}+ \frac{d_{\ths}}{(\log_2 d_{\ths})^2}\right).
\end{equation}

In lemma \ref{y_non_unif} we show that using the NoN greedy routing algorithm, the node $u$ can share an entanglement link with the destination $e$ using $O\left(\frac{d_{\ths}}{(\log_2 d_{\ths})^2}\right)$ number of swap operations. Using the results of lemma \ref{ndth_unif} and \ref{y_non_unif} we prove the following upper bound on the total number of entanglement swap operations. 

\begin{equation}
\max_{s,e\in C_n}E[|CommPath_{s,e}|] \leq O\left(\frac{n}{d_{\ths}}+ \frac{d_{\ths}}{(\log_2 d_{\ths})^2}\right).
\end{equation}

In the next lemma we focus on proving $O(\frac{n}{d_{\ths}})$ bound.

\begin{lemma}
\label{ndth_unif}

For the uniform virtual graphs, constructed from the physical graph $C_n$, if $|D|=1$ and for all $i,j \in [0,n-1]$ if $D_{i,j}\in \{0,1\}$, then for any source destination pair $(s,e)$, all of the algorithms \ref{path_greed}, \ref{path_loc_best_eff} and \ref{path_non_local_best_effort} take $O(\frac{n}{d_{\ths}})$ number of entanglement swap operations for distributing an entangled link between $s$ and a node $u$ such that $\dist_{C_n}(u,e) \leq \frac{d_{\ths}}{2}$.
\end{lemma}

\textit{Outline of the Proof:} Like the proof of lemma \ref{ndth} here we divide the set of nodes of $C_n$ into $m'$ nonempty sets, $Z^{(s,e)}_0, \ldots , Z^{(s,e)}_{m'-1}$, where 
\begin{equation}
\label{eq_m'_unif}
m' := \left\lceil\frac{2|s-e|}{d_{\ths}}\right\rceil.
\end{equation} 

Here for all $0\leq i \leq m'-1$, $$Z^{(s,e)}_i := \left\{w : |w-e| \le |s-e| - \frac{id_{\ths}}{2}\right\}.$$ The greedy path discovery subroutines, start constructing the path from $s \in Z^{(s,e)}_0$. In this process if a node $w \in Z^{(s,e)}_i$, chooses $v$ as a next hop in the path, then either $v \in Z^{(s,e)}_i$ or $v \in Z^{(s,e)}_{i+1}$ ($0 \leq i \leq m'-1$). Note that all the nodes $u \in Z^{(s,e)}_{m'-1}$ are at most $\frac{d_{\ths}}{2}$ distance away from the destination $e$. If for all $0\leq i \leq m'-2$, $\Z^{(s,e)}_i$ denotes the number of nodes the path discovery algorithms visit to discover a node in the set $Z^{(s,e)}_{i+1}$ from a node in the set $Z^{(s,e)}_i$ then this lemma we would like to prove,  
\begin{equation}
\sum_{i=0}^{m'-2} E[\Z^{(s,e)}_i] \leq O\left(\frac{n}{d_{\ths}}\right).
\end{equation}

Here we show that each of $E[\Z^{(s,e)}_i]$ is upper bounded by a constant, which is independent of $n$ and $d_{\ths}$. As a result, we get $\sum_{i=0}^{m'-2} E[\Z^{(s,e)}_i] \leq O\left(m'\right)$.By substituting the value of $m'$ from equation \ref{eq_m_unif} we get $ \sum_{i=0}^{m'-2} E[\Z^{(s,e)}_i] \leq O\left(m'\right) = O\left(\left\lceil\frac{2|s-e|}{d_{\ths}}\right\rceil\right)$. As in $C_n$, for any two nodes $s,e$, $\dist_{C_n}(s,e) \leq \frac{n}{2}$. This implies, $\sum_{i=0}^{m'-2} E[\Z^{(s,e)}_i] \leq O\left(\left\lceil\frac{2|s-e|}{d_{\ths}}\right\rceil\right) \leq O\left(\frac{n}{d_{\ths}}\right)$.

\begin{proof} We consider for any $0\leq i \leq m'-1$ the path discovery algorithm discovers the path from $s$ to a node $w \in Z^{(s,e)}_i$. In the next step of the path discovery algorithm, $w$ choses another neighbour $v$ such that $\dist_{C_n}(v,e)\leq \dist_{C_n}(w,e)$. Now, let $\Pr[w \rightarrow Z^{(s,e)}_{i+1}]$ denotes the probability that $w$ has a neighbour $v$ in the set $Z^{(s,e)}_{i+1}$. From the definition of $\pchose$ (equation \ref{pchoseunifapp}) we have,

\begin{equation}
\pchose(w,v) = 
\begin{cases}
\frac{1}{N_{\leq d_{\ths}}(C_n)},~\dist_{C_n}(w,v) \leq d_{\ths}\\
 0~~~\text{Otherwise},
\end{cases}
\end{equation}

For $C_n$, $N_{\leq d_{\ths}}(C_n) \leq 2d_{\ths}$. Substituting the upper bound on $N_{\leq d_{\ths}}(C_n)$ in the equation above we get,

\begin{equation}
\label{pchose_unif}
\pchose(u,v) = 
\begin{cases}
\frac{1}{2d_{\ths}},~\dist_{C_n}(w,v) \leq d_{\ths}\\
 0~~~\text{Otherwise},
\end{cases}
\end{equation}
 
 Let $Z'^{(s,e)}_{i+1} \subseteq Z^{(s,e)}_{i+1}$ denotes the set of nodes $v$ such that $\forall v \in Z'^{(s,e)}_{i+1}$, $|w-v| \leq d_{\ths}$. Each of the node in $Z'^{(s,e)}_{i+1}$ choose $w$ as a virtual neighbour with probability $\pchose(w,v)$. This implies $w$ has a virtual neighbour $v\in Z'^{(s,e)}_{i+1}$ with probability $|Z'^{(s,e)}_{i+1}|\pchose(w,v)$. In section \ref{power_ring} we mention that in the virtual graph, $w$ has at least $\log_2 d_{\ths}$ such virtual links and each of them belongs to the set $Z^{(s,e)}_{i+1}$ with probability $|Z'^{(s,e)}_{i+1}|\pchose(w,v)$. This implies,


\begin{align}
\nonumber
\Pr[w \rightarrow Z^{(s,e)}_{i+1}] \geq 1 &- (1-|Z'^{(s,e)}_{i+1}|\pchose(w,v))^{\log_2 d_{\ths}}\\ \nonumber
\text{Substituting the value of }& \pchose \text{ from equation \ref{pchose_unif} we get,}\\ \label{low_bound_unif_sf}
\Pr[w \rightarrow Z^{(s,e)}_{i+1}] &\geq 1- \left(1-|Z'^{(s,e)}_{i+1}| \frac{1}{2d_{\ths}}\right)^{\log_2 d_{\ths}}.
\end{align}

Now we focus on giving a lower bound on $|Z'^{(s,e)}_{i+1}|$.

For the lower bound, we consider the nodes in $Z^{(s,e)}_{i}$ which are the furthest from the set $Z'^{(s,e)}_{i+1}$. According to the definition (see equation \ref{eq_zpower}), all of the nodes in $Z'^{(s,e)}_i$ are at most $|s-e| - \frac{id_{\ths}}{2}$ distance away from the destination node $e$ and all of the nodes in $Z'^{(s,e)}_{i+1}$ are at most $|s-e| - \frac{(i+1)d_{\ths}}{2}$ away from the destination $e$. The maximum distance between a node $\omega \in Z^{(s,e)}_{i}$ and the set $Z'^{(s,e)_{i+1}}$ is computed in the following equation,

\begin{equation}
\max_{\omega \in Z^{(s,e)}_i} \min_{v' \in Z^{(s,e)}_{i+1}} |\omega - v'|.
\end{equation}
From the triangle inequality we have, for any three nodes $\omega, v', e$, $|\omega - v'| \geq |\omega - e| - |v' - e|$. This implies, 
\begin{align}
\nonumber
\max_{\omega \in Z^{(s,e)}_i} \min_{v' \in Z^{(s,e)}_{i+1}} |\omega - v'| &\geq \max_{\omega \in Z^{(s,e)}_i} \min_{v' \in Z^{(s,e)}_{i+1}} |\omega - e| - |v' - e|\\ \label{temp_z}
& = \max_{\omega \in Z^{(s,e)}_i} |\omega - e| - \min_{v' \in Z^{(s,e)}_{i+1}}   |v' - e|
\end{align}
As $ \omega \in Z^{(s,e)}_i$ this implies, $|\omega - e| \geq |s-e| - \frac{id_{\ths}}{2}$. Similarly, as $v' \in Z^{(s,e)}_{i+1}$, this implies $|v'-e| \leq |s-e| - \frac{(i+1)d_{\ths}}{2}$. Substituting the values of $|\omega - e|$ and $|v' - e|$ in equation \ref{temp_z} we get, 
\begin{align*}
\max_{\omega \in Z^{(s,e)}_i} \min_{v' \in Z^{(s,e)}_{i+1}} |\omega - v'| & \geq |s-e| - \frac{id_{\ths}}{2} - |s-e| + \frac{(i+1)d_{\ths}}{2}\\
&= \frac{d_{\ths}}{2}.
\end{align*}

The above derivation implies that in the worst case all of the nodes $Z'^{(s,e)}_{i+1}$ are at least $\frac{d_{\ths}}{2}$ distance away from $w$. According to the definition of $Z'^{(s,e)}_{i+1}$, each of the nodes in this set are at most $d_{\ths}$ distance away from the node $w$. This implies, the maximum distance between any two nodes in $Z'^{(s,e)}_{i+1}$ is at least $\frac{d_{\ths}}{2}$. As, in $C_n$, there are at least $d$ number of nodes within $d$ distance from any node $w$. This implies, 

\begin{equation}
\label{set_size}
|Z'^{(s,e)}_{i+1}| \geq \frac{d_{\ths}}{2}. 
\end{equation}
Substituting the value of $|Z'^{(s,e)}_{i+1}|$ in equation \ref{low_bound_unif_sf} we get,
\begin{align}
\nonumber
\Pr[w \rightarrow Z^{(s,e)}_{i+1}] &\geq   1- \left(1-|Z'^{(s,e)}_{i+1}| \frac{1}{2d_{\ths}}\right)^{\log_2 d_{\ths}}
\end{align}
Substituting the value of $|Z'^{(s,e)}_{i+1}|$ we get,
\begin{align}
\Pr[w \rightarrow Z^{(s,e)}_{i+1}] &\geq 1- \left(1- \frac{d_{\ths}}{4d_{\ths}}\right)^{\log_2 d_{\ths}}\\\nonumber
& \geq  1- \left(\frac{1}{4}\right)^{\log_2 d_{\ths}}\\
\label{low_bound_unif}
& \geq \frac{3}{4}~~~~\text{As }d_{\ths} \geq 2.
\end{align}

Here, $\Z^{(s,e)}_{i+1}$ denotes the number of required swap operations for distributing an entangled link between a node $v \in Z^{(s,e)}_{i+1}$ from a node $w \in Z^{(s,e)}_{i}$. From the equation \ref{low_bound_unif} we have that each of such node $w \in Z^{(s,e)}_{i}$ is connected to a link $v \in  Z^{(s,e)}_{i+1}$ with probability at least $\frac{3}{4}$. This implies, each swap operation manages to create an entangled link with a node $v \in Z^{(s,e)}_{i+1}$ with probability at least $\frac{3}{4}$. The phenomenon of this entanglement can be modelled as a sequence of independent trials, where each trial succeeds with probability at least $\frac{3}{4}$. This implies, $\Z^{(s,e)}_{i+1}$ follows a geometric distribution with parameter $\frac{3}{4}$. So, we have, for all $0\leq i \leq m'-2$, $E[\Z^{(s,e)}_{i+1}]\leq \frac{4}{3} \leq 2 = O(1)$. This implies,

\begin{align}
\nonumber
\sum_{i=0}^{m'-2} E[\Z^{(s,e)}_i] &\leq \sum_{i=0}^{m'-2} 2\\
\nonumber
& \leq 2m'\\\nonumber
\text{Substituting the value of }& m' \text{ from equation \ref{eq_m'_unif} we get,}\\
\sum_{i=0}^{m'-2} E[\Z^{(s,e)}_i]& \leq 2\left\lceil\frac{2|s-e|}{d_{\ths}}\right\rceil\\
\nonumber
& \leq \frac{2n}{d_{\ths}}~~\text{As }|s-e| \leq \frac{n}{2}\\
\label{upp_stand}
& = O\left(\frac{n}{d_{\ths}}\right).
\end{align}
This concludes the proof.
\end{proof}
Substituting the upper bound on $\sum_{i=0}^{m'-2} E[\Z^{(s,e)}_i]$ to equation \ref{target} we get,
\begin{equation}
\label{target_sofar_unif}
\max_{s,e \in C_n}E[|CommPath_{s,e}|] \leq O\left(\frac{n}{d_{\ths}} \right) + \max_{s,e \in C_n}E[Y_{s,e}].
\end{equation}

Up to this part, the proofs are the same for all of the proposed algorithms. The analysis changes when we compute $E[Y_{s,e}]$.

\subsection{\textbf{Proof of the upper bound on the swap operations for both local best effort and modified greedy algorithms}} In equation \ref{target_sofar_unif}, the term $Y_{s,e}$ is related to the term $\max_{s,e \in C_n}E[|CommPath_{s,e}|]$. For a source destination pair $(s,e)$, the term $Y_{s,e}$ denotes how many more entanglement swap operations are required for creating an entangled link between $s$ and $e$, given that $s$ already has created an entangled link with a node $u$ such that $\dist_{C_n}(u,e) \leq \frac{d_{\ths}}{2}$. In lemma \ref{y_greed_unif} we show that the modified greedy and the local best effort algorithm can share an entangled state between any two nodes $u,e$ such that $\dist_{C_n}(u,e) \leq \frac{d_{\ths}}{2}$, using only $ O\left(\frac{d_{\ths}}{(\log_2 d_{\ths})^2}\right)$ swap operations. This gives us the required bound for $\max_{s,e}E[Y_{s,e}] \leq O\left(\frac{d_{\ths}}{(\log_2 d_{\ths})^2}\right)$. The proof technique is just a simple adaptation of the proof given in \cite{klein99}. For completeness, here we include the proof.

\begin{lemma}
\label{y_greed_unif}
For the uniform virtual graphs, constructed from the physical graph $C_n$, if $|D|=1$ and for some $u,e \in [0,n-1]$ if $D_{u,e} = 1$ and if $\dist_{C_n}(u,e) \leq \frac{d_{\ths}}{2}$ then the algorithms \ref{path_greed}, \ref{path_loc_best_eff} take $O\left(\frac{d_{\ths}}{(\log_2 d_{\ths})^2}\right)$ number of entanglement swap operations for distributing an entangled link between $u$ and $e$.

\end{lemma}
\begin{proof}
Let us consider a hierarchy of $m+1$ nonempty sets, $X^{(u,e)}_0 \supset X^{(u,e)}_1 \supset X^{(u,e)}_2 \supset \ldots \supset X^{(u,e)}_m$. Here for all $0\leq i \leq m$, $$X^{(u,e)}_i := \left\{u : |u-e| \le \frac{d_{\ths}}{2^i}\right\},$$ 
where 

\begin{equation}
\label{eq_m_unif}
m := O\left(\log_2 \left(\frac{d_{\ths}}{\log_2 d_{\ths}}\right)\right).
\end{equation}

Note that, all of the nodes in $X^{(u,e)}_m$ is at most $\log_2 d_{\ths}$ distance away from the destination node $e$. This implies all of the greedy routing algorithms can discover the destination node $e$ using a constant number of entanglement swap operations. 

In a ring network, $C_n$, one can verify that for all $0 \leq i \leq m$,

\begin{equation}
\label{eqsize1}
|X^{(u,e)}_i| \geq \frac{d_{\ths}}{2^i}.
\end{equation}

As the Algorithms \ref{path_greed} and \ref{path_loc_best_eff} are greedy in nature, this implies they start to discover a path from a node in $u\in X^{(u,e)}_0$ to a node in the set $X^{(u,e)}_{m}$ through the nodes in the sets $X^{(u,e)}_{1}, \ldots , X^{(u,e)}_{m-1}$. The algorithm stops when the packet reaches $X^{(u,e)}_m$. Let the algorithm spends $Y^{(u,e)}_i$ iterations in the set $X^{(u,e)}_i$. According to Algorithm \ref{path_greed} and \ref{path_loc_best_eff} the length of the path is,
\begin{equation}
\label{path_len}
\sum_{i=0}^{m} Y^{(u,e)}_i + Y'^{(u,e)},
\end{equation}
where $Y'^{(u,e)}$ denotes the length of the path for reaching $e$ from a node in $X^{(u,e)}_m$.

For any $0 \leq i \leq m-1$, suppose the path discovery algorithm has discovered a path from $u \in X^{(u,e)}_0$ to a node $u'$ such that $u' \in X^{(u,e)}_i$ but $u' \not\in X^{(u,e)}_{i+1}$. Let $u' \rightarrow X^{(u,e)}_{i+1}$ denotes the event that $u'$ has a virtual neighbour in the set $X^{(u,e)}_{i+1}$. In a similar way, by the notation $u' \not\rightarrow X^{(u,e)}_{i+1}$ we denote the event that $u'$ has no virtual neighbour in $X^{(u,e)}_{i+1}$. Let $v \in X^{(u,e)}_{i+1}$, then from equation \ref{pchoseunifapp} we get 
\begin{equation}
\label{eqpchose}
\pchose(u',v) \geq  \frac{1}{2d_{\ths}}. 
\end{equation}

There are $|X^{(u,e)}_{i+1}|$ such $v$'s and each of them is a virtual neighbour of $u'$ with probability $\pchose(u',v)$. So, the probability that $u'$ is connected with one of them is,
\begin{align*}
\sum_{v \in X^{(u,e)}_{i+1}} \pchose(u',v) &\geq |X^{(u,e)}_{i+1}| \frac{1}{4d_{\ths}}.
\end{align*}

In the above expression, by substituting the value of $|X^{(u,e)}_{i+1}|$ rom equation \ref{eqsize1} we get,
\begin{align*}
&\geq \frac{d_{\ths}}{2^{i+1}} \frac{1}{4d_{\ths}}\\
& = \frac{1}{2^{i+3}}.
\end{align*}

As $u'$ has $k=\log_2 d_{\ths}$ such virtual neighbours and each of them is identical and independently distributed according to $\pchose$. This implies,  

\begin{align*}
\Pr[u' \not\rightarrow X^{(u,e)}_{i+1}] &\leq (1 - |X^{(u,e)}_{i+1}|\pchose(u',v))^k.
\end{align*}
As for all real $x$ and for all $r>0$ we have, $(1-x)^r \leq \exp(-xr)$, this implies,
\begin{align*}
\Pr[u' \not\rightarrow X^{(u,e)}_{i+1}] & \leq \exp(-k|X^{(u,e)}_{i+1}|\pchose(u',v)).
\end{align*}

Substituting the values of $\pchose(u',v)$ and $|X^{(u,e)}_{i+1}|$ from equation \ref{eqpchose} and \ref{eqsize1} in the above inequality we get,
\begin{align*}
\Pr[u' \not\rightarrow X^{(u,e)}_{i+1}] & \leq \exp(-\frac{k}{2^{i+3}})\\
1- \Pr[u' \not\rightarrow X^{(u,e)}_{i+1}] &\geq 1- \exp(-\frac{k}{2^{i+3}})
\end{align*}
Substituting $k = \log_2 d_{\ths}$ in the above inequality we get,

\begin{align*}
1- \Pr[u' \not\rightarrow X^{(u,e)}_{i+1}] &\geq 1- \exp(-\frac{\log_2 d_{\ths}}{2^{i+3}})
\end{align*}
As for any $0\leq x \leq 1$, $1-e^{-x} \geq \frac{x}{2}$, this implies,

\begin{align*}
\Pr[u' \rightarrow X^{(u,e)}_{i+1}] & \geq \frac{\log_2 d_{\ths}}{2^{i+4}}.
\end{align*}

If the algorithms discover the path to a node $u' \in X^{(u,e)}_i$, then with probability at least $\frac{\log_2 d_{\ths}}{2^{i+4}}$, the algorithms discover another node $v \in X^{(u,e)}_{i+1}$. This implies, each step of the algorithm manages to find a neighbour node $v \in X^{(s,e)}_{i+1}$ with probability at least $\frac{\log_2 d_{\ths}}{2^{i+4}}$. The phenomenon of this finding a node $v\in X^{(s,e)}_{i+1}$ can be modelled as a sequence of independent trials, where each trial has two outcomes (success and failure) and each trial succeeds with probability at least $\frac{\log_2 d_{\ths}}{2^{i+4}}$.This implies the random variable $Y^{(u,e)}_i$ follows geometric distribution with parameter $\frac{\log_2 d_{\ths}}{2^{i+4}}$. This implies,

\begin{equation}
\label{eq_y_sofar}
E[Y^{(u,e)}_i] \leq \frac{2^{i+4}}{\log_2 d_{\ths}}.
\end{equation}

In equation \ref{path_len} we get the expected path length,

\begin{align}
\nonumber
E[\sum_{i=0}^{m} Y^{(u,e)}_i &+ Y'^{(u,e)}]\\\nonumber
= \sum_{i=0}^{m} E[Y^{(u,e)}_i] &+ E[Y'^{(u,e)}].
\end{align}
Substituting the value of $E[Y^{(u,e)}_i]$ from equation \ref{eq_y_sofar} in the above expression we get,
\begin{align}
\nonumber
E[\sum_{i=0}^{m} Y^{(u,e)}_i + Y'^{(u,e)}] &\leq  \sum_{i=0}^{m} \frac{2^{i+4}}{\log_2 d_{\ths}}  + E[Y'^{(u,e)}]\\
\label{upp_unif_sofar}
&\leq O\left(\frac{d_{\ths}}{(\log_2 d_{\ths})^2}\right)  + E[Y'^{(u,e)}].
\end{align}

As $m = O\left(\log_2 \left(\frac{d_{\ths}}{\log_2 d_{\ths}}\right)\right)$, this implies from any node $v \in X^{(u,e)}_m$ the destination node $e$ is at most $\log_2 d_{\ths}$ far. All of the proposed algorithms are greedy in nature. This implies, if the path discovery algorithms use the physical neighbours to reach to the destination $e$ then it will take at most $O(\log_2 d_{\ths})$ number of swap operations. Hence, we can conclude that $E[Y'^{(u,e)}] \leq \log_2 d_{\ths} < O\left( \frac{d_{\ths}}{(\log_2 d_{\ths})^2}\right)$ (for all integer $d_{\ths} \geq 2$). Substituting this upper bound in equation \ref{upp_unif_sofar} we get the expected path length is $$E[\sum_{i=0}^{m} Y^{(u,e)}_i + Y'^{(u,e)}] \leq O\left(\frac{d_{\ths}}{(\log_2 d_{\ths})^2}\right).$$
This concludes the proof. 
\end{proof}
The proof of lemma \ref{y_greed_unif} holds for any source destination pair $u,e$ such that $\dist_{C_n}(u,e) \leq \frac{d_{\ths}}{2}$. This implies, we can use this upper bound to prove the upper bound on $\max_{s,e \in C_n}E[Y_{s,e}]$.
Substituting the value of $E[Y_{s,e}]$ in equation \ref{target_sofar_unif} we get,
\begin{equation}
\max_{s,e \in C_n}E[|CommPath_{s,e}|] = O\left(\frac{n}{d_{\ths}} + \frac{d_{\ths}}{(\log_2 d_{\ths})^2}\right)
\end{equation}

This proves the upper bound on the number of required entanglement swap operations for both local best effort and modified greedy routing algorithms. For the ease of the reader's understanding, we restate the uniform part of theorem \ref{ent_swap_begin}

\textbf{Uniform Part of Theorem \ref{ent_swap_begin}: }
\textit{In the continuous model, for a uniform virtual graph, constructed from the physical graph $C_n$, if $|D|=1$ and for all $i,j \in [0,n-1]$ if $D_{i,j}\in \{0,1\}$ then for any source destination pair $s,e$, the expected number of required entangled swap operations for sharing an entangled link between $s,e$ is upper bounded by}
\begin{align}
\nonumber ~&\text{For the routing algorithms \ref{path_greed} and \ref{path_loc_best_eff}}\\ \label{unif_cont}
&O\left(\frac{n}{d_{\ths}}+ \frac{d_{\ths}}{(\log_2 d_{\ths})^2}\right).
\end{align}

Now, we give the upper bound for the NoN local best effort algorithms. 
\subsection{\textbf{Proof of the upper bound on the swap operations for NoN local best effort algorithm}}

\begin{lemma}
\label{y_non_unif}
For the uniform virtual graphs, constructed from the physical graph $C_n$, if $|D|=1$ and for some $u,e \in [0,n-1]$ if $D_{u,e} = 1$ and if for the source destination pair $u,e$, if $\dist_{C_n}(u,e) \leq \frac{d_{\ths}}{2}$ then the algorithms \ref{path_non_local_best_effort} take $O\left(\frac{d_{\ths}}{(\log_2 d_{\ths})^2}\right)$ number of entanglement swap operations for distributing an entangled link between $u$ and $e$.
\end{lemma}
\begin{proof}

Let us consider a hierarchy of $m+1$ nonempty sets, $X^{(u,e)}_0 \supset X^{(u,e)}_1 \supset X^{(u,e)}_2 \supset \ldots \supset X^{(u,e)}_m$ such that for all $0\leq i \leq m$, $$X^{(u,e)}_i := \left\{u : |u-e| \le \frac{d_{\ths}}{(\log_2 d_{\ths})^i}\right\},$$ where

\begin{equation}
\label{eq_m_non_unif}
m = O\left(\log_{(\log_2 d_{\ths})} \left(\frac{d_{\ths}}{\log_2\log_2 d_{\ths}}\right)\right).
\end{equation}
Note that, all of the nodes in $X^{(u,e)}_m$ is at most $\log_2 \log_2 d_{\ths}$ distance away from the destination node $e$. This implies, all of the greedy routing algorithms can discover the destination node $e$ using $\log_2 \log_2 d_{\ths}$ number of entanglement swap operations.

In a ring network, $C_n$, one can verify that for all $0 \leq i \leq m$,

\begin{equation}
\label{eqsize}
|X^{(u,e)}_i| \geq \frac{d_{\ths}}{(\log_2 d_{\ths})^i}.
\end{equation}

The Algorithm \ref{path_non_local_best_effort} starts discovering a path from a node in $u \in X^{(u,e)}_0$ to a node in the set $X^{(u,e)}_{m}$ through the nodes in the sets $X^{(u,e)}_{1}, \ldots , X^{(u,e)}_{m-1}$. The algorithm stops when the it discovers the destination node in $X^{(u,e)}_m$. Let the algorithm spends $Y^{(u,e)}_i$ iterations in the set $X^{(u,e)}_i$. According to Algorithm \ref{path_non_local_best_effort} the length of the path is,
\begin{equation}
\label{path_len_non}
\sum_{i=0}^{m} Y^{(u,e)}_i + Y'^{(u,e)},
\end{equation}
where $Y'^{(u,e)}$ denotes the length of the path for reaching $e$ from a node in $X^{(u,e)}_m$.

For any $0 \leq i \leq m-1$, suppose the path discovery algorithm has discovered a path from $u \in X^{(u,e)}_0$ to a node $u'$ such that $u' \in X^{(u,e)}_i$ but $u' \not\in X^{(u,e)}_{i+1}$. Let $u' \rightarrow_1 X^{(u,e)}_{i+1}$ denotes the event that $u'$ has a virtual neighbour in the set $X^{(u,e)}_{i+1}$. In Algorithm \ref{path_non_local_best_effort} we are interested in the fact that whether $u'$ is connected to $X^{(u,e)}_{i+1}$ via a path of length two or not (neighbour of the neighbour). So, we use the notation $u' \rightarrow_2 X^{(u,e)}_{i+1}$ to denote the event $u'$ is connected to the set $X^{(u,e)}_{i+1}$ via a path of length two. Similarly, the complement of this event is denoted by $u' \not\rightarrow_2 X^{(u,e)}_{i+1}$.

Let $v \in X^{(u,e)}_{i+1}$, then $\pchose(u',v) \geq  \frac{1}{2d_{\ths}}$ (see equation \ref{pchoseunifapp}). 

There are $|X^{(u,e)}_{i+1}|$ such $v$'s and each of them is a virtual neighbour of $u'$ with probability $\pchose(u',v)$. So, the probability that $u'$ is connected with at least one of them is,
\begin{align*}
\sum_{v \in X^{(u,e)}_{i+1}} \pchose(u',v) &\geq |X^{(u,e)}_{i+1}| \frac{1}{2d_{\ths}}
\end{align*}
Substituting the value of $|X^{(u,e)}_{i+1}|$ from equation \ref{eqsize} in the above expression we get,
\begin{align*}
\sum_{v \in X^{(u,e)}_{i+1}} \pchose(u',v) &\geq \frac{d_{\ths}}{(\log_2 d_{\ths})^{i+1}} \frac{1}{2d_{\ths}}\\
& = \frac{1}{2 (\log_2 d_{\ths})^{i+1}}.
\end{align*}

As $u'$ has $k=\log_2 d_{\ths}$ such virtual neighbours and each of them is identical and independently distributed according to $\pchose$. This implies,  

\begin{align*}
\Pr[u' \not\rightarrow X^{(u,e)}_{i+1}] &\leq (1 - |X^{(u,e)}_{i+1}|\pchose(u',v))^k.
\end{align*}

Each of the neighbours of $u'$ has again $k$ such virtual neighbours and each of those virtual neighbours are identically and independently distributed according to $\pchose$ . This implies,

\begin{align*}
\Pr[u' \not\rightarrow_2 X^{(u,e)}_{i+1}] &\leq (1 - |X^{(u,e)}_{i+1}|\pchose(u',v))^{k^2}
\end{align*}
As for all real $x$ and $r >0$ we have, $(1-x)^r \leq \exp(-xr)$. This implies 
\begin{align*}
\Pr[u' \not\rightarrow_2 X^{(u,e)}_{i+1}]  &\leq \exp(-k^2|X^{(u,e)}_{i+1}|\pchose(u',v))
\end{align*}
In the above expression, by substituting the value of $|X^{(u,e)}_{i+1}|$ and $\pchose(u',v))$ from equation \ref{eqsize} and equation \ref{pchose_unif} we get,
\begin{align*}
\Pr[u' \not\rightarrow_2 X^{(u,e)}_{i+1}]  &\leq \exp\left(-\frac{k^2}{2(\log_2 d_{\ths})^{i+1}}\right)
\end{align*}
By substituting $k = \log_2 d_{\ths}$ in the above expression we get

\begin{align*}
 \Pr[u' \not\rightarrow_2 X^{(u,e)}_{i+1}] & \geq \exp\left(-\frac{1}{2(\log_2 d_{\ths})^{i-1}}\right)\\
1- \Pr[u' \not\rightarrow_2 X^{(u,e)}_{i+1}] &\geq 1- \exp\left(-\frac{1}{2(\log_2 d_{\ths})^{i-1}}\right)\\
\end{align*}

As for any $0\leq x \leq 1$, $1-e^{-x} \geq \frac{x}{2}$, this implies,

\begin{align*}
\Pr[u' \rightarrow_2 X^{(u,e)}_{i+1}]  \geq \frac{1}{4(\log_2 d_{\ths})^{i-1}}&.
\end{align*}

If the algorithms discover the path to a node $u' \in X^{(u,e)}_i$, then with probability at least $\frac{1}{4(\log_2 d_{\ths})^{i-1}}$, the algorithms discover another node $v \in X^{(u,e)}_{i+1}$. This implies, each step of the algorithm manages to find a neighbour node $v \in X^{(s,e)}_{i+1}$ with probability at least $\frac{1}{4(\log_2 d_{\ths})^{i-1}}$. The phenomenon of this finding a node $v\in X^{(s,e)}_{i+1}$ can be modelled as a sequence of independent trials, where each trial has two outcomes (success and failure) and each trial succeeds with probability at least $\frac{1}{4(\log_2 d_{\ths})^{i-1}}$.This implies the random variable $Y^{(u,e)}_i$ follows geometric distribution with parameter $\frac{1}{4(\log_2 d_{\ths})^{i-1}}$. This implies,

\begin{equation}
\label{temp_non}
E[Y^{(u,e)}_i] \leq 4(\log_2 d_{\ths})^{i-1}.
\end{equation}

Now, from equation \ref{path_len_non} we get,

\begin{align*}
E[\sum_{i=0}^{m} Y^{(u,e)}_i + Y'^{(u,e)}] &= \sum_{i=0}^{m} E[Y^{(u,e)}_i] + E[Y'^{(u,e)}]\\
\text{Substituting the value of }& E[Y^{(u,e)}_i] \text{ from equation \ref{temp_non} we get}\\
\leq \sum_{i=0}^{m}4(\log_2 d_{\ths})^{i-1}& + E[Y'^{(u,e)}]\\
\leq O\left(\frac{d_{\ths}}{(\log_2 d_{\ths})^2}\right) & + E[Y'^{(u,e)}].
\end{align*}

As $m = O\left(\log_{(\log_2 d_{\ths})} \left(\frac{d_{\ths}}{\log_2 \log_2 d_{\ths}}\right)\right)$, this implies from any node $v \in X^{(u,e)}_m$ the destination node $e$ is at most $\log_2\log_2 d_{\ths}$ far. NoN local best effort algorithm is greedy in nature. This implies, if the path discovery algorithm discovers $v$ as the next hop from $u$, then $\dist_{C_n}(v,e) < \dist_{C_n}(u,e)$. Hence, we can conclude that $E[Y'^{(u,e)}] \leq \log_2 \log_2 d_{\ths} < \frac{d_{\ths}}{(\log_2 d_{\ths})^2}$. This implies, $$ E[\sum_{i=0}^{m} Y^{(u,e)}_i + Y'^{(u,e)}] \leq O\left(\frac{d_{\ths}}{(\log_2 d_{\ths})^2}\right).$$
This concludes the proof. 
\end{proof}

The proof of lemma \ref{y_non_unif} holds for any source destination pair $u,e$ such that $\dist_{C_n}(u,e) \leq \frac{d_{\ths}}{2}$. This implies, we can use this upper bound to prove the upper bound on $\max_{s,e \in C_n}E[Y_{s,e}]$.
Substituting the value of $\max_{s,e \in C_n}E[Y_{s,e}]$ in equation \ref{target_sofar_unif} we get,
\begin{equation}
\max_{s,e \in C_n}E[|CommPath_{s,e}|] = O\left(\frac{n}{d_{\ths}} + \frac{d_{\ths}}{(\log_2 d_{\ths})^2}\right)
\end{equation}

This proves the upper bound on the number of required entanglement swap operations for both local best effort and modified greedy routing algorithms. For the ease of the reader's understanding, we restate the uniform part of theorem \ref{ent_swap_begin}.

\textbf{Uniform Part of Theorem \ref{ent_swap_begin}: }
\textit{In the continuous model, for a uniform virtual graph, constructed from the physical graph $C_n$, if $|D|=1$ and for all $i,j \in [0,n-1]$ if $D_{i,j}\in \{0,1\}$ then for any source destination pair $s,e$, the expected number of required entangled swap operations for sharing an entangled link between $s,e$ is upper bounded by}
\begin{align}
\nonumber ~&\text{For the routing algorithm \ref{path_non_local_best_effort}}\\ \label{unif_cont}
&O\left(\frac{n}{d_{\ths}}+ \frac{d_{\ths}}{(\log_2 d_{\ths})^2}\right).
\end{align}


\end{document}